\documentclass[a4paper,reqno]{amsart}

\usepackage[utf8]{inputenc}
\usepackage[T1]{fontenc}

\usepackage{braket}
\usepackage{epigraph}
\usepackage{microtype}
\usepackage{hyphenat}
\usepackage{amsmath,amssymb,amsfonts,amsthm}
\usepackage{mathtools}
\usepackage{mathrsfs}
\usepackage{dsfont}
\usepackage{siunitx}
\usepackage{cancel}

\usepackage[dvipsnames]{xcolor}
\usepackage{graphicx}
\usepackage[all,cmtip]{xy} 

\usepackage{tikz}
\usepackage{tikz-cd}
\usepackage{pgfplots}
\usepackage{tikz-3dplot}
\usetikzlibrary{calc, shadings, decorations.markings, shapes, arrows}
\pgfplotsset{compat=1.18}
\usepackage{wrapfig}
\usepackage{subfiles}

\usepackage{booktabs}
\usepackage{tabularx}
\usepackage{enumitem}
\usepackage{longtable}
\usepackage{rotating}

\usepackage{xspace}
\usepackage[printonlyused,withpage]{acronym}
\usepackage{lipsum} 
\usepackage{listings}

\usepackage[square, numbers]{natbib}
\citeindextrue

\usepackage[colorlinks=true, linkcolor=blue, citecolor=blue, urlcolor=blue]{hyperref}

\definecolor{light-gray}{gray}{0.95}
\usepackage[color=light-gray,textwidth=2cm]{todonotes}

\def\XXint#1#2#3{{\setbox0=\hbox{$#1{#2#3}{\int}$} 
		\vcenter{\hbox{$#2#3$}}\kern-.5\wd0}}

\renewcommand{\theta}{\vartheta}
\renewcommand{\epsilon}{\varepsilon}

\makeatletter
\renewcommand{\fnum@figure}{\textsc{\figurename~\thefigure}} 
\makeatother

\newtheoremstyle{classicthm}
{12pt}{12pt}{\slshape}{}{\bfseries}{.}{.5em}{}

\theoremstyle{classicthm}
\newtheorem{theoremd}{Theorem}[section]
\newenvironment{theorem}{\begin{theoremd}}{ \end{theoremd}}

\newtheorem{theoremd*}{Theorem}
\newenvironment{theorem*}{\begin{theoremd*}}{ \end{theoremd*}}

\newtheorem{cord}[theoremd]{Corollary}
\newenvironment{cor}{\begin{cord}}{ \end{cord}}

\newtheorem{lemd}[theoremd]{Lemma}

\newtheorem{propd}[theoremd]{Proposition}
\newenvironment{prop}{\begin{propd}}{ \end{propd}}

\newtheorem{definitiond}[theoremd]{Definition}
\newenvironment{definition}{\begin{definitiond}}{ \end{definitiond}}

\DeclareMathOperator{\sdet}{sdet}

\newtheoremstyle{remarkstyle}
{12pt}{12pt}{\upshape}{}{\itshape}{.}{.5em}{} 

\theoremstyle{remarkstyle}

\newtheorem{remd}[theoremd]{Remark}
\newenvironment{rem}{\begin{remd}}{	\end{remd}}

\newtheorem{exampled}[theoremd]{Example}

\newtheorem{ossd}[theoremd]{Observation}

\newtheorem{noted}{Note}[theoremd]

\title[Milnor metric from field theory]{Milnor metric for Morse--Smale flows from field theory} 

\author{Giovanni Molinari}
\address{Perimeter Institute for Theoretical Physics\\
31 Caroline Street North, Waterloo, N2L 2Y5,  Ontario, Canada}
\address{University of Waterloo}
\email{***} 

\author{Michele Schiavina}
\address{Department of Mathematics, University of Pavia, Via Ferrata 5, 27100 Pavia, Italy}
\address{INFN Sezione di Pavia, via Bassi 6, 27100 Pavia, Italy}
\email{michele.schiavina@unipv.it}

\date{\today} 

\begin{document}

\begin{abstract}
We study the partition function of Abelian $BF$ theory with an axial gauge fixing condition given by a Morse--Smale vector field, and we show that it recovers the Milnor metric on the determinant line on the twisted cohomology of the manifold. To do this we use the Batalin--Vilkovisky formalism and, in particular, a BV pushforward to cohomology in two steps. In this context, the Milnor metric serves as the generalisation of the Ruelle dynamical zeta function evaluated at zero for systems containing both closed orbits and critical points. This, together with a classic result due to Schwarz on the partition function of $BF$ 
theory in the Lorenz gauge, provides a field-theoretic realisation of Fried's conjecture, which is true for Morse--Smale flows.

\end{abstract}

\maketitle

\tableofcontents

\section{Introduction}
\label{sec:Intro}

Gauge theories are Lagrangian field theories characterised by their invariance under an infinite dimensional local Lie group.
Quantisation of said models necessitates a careful treatment as the group action complicates standard approaches to all quantisation schemes currently available. 

Among gauge theories, a class of particular interest is given by topological quantum field
theories (TQFTs), which are independent of additional structure on the underlying base manifold, and depend solely on its global topological structure. Beyond their role in condensed matter physics for the description of topological insulators and the fractional quantum Hall effect \cite{ChoMoore}, TQFTs have shown to be a powerful tool for pure mathematics, as they allow us to investigate manifold invariants \cite{Witten88,Witten89} and certain dynamical systems defined on them.

One of the first prominent examples of this connection is represented by the work of Schwarz~\cite{Sch78, Sch79}, who showed that the partition function of a certain class of degenerate quadratic functionals, now known as $BF$ theory, computes the Ray--Singer analytic torsion~\cite{Ray71}. In Schwarz's original construction, the ill-defined Gaussian integral arising from the degeneracy of the functional was handled through a resolution of the kernel of the relevant differential operator by means of a chain complex. Such a procedure is often called ``gauge fixing'', and it allows to interpret the analytic torsion, defined as an alternating product of regularised determinants of Laplacians, as the result of the calculation of the partition function with a particular choice of gauge. (See \cite[Section 4.1]{HadfieldKandelSchiavina20Ruelle} for a summary relevant to our paper.) In modern terms, this fact can be stated precisely within the Batalin--Vilkovisky (BV) formalism~\cite{Batalin83, Batalin84, Sch93}. Gauge-fixing is the choice of a Lagrangian subspace within an appropriate symplectic graded space of fields~\cite{Sch93, Mnev17}. 

In finite dimensions, the partition function is a Gaussian Berezinian integral over the chosen Lagrangian subspace, and the central result underlying the BV formalism is that the result is invariant under continuous deformations of the Lagrangian gauge-fixing subspace, provided that the action function satisfies the so-called quantum master equation.\footnote{This is a Maurer--Cartan condition for the action as an element in an exact Gerstenhaber/BV algebra.} One then gets the statement of ``gauge--fixing invariance'' of the Partition function.\footnote{Or rather local constancy of classes of gauge-fixing conditions, see \cite{Schiavina25}. This also extends to physical observables, see e.g.\ \cite{Cattaneo23}.}

In the infinite-dimensional setting of field theory this statement is far from trivial and requires a careful adaptation to be even formulated, and a verification on a case-by-case basis. Nevertheless, the gauge-fixing independence ``expectation'' for partition functions of field theory, provides the guiding principle that originated the present work. 

Indeed, looking at $BF$ theory within the BV framework, by choosing an auxiliary Riemannian metric $g$ on the manifold, one can define a gauge-fixing condition, usually referred to as the Lorenz/metric gauge, which picks out a Lagrangian subspace given by the image of the Hodge co-differential inside (two copies of) the space of differential forms.\footnote{This is the space of fields, see Section \ref{sec: BF-theory-and-Schwarz}.} In this setting, the partition function of $BF$ theory is shown to recover exactly the Ray--Singer analytic torsion, as detailed in~\cite{HadfieldKandelSchiavina20Ruelle}.

An alternative gauge-fixing condition for the same theory arises when the manifold is equipped with a vector field $V$ giving rise to Morse--Smale flow. The ensuing gauge-fixing condition, hereafter referred to as the axial gauge, is defined via the contraction $\iota_V$, and the associated gauge-fixed partition function encodes the periodic orbit structure of the flow rather than the spectral geometry of the Laplacian. The relationship between these two alternative calculations is the subject of the celebrated Fried's conjecture~\cite{Fried87, Fried86}, which, in its original formulation, relates the value at zero of the Ruelle dynamical zeta function for geodesic flows over hyperbolic manifolds to the Ray--Singer analytic torsion, thereby identifying a dynamical object with a spectral one.

This field-theoretic perspective on Fried's conjecture was first precisely formulated in the context of Anosov flows in~\cite{HadfieldKandelSchiavina20Ruelle}, specifically for Anosov--Reeb vector fields on contact manifolds (with vanishing twisted de Rham cohomology). By leveraging the Atiyah--Bott--Guillemin trace formula~\cite{Guillemin77, Atiyah64, Atiyah68}, the authors demonstrated that the partition function of Abelian $BF$ theory in the (Anosov) axial gauge exactly recovers the value at zero of the Ruelle zeta function. This was later improved by \cite{Schiavina23} to show that one can recover via field-theoretic methods the entirety of the Ruelle zeta function (as a meromorphic function). 


In this paper we will apply similar techniques to the superficially simpler case of Morse--Smale (MS) dynamical systems. If, in contrast to Anosov flows, MS vector fields present a \emph{finite} number of hyperbolic fixed points and hyperbolic closed orbits\footnote{The manifold partitions into stable submanifolds, defined as the set of points whose trajectories converge toward a critical element, i.e.\ a fixed point or a closed orbit.} satisfying a global transversality condition~\cite{Palis82}, in this work we will explore the more involved case of non-acyclic complexes.

 The extension of the analysis of \cite{HadfieldKandelSchiavina20Ruelle} to Morse--Smale flows relies on the spectral theory of the Lie derivative developed in~\cite{DangRiviere17Anisotropic, DangRiviere17Topology, DangRiviere17Resonances}. When acting on the anisotropic Sobolev spaces $\mathcal{H}^k_m(M, \mathcal{E})$ (see Appendix \ref{app:Sobolev}), the operator $-\mathcal{L}_{V, \nabla}^{(k)}$ exhibits a discrete spectrum, whose eigenvalues are called Pollicott--Ruelle resonances. The presence of fixed points ensures that the zero (generalised) eigenvalue is part of the spectrum, leading to the finite-dimensional zero-resonant complex $(C^\bullet_{V, \nabla}(0), d^\nabla)$ whose cohomology is canonically isomorphic to the twisted de Rham cohomology $H^\bullet(M, \mathcal{E})$, see~\cite{DangRiviere17Topology}.

A further key result, established in~\cite{DangRiviere17Topology} and systematically employed in this paper, is a Hodge-type decomposition. For any anisotropic current $\psi \in \mathcal{H}^k_m(M, \mathcal{E})$, we have the chain homotopy equation
\begin{equation}
    \mathbb{I} = \pi_0 + d^\nabla \circ \mathfrak{h}_V + \mathfrak{h}_V \circ d^\nabla,
\end{equation}
where $\pi_0$ is the spectral projector onto the zero-eigenspace and $\mathfrak{h}_V$ is a homotopy operator defined via the contraction and the inverse of the Lie derivative, both w.r.t.\ $V$. This identity induces a direct sum decomposition of the field space into three summands which is strongly reminiscent of the classical Hodge decomposition and which provides the primary mechanism for the construction of the axial gauge Lagrangian, and the setup for homotopy transfer to the generalised $0$-eigenspaces (the image of $\pi_0$). In field-theoretic terms, this is a procedure called BV pushforward~\cite{Cattaneo14, Cattaneo18a}.

Indeed, the non-triviality of the twisted cohomology $H^\bullet(M, \mathcal{E})$ in the Morse--Smale setting marks a substantial departure from the acyclic Anosov case treated in~\cite{HadfieldKandelSchiavina20Ruelle}. To handle the zero-modes associated with the flow's fixed points, we process the partition function via the BV pushforward paradigm, which can be seen as a framework for the Wilsonian strategy of accounting for degrees of freedom across different scales~\cite{Costello11}. 

In this case, we first ``integrate out'' all modes except zero-modes, and then we identify the twisted de Rham cohomology within the cohomology on the resulting residual space, on which effective BV data is given. This sets us up for a second BV pushforward. Note that this approach, in the case of non-Abelian BF theory with Morse--Smale gauge fixing, leads to a field-theoretic realisation of the Fukaya--Morse category \cite{CLMY}, where the structural higher $A_\infty$ relations are recovered by the classical master equation of the effective BV action. Here we are interested in fully computing the partition function (in the Abelian case), and its comparison with a different gauge fixing condition. An older implementation of the axial gauge for Morse vector fields in Chern--Simons theory see \cite{FukayaCS}.

In the first stage of the BV pushforward, we account for the infinite-dimensional subspace $\text{Im}(d_\nabla)\oplus \text{Im}(\iota_V)$---the \emph{fluctuation sector}--- which is a direct summand in anisotropic forms due to the chain homotopy decomposition. The core result used in this step is the identification of the Ruelle zeta function evaluated at zero $\mathfrak{R}_{V, \rho}(0)$ as the flat-regularised superdeterminant of the restricted Lie derivative $\widetilde{\mathcal{L}}_{V, \nabla}^{(k)}$. Using the trace formulas of Guillemin~\cite{Guillemin77} and Dang--Rivi\`ere~\cite{DangRiviere17Topology}, it is shown that:
\begin{equation}
  \mathfrak{R}_{V, \rho}(0) = \prod_{k=0}^{n-1}
  \left[\det\nolimits^\flat \!\left(\widetilde{\mathcal{L}}^{(k)}_{V, \nabla}\right)\right]^{(-1)^{n+k}}.
\end{equation}
The analyticity of these regularised determinants is ensured by the restricted resolvent, which is holomorphic at zero because the restriction to the fluctuation sector removes the singularities associated with the zero resonances. This first integration step then yields an effective action and an associated volume element defined on the finite-dimensional complex $(C^\bullet_{V, \nabla}(0), d^\nabla)$.

In the second stage of the BV pushforward in the axial gauge, we isolate the purely cohomological content of the resonant complex $C^\bullet_{V, \nabla}(0)$. To perform the gauge fixing associated with the second BV pushforward, we need to introduce an auxiliary metric $h_{\mathrm{aux}}$. The freedom of choice we have allows us to choose it such that the canonical isomorphism between the resonant complex and the Thom--Smale complex $(C^\bullet_{\mathrm{TS}}(V, \nabla), d_{\mathrm{TS}})$ is promoted to an isometry. A central result of our analysis, established in Theorem \ref{thm:equality-of-torsion}, is that under this isometric isomorphism the second integration stage yields exactly the combinatorial torsion of the Thom--Smale complex, whose differential encodes the connecting flow lines between critical points~\cite{Hutchings02Reidemeister, Hutchings02, Hutchings99Circle,Dang21}. We are thus able to successfully recover the fixed-point contribution to the axial partition function. 

By combining this fixed-point sector with the closed-orbit contribution encoded by the Ruelle zeta function evaluated at zero (obtained from the first BV-pushforward), the partition function of $BF$ theory in the axial gauge effectively reassembles the Milnor metric on the determinant line of the cohomology~\cite{Shen21}.

The comparison between these two gauge-fixing results suggests that the (conjectural) gauge-fixing independence of the Abelian $BF$ partition function implies the equality between the Ray--Singer and Milnor metrics. This identity, which constitutes a more general version of the Fried's conjecture in the Morse--Smale case, as proven in~\cite{Shen21}, is thus reinterpreted as a manifestation of the requirement that distinct gauge-fixing choices for the same physical theory must yield identical topological invariants.

The paper is structured as follows. 

In Section \ref{sec:BV_intro}, we provide a review of the Batalin--Vilkovisky formalism, focusing on the handling of reducible gauge symmetries and the definition of the BV pushforward as a tool to compute effective actions.

In Section \ref{sec:chain_homotopy}, we introduce the spectral properties of the twisted Lie derivative along a Morse--Smale vector field acting on anisotropic Sobolev spaces, establishing the fundamental chain homotopy equation and characterising the generalised zero-eigenspace.

In Section \ref{sec:algebraic_torsion}, we establish an isometric isomorphism between the zero-resonant complex and the Thom--Smale complex, proving the identity of their numerical torsions and providing a topological interpretation of the resonant states at zero.

In Section \ref{sec: Ruelle-zeta-and-Torsion-as-determinants}, we provide the spectral realisation of the Ruelle zeta function at zero and the definition of the Ray--Singer analytic torsion, which are both expressed as an alternating products of regularised determinants of the Lie derivative and the Laplacian, respectively.

In Section \ref{sec: BF-theory-and-Schwarz}, we introduce the twisted Abelian $BF$ theory and establish its Batalin--Vilkovisky formulation, providing the necessary tools to handle reducible gauge symmetries within a unified graded framework.

In Section \ref{sec: metric-gauge}, we perform the gauge-fixing of the theory in the metric gauge, reviewing the result relating the partition function in the metric/Lorenz gauge with the Ray--Singer metric, thus recovering within the BV framework the results originally established by Schwarz in \cite{Sch78, Sch79} (see also \cite{HadfieldKandelSchiavina20Ruelle}).

In Section \ref{sec: axial-gauge}, we implement the axial gauge-fixing defined by the choice of a Morse--Smale vector field; by means of a double BV pushforward procedure, we demonstrate that the partition function recovers the Milnor metric, which generalises the evaluation of the Ruelle zeta function at zero by incorporating fixed-point contributions. Lastly, we introduce a paradigm to rigorously justify the triviality of the Jacobian determinant associated with the interior product, yielding the identity $\det'(\iota_V)=1$.

Finally, we collect several technical results in the Appendices. Appendix \ref{app:Morse--Smale} reviews the fundamental properties of Morse--Smale flows, while Appendix \ref{app:Sobolev} summarises the functional framework of anisotropic Sobolev spaces and the algebraic properties of Pollicott--Ruelle resonances. Finally, in Appendix \ref{app:torsion_isomorphism}, we briefly recall the basic notions of algebraic torsion for finite-dimensional complexes and prove its invariance under isometric isomorphisms.
\section*{Acknowledgements}
This paper is based on, and an extension of, G.M.'s Master Thesis at the University of Pavia. We would like to thank C.\ Dappiaggi for providing support and valuable comments during the process.


\section{Lagrangian Field Theory and the Batalin--Vilkovisky Formalism}
\label{sec:BV_intro}

This section provides a self-contained introduction to foundations of field theory with symmetries within the Batalin--Vilkovisky (BV) formalism.

The primary purpose of this framework is to enable the rigorous quantisation of gauge theories. By employing homological algebra in the form of a combination of the Koszul--Tate and Chevalley--Eilenberg complexes, we replace the degenerate critical locus of gauge-invariant action functionals with a complex. In finite dimensions, this ensures a well-defined and gauge-fixing independent\footnote{Over smooth families of gauge-fixing conditions.} path integral. While this overview is not intended to be exhaustive, it aims to establish the necessary background and a common language for the subsequent discussion. Extensive details on the results and constructions assumed in this exposition can be found in the established literature on the subject; see, e.g., \cite{And92, DF99,Del17,Blohmann_LFT,SchiavinaSchnitzer25} for a cohomological approach to local Lagrangian field theory, and \cite{Hen90, Sta98, Cattaneo14, Cattaneo23} for the BV formalism.

\subsection{Classical field theory, symmetries, and the quantisation problem}
\label{sec:problem-path-integral}

We begin by establishing the standard framework for classical Lagrangian field theories. We define the space of
classical fields $\mathcal{F}_M\doteq \Gamma(M, \mathcal{E})$ as the space of smooth sections of a
suitable vector bundle $\mathcal{E}$ over a closed manifold $M$. The dynamics of the system is governed by a local action
functional $S_M \colon \mathcal{F}_M \to \mathbb{R}$. The property of locality implies
that $S_M$ can be expressed as the integral of a Lagrangian density $L_M$ depending only
on the fields and a finite number of their derivatives. Formally, $L_M$ is a local
top-form defined on the $k$-th jet bundle $J^k \mathcal{E}$, and the action reads:
\begin{equation}
    S_M[\phi] = \int_M L_M(j^k \phi), \qquad \phi \in \mathcal{F}_M,
\end{equation}
where $j^k \phi$ denotes the $k$-jet prolongation of the section $\phi$. The physical
configurations of the theory correspond to the critical points of the action functional.
The space of such solutions is known as the \emph{Euler--Lagrange locus}:
\begin{equation}
    EL[S_M] = \{ \phi \in \mathcal{F}_M \mid \delta S_M[\phi] = 0 \},
\end{equation}
where $\delta S_M$ is the variational derivative of $S_M$. For a rigorous exposition of the geometric structures underlying local field theory, we refer the reader to the foundational work of Deligne and Freed \cite{DF99} (see also \cite{Del17, Blohmann_LFT}), and to Schiavina and Schnitzer \cite{SchiavinaSchnitzer25} for a recent perspective.

A particularly interesting scenario arises when the action functional enjoys a local
symmetry. Geometrically, this is described by a smooth, involutive distribution
$\mathcal{D}_M \subset T\mathcal{F}_M$ such that $X(S_M) = 0$ for all vector fields $X\in\Gamma(\mathcal{D}_M)$. The presence
of such gauge symmetry implies that the critical points of $S_M$ are highly
degenerate, and the true space of physically inequivalent configurations is the quotient
space $EL[S_M]/\mathcal{D}_M$, which is often singular.

Quantisation aims to transition from the classical commutative algebra of observables
to a noncommutative algebra of quantum operators. A standard formal approach to this
problem relies on the path integral formulation, where the partition function is heuristically though of as the integral:
\begin{equation}
\label{eq:formal_path_integral}
    Z = ``\int_{\mathcal{F}_M} \exp\!\left(\frac{i}{\hbar} S_M[\phi] \right) \mathcal{D}\phi\ ".
\end{equation}
Since a rigorous measure theory on the infinite-dimensional space $\mathcal{F}_M$ is
generally unavailable, the integral in Equation~\eqref{eq:formal_path_integral} is
treated as a formal power series in the Planck constant $\hbar$, obtained via a
stationary phase approximation around the critical points of $S_M$. However, due to the
degeneracy induced by $\mathcal{D}_M$, the Hessian of $S_M$ is not invertible along the
gauge orbits, causing the stationary phase approximation to fail. To achieve a meaningful
perturbative expansion one must break this degeneracy -- a procedure known as gauge fixing
-- which is rigorously handled by the Batalin--Vilkovisky formalism
\cite{Batalin83, Batalin84}.

\subsection{The cohomological approach and the BV complex}
\label{sec:BV-data}

To treat the singularities of the moduli space of classical solutions
$EL[S_M]/\mathcal{D}_M$, we seek a differential graded commutative algebra
$(\mathcal{A}^\bullet, Q)$ such that its cohomology matches the algebra of smooth
functions on the physical moduli space:
\begin{equation}
    H^{-i}(\mathcal{A}^\bullet, Q) = 0 \quad \text{for } i > 0,
    \qquad
    H^0(\mathcal{A}^\bullet, Q) \simeq C^\infty(EL[S_M]/\mathcal{D}_M).
\end{equation}
This resolution is achieved by first resolving the Euler--Lagrange equations through the
introduction of antifields of negative ghost number (Koszul--Tate complex), and then
encoding the invariants of the gauge group action via ghosts of positive ghost number
(Chevalley--Eilenberg complex); we refer to \cite{Sta98, FR13} for a detailed treatment
of these constructions.

The Batalin--Vilkovisky formalism geometrises this algebraic structure \cite{Sch93}.
The algebra $\mathcal{A}^\bullet$ is interpreted as the ring of smooth functions over a
$\mathbb{Z}$-graded vector space $\mathcal{F}_{BV}$, known as the space of BV fields,
whose degree-zero component coincides with the original space of classical fields
$\mathcal{F}_M$. Furthermore, $\mathcal{F}_{BV}$ is endowed with an odd symplectic
structure $\Omega_{BV}$ of ghost number $-1$ and an odd vector field of degree $+1$,
$Q \in \mathfrak{X}[1](\mathcal{F}_{BV})$, which is nilpotent, $[Q, Q] = 0$. The
classical action $S_M$ is extended to a degree-zero functional
$\mathcal{S}_M \colon \mathcal{F}_{BV} \to \mathbb{R}$ that acts as the Hamiltonian
generator of $Q$ with respect to the odd symplectic form:
\begin{equation}
    \iota_Q \Omega_{BV} = \delta \mathcal{S}_M.
\end{equation}
The nilpotency of $Q$ immediately implies that the extended master action $\mathcal{S}_M$
satisfies a fundamental non-linear differential equation, the \emph{Classical Master
Equation} (CME):
\begin{equation}
\label{eq:CME}
    \{\mathcal{S}_M,\, \mathcal{S}_M\}_{BV} = 0,
\end{equation}
where $\{\cdot, \cdot\}_{BV}$ is the odd Poisson bracket naturally induced by
$\Omega_{BV}$.


\subsection{Half-densities, the canonical BV Laplacian, and Lagrangian subspaces}
\label{sec:half_densities}
In finite dimensions, one can evaluate the path integral directly. To exploit the full potential of the BV integration and quantisation framework, it is convenient to phrase it in terms of the integration of half-densities on graded submanifolds. Addressing how this construction can be rigorously generalised to infinite dimensions in certain special cases will be the centre of our subsequent discussion. 

Let us recall a few basic notions.
\begin{definition}
\label{def:density}\label{def:graded_density}
Let $W$ be a finite-dimensional vector space and let $F(W)$ be the set of bases of $W$.
A density of weight $\alpha \in \mathbb{R}$ on $W$ is a map
$\mu \colon F(W) \to \mathbb{R}^+$ such that for all
$g \in \mathrm{GL}(\mathbb{R}^{\dim W})$,
\begin{equation}
    \mu(g \cdot \underline{w}) = |\det(g)|^{\alpha}\, \mu(\underline{w}),
    \qquad \forall\, \underline{w} = (w_1, \dots, w_{\dim W}) \in F(W).
\end{equation}
The space of such densities is denoted $\mathrm{Dens}^{\alpha}(W)$. Standard densities
correspond to $\alpha = 1$, while half-densities correspond to $\alpha = 1/2$.

Let $W^\bullet = \bigoplus_{k \in \mathbb{Z}} W^k$ be a graded vector space. The space
of densities of weight $\alpha \in \mathbb{R}$ is defined as
\begin{equation}
    \mathrm{Dens}^{\alpha}(W^\bullet)
    = \bigotimes_{k \in \mathbb{Z}}
      \bigl(\mathrm{Dens}^{\alpha}(W^k)\bigr)^{(-1)^k},
\end{equation}
where $\bigl(\mathrm{Dens}^{\alpha}(W^k)\bigr)^{-1}$ denotes its topological dual space.
\end{definition}

One of the central results that rests at the foudation of the BV formalism is that the space of half-densities on a
$(-1)$-shifted symplectic graded vector space carries a canonical differential operator
\cite{Khudaverdian2004, Severa06,Severa2026}.

\begin{theorem}
\label{thm:existence_BV_laplacian}
Let $(\mathcal{F}_{BV}, \Omega_{BV})$ be a finite dimensional $(-1)$-shifted symplectic graded vector
space. There exists a canonical, second-order differential operator
\begin{equation}
    \Delta \colon
    \mathrm{Dens}^{1/2}(\mathcal{F}_{BV}) \to \mathrm{Dens}^{1/2}(\mathcal{F}_{BV}),
\end{equation}
known as the BV Laplacian, which has degree $+1$ and squares to zero,
$\Delta^2 = 0$. Moreover, in any global Darboux coordinate system\footnote{Note, albeit non-canonical, they always exist for odd symplectic manifolds.} $(x^i, \xi_i)$ such
that $\Omega_{BV} = \sum_i \delta\xi_i \wedge \delta x^i$, the operator $\Delta$ takes
the local form:
\begin{equation}
    \Delta = \sum_i \frac{\partial^2}{\partial x^i\, \partial \xi_i}.
\end{equation}
\end{theorem}

\begin{rem}
\label{rem:Delta_mu}
A Berezinian measure $\mu \in \mathrm{Dens}(\mathcal{F}_{BV})$ is said to be
\emph{compatible} with $\Omega_{BV}$ if $\Delta(\mu^{1/2}) = 0$. Given such a
compatible measure, one defines a measure-dependent operator
$\Delta_\mu \colon C^\infty(\mathcal{F}_{BV}) \to C^\infty(\mathcal{F}_{BV})$ on
functions via the intertwining relation:
\begin{equation}
\label{eq:intertwining}
    \Delta_\mu(f) \cdot \mu^{1/2} \doteq \Delta(f \cdot \mu^{1/2}),
    \qquad \forall f \in C^\infty(\mathcal{F}_{BV}).
\end{equation}
This operator inherits $\Delta_\mu^2 = 0$ from $\Delta^2 = 0$, and one can verify
directly that it admits the equivalent description
\begin{equation}
\label{eq:BV_Laplacian_mu}
    \Delta_\mu f = -\frac{1}{2}\operatorname{div}_\mu(X_f),
\end{equation}
where $X_f$ is the Hamiltonian vector field of $f$ with respect to $\Omega_{BV}$.
Furthermore, the intertwining relation given by Equation \eqref{eq:intertwining} implies the following
compatibility relations with the algebra and Poisson structures of
$C^\infty(\mathcal{F}_{BV})$:
\begin{align}
\label{eq:BV_Leibniz}
    \Delta_\mu(fg)
    &= (\Delta_\mu f)\,g + (-1)^{|f|} f\,(\Delta_\mu g)
       + (-1)^{|f|}\{f,\, g\}_{BV}, \\[4pt]
\label{eq:BV_Poisson}
    \Delta_\mu\{f,\, g\}_{BV}
    &= \{\Delta_\mu f,\, g\}_{BV}
       + (-1)^{|f|+1}\{f,\, \Delta_\mu g\}_{BV},
\end{align}
making the tuple
$\bigl(C^\infty(\mathcal{F}_{BV}),\;\cdot\;,\;\{\cdot,\cdot\}_{BV},\;\Delta_\mu\bigr)$
into a BV algebra \cite{Cattaneo06}.
\end{rem}

We now introduce the notion of a finite-dimensional Lagrangian subspace, which will be the central
geometric object encoding gauge fixing.

\begin{definition}
\label{def:lagrangian_subspace}
Let $(\mathcal{F}_{BV}, \Omega_{BV})$ be a $(-1)$-shifted symplectic graded vector
space. A subspace $\mathcal{L} \subset \mathcal{F}_{BV}$ is called Lagrangian if it is
isotropic and admits an isotropic complement.\footnote{This is sometimes called a ``split Lagrangian'' subspace, and in infinite dimensions it is stronger than requiring it to be isotropic and coisotropic, a property which is then referred to as Lagrangianity. We will not make use of the weaker notion here.} That is,
\begin{equation}
    \Omega_{BV}\big|_{\mathcal{L}} = 0,
\end{equation}
and there exists another isotropic subspace $\mathcal{L}' \subset \mathcal{F}_{BV}$ such that $\mathcal{F}_{BV} = \mathcal{L} \oplus \mathcal{L}'$.
\end{definition}

\begin{rem}
\label{rem:infinite_dim_lagrangian}
We adopt this specific formulation because it generalises well to infinite dimensions. As we will discuss in Sections \ref{sec: metric-gauge} and \ref{sec: axial-gauge}, requiring an isotropic complement provides the appropriate notion of a Lagrangian submanifold needed to rigorously define gauge fixing in the infinite-dimensional BV formalism.
\end{rem}

Before the BV integral can be properly constructed, we require the following technical
proposition, which serves as the key link making BV integration intrinsically
well-defined \cite{Sch93}.

\begin{prop}
\label{prop:restriction_half_density}
Let $(\mathcal{F}_{BV}, \Omega_{BV})$ be a $(-1)$-shifted symplectic graded vector
space and let $\mathcal{L} \subset \mathcal{F}_{BV}$ be a Lagrangian subspace. Given a
half-density $\xi \in \mathrm{Dens}^{1/2}(\mathcal{F}_{BV})$, its restriction to
$\mathcal{L}$ yields a $1$-density on $\mathcal{L}$:
\begin{equation}
    \xi\big|_{\mathcal{L}} \in \mathrm{Dens}(\mathcal{L}).
\end{equation}
\end{prop}

Proposition~\ref{prop:restriction_half_density} tells us that a half-density, which
generally cannot be integrated over the entire odd-symplectic space, becomes a canonical
volume form when restricted to a Lagrangian subspace. This defines the BV integral.

\begin{definition}
\label{def:BV_integral}
Let $\mathcal{L} \subset \mathcal{F}_{BV}$ be a Lagrangian subspace. The BV
integral of a half-density $\xi \in \mathrm{Dens}^{1/2}(\mathcal{F}_{BV})$ over
$\mathcal{L}$ is defined as the composition:
\begin{equation}
    \int_{\mathcal{L}} \colon
    \mathrm{Dens}^{1/2}(\mathcal{F}_{BV})
    \xrightarrow{\;\cdot\,|_{\mathcal{L}}\;}
    \mathrm{Dens}(\mathcal{L})
    \xrightarrow{\;\int_{\mathcal{L}}\;}
    \mathbb{C},
    \qquad
    \xi \longmapsto \int_{\mathcal{L}} \xi\big|_{\mathcal{L}}.
\end{equation}
\end{definition}

This construction possesses the following fundamental properties; see \cite{Batalin84, Sch93} for a proof.

\begin{theorem}[Batalin--Vilkovisky--Schwarz]
\label{thm:BVS}
Let $(\mathcal{F}_{BV}, \Omega_{BV})$ be a $(-1)$-shifted symplectic graded vector
space equipped with the canonical BV Laplacian $\Delta$ of
Theorem~\ref{thm:existence_BV_laplacian}.
\begin{enumerate}
    \item[\emph{(I)}] For any half-density
    $\xi \in \mathrm{Dens}^{1/2}(\mathcal{F}_{BV})$ and any Lagrangian subspace
    $\mathcal{L} \subset \mathcal{F}_{BV}$, the integral of a $\Delta$-exact
    half-density vanishes:
    \begin{equation}
        \int_{\mathcal{L}} \Delta\xi = 0,
    \end{equation}
    assuming convergence of the integral.
    \item[\emph{(II)}] Let $\xi \in \mathrm{Dens}^{1/2}(\mathcal{F}_{BV})$ be
    $\Delta$-closed, \textit{i.e.}, $\Delta\xi = 0$. If $\mathcal{L}_0$ and
    $\mathcal{L}_1$ are two Lagrangian subspaces connected by a smooth family of
    Lagrangian subspaces, then:
    \begin{equation}
        \int_{\mathcal{L}_0} \xi = \int_{\mathcal{L}_1} \xi.
    \end{equation}
\end{enumerate}
\end{theorem}

\begin{rem}
Theorem~\ref{thm:BVS} implies that the BV integral of a $\Delta$-closed half-density
depends only on the Lagrangian homotopy class of the subspace $\mathcal{L}$. 
In the infinite-dimensional setting, constructing an explicit Lagrangian homotopy between distinct gauge fixings is generally difficult. Nevertheless, the underlying principle of gauge independence provides the physical motivation for our approach: by explicitly evaluating the partition function in two different gauges and comparing the results, we recover the expected equality between the analytic torsion and the value at zero of the Ruelle zeta function.
\end{rem}

\subsection{Regularised determinants and flat traces}
\label{sec:flat_traces}

In the study of quadratic field theories, the evaluation of the partition function in Equation~\eqref{eq:formal_path_integral} is formally linked to the determinant of the differential operator featured in the action functional. In infinite dimensions, such determinants are generally ill-defined {as integrals} and are instead {defined as} the regularised determinant of the associated differential operator. While the zeta-function regularisation is a prevalent choice in spectral geometry (see, e.g.,~\cite{Sch78, Sch79}), the specific connections with Morse--Smale dynamical systems make it more natural to employ the notion of flat trace regularisation, inspired by the seminal works of Atiyah and Bott~\cite{Atiyah64, Atiyah68}.

We shall not provide a rigorous definition of flat traces on manifolds here, referring
the reader to~\cite{Baladi18} for a detailed treatment. Within this general overview, we
will assume that such traces are well-defined for the differential operators under
consideration, ensuring to properly justify their existence later when specifying our
exact framework.

\begin{definition}
\label{def:flat_determinant}
Let $W$ be a suitable inner product space and $A \colon W \to W$ a differential operator
such that its flat trace $\operatorname{tr}^\flat(e^{-t(A+\lambda)})$ is well-defined
for $\lambda \in \mathbb{C}$. We define the {flat determinant} of $A + \lambda$ as:
\begin{equation}
    \log \det{}^\flat (A + \lambda)
    \doteq -\frac{d}{ds}\bigg|_{s=0}
       \left[ \frac{1}{\Gamma(s)}
              \int_0^\infty t^{s-1}
              \operatorname{tr}^\flat\!\left( e^{-t(A+\lambda)} - \Pi_\lambda \right) dt
       \right],
\end{equation}
where $\Pi_\lambda$ is the spectral projector onto the kernel of $A+\lambda$, provided
the integral converges.
\end{definition}

\begin{rem}
If the operator $A$ is such that the heat kernel $e^{-t(A+\lambda)}$ is trace-class,
the flat trace coincides with the standard functional trace,
$\operatorname{tr}(e^{-t(A+\lambda)}) = \operatorname{tr}^\flat(e^{-t(A+\lambda)})$.
\end{rem}

To handle field theories formulated on $\mathbb{Z}$-graded vector spaces, we extend the 
previous construction to graded operators. For a degree-preserving linear map 
$A$ on a graded vector space $W^\bullet$, we define its {flat superdeterminant} 
as the alternating product of the flat determinants on each homogeneous subspace. 
This provides a formal assignment for Gaussian-type integrals which, provided $A$ 
is non-degenerate on $W^\bullet$, recovers the inverse square root of the 
superdeterminant, consistent with the result obtained for non-degenerate 
Gaussian integration in the finite-dimensional case.

\begin{definition}
\label{def:partition_function_flat}
Let $W^\bullet$ be a $\mathbb{Z}$-graded vector space endowed with an inner product
$(\cdot, \cdot)$ and let $\mathcal{S} \colon W^\bullet \to \mathbb{R}$ be a quadratic
functional of the form $\mathcal{S}[x] = \frac{1}{2} \int_M (x, Ax)$ for some graded
differential operator $A \colon W^\bullet \to W^\bullet$. The formal Gaussian integral
associated with $\mathcal{S}$ is defined through the flat superdeterminant of $A$ as:
\begin{equation}
    \int_{W^\bullet} e^{\frac{i}{\hbar} \mathcal{S}}
    \doteq \bigl| {\sdet}^\flat (A) \bigr|^{-\frac{1}{2}},
\end{equation}
provided that $A$ is non-degenerate on $W^\bullet$.
\end{definition}

The regularisation procedure established here allows one to circumvent the analytical ambiguities associated with non-elliptic operators. However, for quadratic field theories where the operator $A$ possesses a non-trivial kernel, assigning a value to the partition function requires a broader framework. This approach involves partitioning the field space into the degrees of freedom spanning the kernel of $A$, hereafter referred to as residual fields, and the fluctuations. The systematic treatment of this decomposition, which combines flat-trace regularisation with the BV pushforward paradigm, will be the subject of Subsection \ref{sec:inf-dim-gaussian}.

\subsection{Gauge fixing and the Quantum Master Equation}

In the geometric interpretation of the BV formalism, integrating over the entire space $\mathcal{F}_{BV}$ is physically redundant because it overcounts equivalent gauge configurations. This source of redundancy, in infinite dimensions, leads to a divergence ``off the bat'' as the gauge group of which one would be computing the volume is infinite dimensional. Hence, one must restrict the integration domain. A {gauge fixing} corresponds to
the choice of a Lagrangian subspace\footnote{This should be additionally transverse to the group orbits, to count relevant configurations up to a finite number of times.} $\mathcal{L} \subset \mathcal{F}_{BV}$
(Definition~\ref{def:lagrangian_subspace}), and the partition function of the
gauge-fixed theory is defined by (in finite dimensions):
\begin{equation}
    Z(\mathcal{S}_M,\, \mathcal{L}) \doteq Z\!\left(\mathcal{S}_M\big|_{\mathcal{L}}\right)
    = \int_{\mathcal{L}} \exp\!\left(\frac{i}{\hbar} \mathcal{S}_M\right)
      \sqrt{\mu}\,\big|_{\mathcal{L}},
\end{equation}
where $\sqrt{\mu}\big|_{\mathcal{L}}$ is the volume form induced on $\mathcal{L}$ by a
compatible Berezinian $\mu$ on $\mathcal{F}_{BV}$ (cf.\ Remark~\ref{rem:Delta_mu}).

The gauge-fixing independence of $Z(\mathcal{S}_M, \mathcal{L})$ is a direct consequence
of Theorem~\ref{thm:BVS}, applied to the half-density
$\xi = \exp(i\mathcal{S}_M/\hbar)\,\mu^{1/2}$. Indeed, item~(II) of that theorem
guarantees invariance under smooth deformations of $\mathcal{L}$ whenever $\xi$ is
$\Delta$-closed. Using the intertwining relation in Equation \eqref{eq:intertwining} and the BV
algebra relations of Equation \eqref{eq:BV_Leibniz}, one computes directly that the condition
$\Delta\xi = 0$ is equivalent to requiring that $\mathcal{S}_M$ satisfies the
\emph{Quantum Master Equation} (QME):
\begin{equation}
\label{eq:QME}
    \frac{1}{2}\{\mathcal{S}_M,\, \mathcal{S}_M\}_{BV}
    + i\hbar\, \Delta_\mu \mathcal{S}_M = 0.
\end{equation}
Hence, whenever Equation~\eqref{eq:QME} holds, the partition function
$Z(\mathcal{S}_M, \mathcal{L})$ is independent of the choice of gauge fixing.

\begin{rem}
\label{rem:infinite_dim}
It is important to stress that Theorem~\ref{thm:BVS} is strictly valid only in finite
dimensions. In infinite dimensions, a number of complications arise. 
First of all, for \emph{quadratic} functionals on infinite dimensional spaces of fields, this integration is \emph{defined} as the flat superdeterminant of a differential operator
(cf.\ Definition~\ref{def:partition_function_flat}), after evaluation on the gauge-fixing subspace. However, in such an infinite-dimensional context the standard
theorems guaranteeing gauge-fixing independence do not apply directly, and they require adaptation. 
\end{rem}

\subsection{The BV pushforward and Wilsonian effective actions}
\label{sec:BV_pushforward}

In quantum field theory, a powerful technique inspired by the Wilsonian renormalisation
group involves splitting the fields into high-energy (ultraviolet) and low-energy
(infrared) modes. The {BV pushforward} is the rigorous mathematical formalisation (and generalisation) of this prescription \cite{CattaneoRossi05,losev,Costello2007rbv,Mnev2007,Costello11,Cattaneo18a}.

By performing a path integral over a subset of fields (e.g.\ ultraviolet modes), one obtains an effective action, a function of the remaining fields, that encapsulates the
quantum effects of what has been integrated-out.  Concretely, let us suppose that the field space under consideration is decomposed into two odd-symplectic
components. 

\begin{prop}
\label{prop:factorisation_density}
Suppose $\mathcal{F}_{BV} = \mathcal{F}_{BV}' \oplus \mathcal{F}_{BV}''$ with
$\Omega_{BV} = \Omega_{BV}' + \Omega_{BV}''$. Then the space of half-densities
factorises accordingly:
\begin{equation}
    \mathrm{Dens}^{1/2}(\mathcal{F}_{BV})
    \simeq
    \mathrm{Dens}^{1/2}(\mathcal{F}_{BV}')
    \otimes \mathrm{Dens}^{1/2}(\mathcal{F}_{BV}'').
\end{equation}
Moreover, the canonical BV Laplacian splits as $\Delta = \Delta' + \Delta''$, where
$\Delta'$ acts exclusively on the first factor and $\Delta''$ on the second.
\end{prop}

We are now in a position to introduce the BV pushforward, which acts by
integrating out only the degrees of freedom belonging to the subset
$\mathcal{F}_{BV}''$.

\begin{definition}
\label{def:BV_pushforward}
Let $\mathcal{F}_{BV} = \mathcal{F}_{BV}' \oplus \mathcal{F}_{BV}''$ be a decomposition
as per Proposition~\ref{prop:factorisation_density}, and let
$\mathcal{L}'' \subset \mathcal{F}_{BV}''$ be a Lagrangian subspace. The BV pushforward is the map:
\begin{equation}
\begin{split}
    \int_{\mathcal{L}''} &\colon
    \mathrm{Dens}^{1/2}(\mathcal{F}_{BV})
    \longrightarrow
    \mathrm{Dens}^{1/2}(\mathcal{F}_{BV}'), \\[4pt]
    \int_{\mathcal{L}''} \xi
    &\doteq \xi' \cdot \int_{\mathcal{L}''} \xi'',
    \qquad
    \forall\, \xi = \xi' \otimes \xi''
    \in \mathrm{Dens}^{1/2}(\mathcal{F}_{BV}')
       \otimes \mathrm{Dens}^{1/2}(\mathcal{F}_{BV}'').
\end{split}
\end{equation}
\end{definition}

This construction promotes to a chain map the map that integrates over a subset of the field space.

\begin{theorem}
\label{thm:BV_pushforward_properties}
Let $\int_{\mathcal{L}''}$ be the BV pushforward of
Definition~\ref{def:BV_pushforward}. The following properties hold:
\begin{enumerate}
    \item[\emph{(I)}]
    The pushforward intertwines the BV
    Laplacians:
    \begin{equation}
    \label{eq:chain_map}
        \Delta'\!\left( \int_{\mathcal{L}''}\xi \right)
        = \int_{\mathcal{L}''} \Delta\, \xi,
        \qquad
        \forall\, \xi \in \mathrm{Dens}^{1/2}(\mathcal{F}_{BV}).
    \end{equation}
    \item[\emph{(II)}]
    If
    $\xi \in \mathrm{Dens}^{1/2}(\mathcal{F}_{BV})$ satisfies $\Delta\xi = 0$, and
    $\mathcal{L}_0'', \mathcal{L}_1'' \subset \mathcal{F}_{BV}''$ are two Lagrangian
    subspaces connected by a smooth Lagrangian homotopy, then their pushforwards differ
    only by a $\Delta'$-exact term:
    \begin{equation}
    \label{eq:gauge_indep_pushforward}
        \int_{\mathcal{L}_1''} \xi - \int_{\mathcal{L}_0''} \xi = \Delta' \Psi,
    \end{equation}
    for some $\Psi \in \mathrm{Dens}^{1/2}(\mathcal{F}_{BV}')$.
\end{enumerate}
\end{theorem}

In the quantum field theory setting, one considers $\xi = \exp(i\mathcal{S}_{BV}/\hbar)\, \mu^{1/2}$ as the half-density, where
$\mathcal{S}_{BV} \colon \mathcal{F}_{BV} \to \mathbb{C}$ is the action functional
and $\hbar$ is a formal parameter. By Remark~\ref{rem:Delta_mu}, the requirement
that $\Delta\xi = 0$ is equivalent to $\mathcal{S}_{BV}$ satisfying Equation \eqref{eq:QME}.

\begin{definition}
\label{def:effective_action}
Suppose $\mathcal{F}_{BV} = \mathcal{F}_{BV}' \oplus \mathcal{F}_{BV}''$ and let
$\mathcal{L}'' \subset \mathcal{F}_{BV}''$ be a gauge-fixing Lagrangian subspace. Let
$\mathcal{S}_{BV}$ be a functional such that
$\xi = \exp(i \mathcal{S}_{BV}/\hbar)\,\mu^{1/2}$ is $\Delta$-closed. The
{effective BV action}
$\mathcal{S}_{BV}' \in C^\infty(\mathcal{F}_{BV}')$ is defined via the BV pushforward
as:
\begin{equation}
\label{eq:effective_action}
    \exp\!\left(\frac{i}{\hbar}\mathcal{S}_{BV}'\right) (\mu')^{1/2}
    \doteq
    \int_{\mathcal{L}''} \exp\!\left(\frac{i}{\hbar}\mathcal{S}_{BV}\right) \mu^{1/2},
\end{equation}
assuming the Berezinian factorises as $\mu = \mu' \otimes \mu''$.
\end{definition}

A consequence of the chain map property in \eqref{eq:chain_map} is that if the
total action $\mathcal{S}_{BV}$ satisfies the QME on the entire space
$\mathcal{F}_{BV}$, the effective action $\mathcal{S}_{BV}'$ automatically satisfies
the QME on the reduced space $\mathcal{F}_{BV}'$. Furthermore, by
Equation \eqref{eq:gauge_indep_pushforward}, changing the gauge-fixing condition, \textit{i.e.} the
Lagrangian subspace $\mathcal{L}''$, leaves the underlying physical observables unchanged. 

\subsection{Partition Function of Degenerate Quadratic Actions}
\label{sec:inf-dim-gaussian}

We now combine the theory of BV pushforward and flat-trace regularisation of differential operators, providing a rigorous definition for the path integral of quadratic functionals in the infinite-dimensional setting. Following the paradigm introduced in Subsection \ref{sec:BV_pushforward}, we address the case where the operator defining the action functional possesses a non-trivial kernel, identifying the resulting partition function as a half-density on the space of residual fields.

Let $A \colon \mathcal{F}_{BV} \to \mathcal{F}_{BV}$ be a linear degenerate operator on the BV field space such that the BV action is given by $\mathbb{S}_{BV}(v) = \frac{1}{2} \langle Av, v \rangle$. To properly define the partition function, we assume that $\mathcal{F}_{BV}$ admits a direct sum decomposition into symplectic subspaces:
\begin{equation}
\label{eq:field_decomposition_residual}
    \mathcal{F}_{BV} = \mathcal{F}_{BV}' \oplus \mathcal{F}_{BV}'',
\end{equation}
where $\mathcal{F}_{BV}'$ is a finite-dimensional subspace containing the kernel\footnote{The finite-dimensional requirement can be relaxed, but we will generally require the zero-modes, i.e. (generalised) eigenvectors of $A$ to be in $\mathcal{F}'_{BV}$. There are multiple choices of residual fields, which generally form a poset and are linked by BV pushforward \cite{Cattaneo18a}.} of $A$, while $\mathcal{F}_{BV}''$ is an infinite-dimensional complement, such that the restricted operator $A'' \doteq A|_{\mathcal{F}_{BV}''}$ is non-degenerate. By calling the fields in $\mathcal{F}_{BV}''$ {fluctuations}, the choice of a gauge-fixing condition is implemented by selecting a Lagrangian subspace $\mathcal{L} \subset \mathcal{F}_{BV}''$. 
In the infinite-dimensional setting, we define the {primed partition function} $\mathcal{Z}'$ as the result of a formal BV pushforward of the half-density $\xi = \exp(i\mathbb{S}_{BV}/\hbar)\mu^{1/2}$ over $\mathcal{L}$. The use of the primed notation emphasises that we are not interested in evaluating the full partition function over the entire space of fields; instead, we ``integrate out'' the infinite-dimensional fluctuations to obtain a residual partition function over the (finite-dimensional) space of residual fields $\mathcal{F}_{BV}'$.

For the quadratic functionals considered in this work, the assignment is given by:
\begin{equation}
\label{eq:residual_partition_function}
    \mathcal{Z}'(A) \coloneqq \exp\left(\frac{i}{\hbar}\mathbb{S}_{BV}'\right) \cdot |\sdet \nolimits^\flat(A'')|^{-1/2} \cdot (\mu')^{1/2},
\end{equation}
where $\mathbb{S}_{BV}' \doteq \mathbb{S}_{BV}|_{\mathcal{F}_{BV}'}$ is the action restricted to the residual fields, $|\sdet \nolimits^\flat(A'')|$ is the flat-trace regularised superdeterminant of the operator restricted to the fluctuations, and $(\mu')^{1/2}$ is a half-density on $\mathcal{F}_{BV}'$. The full numerical partition function of the theory would be recovered only after a subsequent integration over the finite-dimensional space of residual fields. This two-stage integration procedure provides the general mechanism that will be employed in the following chapters to investigate the Abelian $BF$ theory under different gauge-fixing conditions. 

We conclude by noting that, in greater generality like for non quadratic action functionals, the partial partition function of Eq. \eqref{eq:residual_partition_function}, a.k.a. the result of the BV pushforward on $\mathcal{F}''_{BV}$, receives contributions from the whole perturbative series of Feynman diagrams. A couple examples relevant to this paper are \cite{CLMY}, where the corrections are computed for non Abelian $BF$ theory in the Morse-Smale case, and \cite{Schiavina23} where the full partition function is computed for $BF$ theory with a nonlocal yet quadratic interaction term, with a gauge fixing condition given by an acyclic Anosov vector field. 
\section{The Chain Homotopy Equation and Pollicott--Ruelle Resonances}
\label{sec:chain_homotopy}

In this, we present some of the results obtained in \cite{DangRiviere17Topology}. The central goal is to show that the spectrum of the Lie derivative along a Morse--Smale vector field allows one to construct a finite-dimensional chain complex, linking the dynamical features of the flow to the underlying topology of the manifold. However, the Lie derivative acting on the standard space of smooth differential forms does not generally exhibit good spectral properties, and it is thus necessary to suitably enlarge the domain of the operator. By extending the action of the Lie derivative to a specific space with anisotropic Sobolev regularity, one can ensure that it acts as an operator with a well-behaved discrete spectrum, whose eigenvalues are known as Pollicott--Ruelle resonances. Before delving into this spectral analysis, we establish the geometric framework and the associated functional spaces that will be used throughout the remainder of the paper.

\begin{proof}[Conventions and geometric setup]
Let $(M, g)$ be a compact, connected, and oriented Riemannian manifold of dimension $n$ with empty boundary. We consider a $C^\infty$-linearisable Morse--Smale vector field $V \in \mathfrak{X}(M)$, as defined in Appendix \ref{app:Morse--Smale}. See also \cite[Section 3.3]{DangRiviere17Anisotropic} and \cite[Section 3.1]{DangRiviere17Resonances} for more details.

Because $M$ is compact, the flow map $(\phi^t_V)_{t \in \mathbb{R}}$ generated by $V$ is complete, meaning it is globally defined on $M$ for all $t \in \mathbb{R}$. Furthermore, $\mathcal{E} \to M$ denotes a smooth vector bundle endowed with a flat connection $\nabla$ constructed from a unitary representation $\rho \colon \pi_1(M) \to \operatorname{GL}(E)$ as in \cite{MathaiWu11Twisted}. We let $(\Omega^\bullet(M, \mathcal{E}), d^\nabla)$ denote the associated twisted de Rham complex. Given these data, we denote by $\mathcal{H}_m^k(M, \mathcal{E})$ the anisotropic Sobolev spaces of currents of degree $k$ (and order function $m$), whose constitutive details are deferred to Appendix \ref{app:Sobolev}, closely following \cite{DangRiviere17Anisotropic, DangRiviere17Topology}.
\end{proof}

\subsection{Chain Homotopy Equation}
The following theorem outlines a Hodge-type decomposition that holds for anisotropic Sobolev currents of any degree.
A detailed derivation can be found in \cite[Section 4.2]{DangRiviere17Topology}.

\begin{theorem}
\label{thm:chain_homotopy}
Let $\pi_{0} \colon \mathcal{H}_m^\bullet(M, \mathcal{E}) \to C^\bullet_{V, \nabla}(0)$ be the spectral projector onto the generalised zero-eigenspace of the Lie derivative $\mathcal{L}_{V, \nabla}$, as per Definition \ref{def:projector}, and let $\iota_V$ denote the interior product. Defining the operator 
\begin{equation}
    \mathfrak{h}_V^{(k)} \doteq \iota_V \circ \left(\mathcal{L}_{V, \nabla}^{(k)}\right)^{-1} \circ \left(\mathbb{I} - \pi_0^{(k)}\right),
\end{equation}
the following chain homotopy equation holds for all $\psi \in \mathcal{H}^k_m(M, \mathcal{E})$:
\begin{equation}
    \mathbb{I} = \pi_0 + d^{\nabla}\circ\mathfrak{h}_V + \mathfrak{h}_V\circ d^{\nabla}
\end{equation}
\end{theorem}

An immediate and elegant consequence of the previous theorem is that the anisotropic Sobolev spaces split into a direct sum.

\begin{cor}
\label{cor:direct_sum_decomposition}
The anisotropic Sobolev spaces $\mathcal{H}^\bullet_m(M, \mathcal{E})$ admit the direct decomposition:
\begin{equation}
\label{eq:chain-homotopy-decomp}
    \mathcal{H}^\bullet_m(M, \mathcal{E}) = C^\bullet_{V, \nabla}(0) \, \oplus \, \text{Im}\left(d^{\nabla}\big|_{\text{Ker}\pi_0}\right) \, \oplus \, \text{Im}\left(\iota_V\big|_{\text{Ker}\pi_0}\right).
\end{equation}
\end{cor}

\begin{proof}
Theorem \ref{thm:chain_homotopy} immediately implies that any $\psi \in \mathcal{H}_m^k(M, \mathcal{E})$ can be written as a sum of three elements:
\begin{equation*}
    \psi = \omega + d^{\nabla}(\alpha) + \iota_V(\beta),
\end{equation*}
where $\omega = \pi_0^{(k)}(\psi) \in C^k_{V, \nabla}(0)$, $\alpha = \mathfrak{h}_V^{(k)}(\psi) \in \text{Ker}(\pi_0^{(k-1)})$, and $\beta = \big(\mathcal{L}_{V, \nabla}^{(k+1)}\big)^{-1}(\mathbb{I}-\pi_0^{(k+1)}) (d^\nabla \psi) \in \text{Ker}(\pi_0^{(k+1)})$. Thus, the space is generated by the sum of these three subspaces. We must now prove that this sum is direct. Because $\pi_0^{(k)}$ is a projector, we have that $\mathcal{H}_m^k(M, \mathcal{E}) = \text{Im}(\pi_0^{(k)}) \oplus \text{Im}(\mathbb{I}- \pi_0^{(k)})$. Therefore, it suffices to show that the second and third subspaces -- both of which lie entirely within $\text{Im}(\mathbb{I}- \pi_0^{(k)})$ -- have a trivial intersection.

Let $\eta \in \text{Im}(d^{\nabla}) \cap \text{Im}(\iota_V) \cap \text{Im}(\mathbb{I}- \pi_0^{(k)})$. Because the operators $d^\nabla$ and $\iota_V$ are nilpotent, it follows immediately that:
\begin{equation*}
    d^{\nabla} \eta = 0 \quad \text{and} \quad \iota_V \eta = 0.
\end{equation*}
Evaluating the Lie derivative of $\eta$:
\begin{equation*}
    \mathcal{L}_{V, \nabla} (\eta) = (d^{\nabla} \circ \iota_V + \iota_V \circ d^{\nabla}) (\eta) = 0.
\end{equation*}
However, the Lie derivative is linear and strictly invertible on $\text{Im}(\mathbb{I}- \pi_0^{(k)})$ as per Proposition \ref{prop:operator_toolbox}, and hence $\eta \equiv 0$. Consequently, the intersection is trivial and the direct sum decomposition holds.
\end{proof}

\subsection{The Generalised Eigenspace at Zero and its Dimension}
\label{subsec:eigenspace_zero}

The generalised eigenspaces for the zero eigenvalue $C^k_{V, \nabla}(0)$ play an important role in the previous decomposition, hence it is crucial to examine their algebraic and topological properties. Specifically, we will observe its deep connection to the twisted de Rham cohomology, which is the foundational reason why dynamical zeta functions can capture topological invariants. We begin by observing that the differential $d^\nabla$ naturally restricts to the zero-eigenspace.

\begin{prop}
\label{prop:resonant_complex}
The sequence of vector spaces
\begin{equation}
\label{eq:resonant_complex_sequence}
    0 \xrightarrow{d^\nabla} C_{V, \nabla}^0(0) \xrightarrow{d^\nabla} C_{V, \nabla}^1(0) \xrightarrow{d^\nabla} \dots \xrightarrow{d^\nabla} C_{V, \nabla}^n(0) \xrightarrow{d^\nabla} 0
\end{equation}
forms a finite-dimensional cochain complex with respect to the restricted differential $d^\nabla$. We denote its $k$-th cohomology group by $H^k(C_{V, \nabla}^\bullet(0), d^\nabla)$.
\end{prop}

\begin{proof}
By item (II) of Proposition \ref{prop:operator_toolbox}, the exterior derivative $d^\nabla$ commutes with the spectral projector $\pi_0^{(k)}$. If $\psi_0 \in C^k_{V, \nabla}(0)$, then $\psi_0 = \pi_0^{(k)}(\psi_0)$. Applying the differential yields:
\begin{equation*}
    d^\nabla(\psi_0) = d^\nabla \big( \pi_0^{(k)}(\psi_0) \big) = \pi_0^{(k+1)} \big( d^\nabla(\psi_0) \big).
\end{equation*}
Thus, $d^\nabla(\psi_0)$ lies in the image of $\pi_0^{(k+1)}$, which is exactly $C^{k+1}_{V, \nabla}(0)$. Since $(d^\nabla)^2 = 0$ holds universally on $\mathcal{H}_m^\bullet(M, \mathcal{E})$, the property is trivially inherited by the subspaces $C^\bullet_{V, \nabla}(0)$.
\end{proof}


The fundamental topological insight brought forward in \cite{DangRiviere17Topology} is that the cohomology of this complex exactly recovers the smooth twisted de Rham cohomology of the underlying manifold. We refer the reader to \cite[Theorem 2.1 and Section 4.3]{DangRiviere17Topology} for the proof.

\begin{theorem}
\label{thm:spectral_realization_cohomology}
Let $(C^\bullet_{V, \nabla}(0), d^\nabla)$ be the differential complex defined in Equation \eqref{eq:resonant_complex_sequence}. Then, the spectral projector $\pi_0^{(k)} \colon \Omega^k(M, \mathcal{E}) \to C^k_{V, \nabla}(0)$ induces a canonical isomorphism:
\begin{equation}
    H^k(M, \mathcal{E}) \xrightarrow{\sim} H^k(C^\bullet_{V, \nabla}(0), d^\nabla), \qquad \forall \; 0 \le k \le n.
\end{equation}
\end{theorem}

\begin{rem}
\label{rem:cohomology_isomorphisms}
As an immediate consequence of Theorem \ref{thm:spectral_realization_cohomology} and the density of smooth forms in the anisotropic Sobolev spaces, the cohomologies computed on smooth, anisotropic and distributional forms
as well as the finite-dimensional resonant complex are all mutually isomorphic:
\begin{equation}
\label{eq:iso-cohomologies}
    H^k(C^\bullet_{V, \nabla}(0), d^\nabla) \simeq H^k(M, \mathcal{E}) \simeq H^k(\mathcal{H}_m^\bullet(M, \mathcal{E}), d^\nabla) \simeq H^k(\mathcal{D}'^\bullet(M, \mathcal{E}), d^\nabla).
\end{equation}
\end{rem}

This isomorphism implies that computing homological invariants (such as the Euler characteristic) via the resonant complex $(C^\bullet_{V, \nabla}(0), d^\nabla)$ correctly yields the standard topological invariants of the smooth manifold $M$. Furthermore, this complex is finite-dimensional; we substantiate this property below by explicitly computing its dimension and constructing a distinguished basis. For a detailed review of Morse--Smale flows, see \cite{Mol26} and Appendix \ref{app:Morse--Smale}.


Before stating the main theorem, we must introduce a purely topological factor associated with the closed orbits of the flow.

\begin{definition}
\label{def:twist-operator}
Let $\Lambda \subset M$ be a closed orbit of the Morse--Smale vector field $V$. The unstable manifold $W^u(\Lambda)$ defined in Equation \eqref{eq:invariant-manifolds-orbit} is said to be untwisted if it is orientable, and {twisted} otherwise. Consequently, the {twist factor} of the orbit is defined as:
\begin{equation}
\label{eq:twist-operator}
    \Delta(\Lambda) \coloneqq
    \begin{cases}
      +1 & \text{if } W^u(\Lambda) \text{ is untwisted,} \\
      -1 & \text{if } W^u(\Lambda) \text{ is twisted.}
    \end{cases}
\end{equation}
\end{definition}

The following result deals with the dimension of the space $C^k_{V, \nabla}(0)$ and shows under which hypothesis can be evaluated explicitly. See \cite[Section 5.1]{DangRiviere17Resonances} for details.

\begin{prop}
\label{prop:dimension-of-Ck}
If a Morse-Smale vector field $V$ is \emph{non-aligned}, meaning that for every closed orbit $\Lambda \subset M$ of the flow, the twist factor $\Delta(\Lambda)$ is not an eigenvalue of the holonomy map $\rho(\Lambda)$, then the dimension of the space of zero-resonant states is:
\begin{equation}
\label{eq:dim-resonant-space}
\begin{aligned}
    \dim\big(C^k_{V, \nabla}(0)\big) &= N \cdot c_k(V), \\
    c_k(V) &\coloneqq \# \big\{ a \in M \mid V(a) = 0 \; \wedge \; \dim(W^s(a)) = k \big\},
\end{aligned}
\end{equation}
where $N = \operatorname{rk}(\mathcal{E})$ and $W^{s}(a)$ is given by Equation \eqref{eq:unstable-point-app}.
\end{prop}

The assumption of $V$ being non-aligned has a clear geometric interpretation: no vector in the fibre $\mathcal{E}_p$ returns to itself (up to the sign $\Delta(\Lambda)$) after being parallel transported around a closed orbit. This will be assumed throughout the remainder of this work.

\begin{rem}
\label{rem:preferred-basis-PR}
The proof of Proposition \ref{prop:dimension-of-Ck} relies on the explicit construction of a preferred basis for $C^k_{V, \nabla}(0)$. Let $a \in M$ be a critical point such that $\dim(W^s(a)) = k$. In a sufficiently small, simply connected neighbourhood $U$ of $a$, one can introduce local coordinates $(x_1, \dots, x_k, y_1, \dots, y_{n-k})$ such that $W^s(a) \cap U$ is given by $y=0$ and $W^u(a) \cap U$ is given by $x=0$. On this neighbourhood, we can choose a local $\nabla$-parallel frame $\{e_1^a, \dots, e_N^a\}$ for the vector bundle $\mathcal{E}$. The component of the complete basis associated with the critical point $a$ consists of $N$ specific $k$-currents $\{U_{a, j}\}_{j=1}^{N}$, whose distributional support is precisely the closure of the unstable manifold $\overline{W^u(a)}$. In local coordinates, these currents take the explicit form:
\begin{equation}
    U_{a, j} = \delta_0(x_1, \dots, x_k) \, dx_1 \wedge \dots \wedge dx_k \otimes e_j^a(x,y),
\end{equation}
where $\delta_0(x_1, \dots, x_k)$ is the Dirac delta distribution supported at the origin of the stable coordinates (ensuring the support lies on the unstable manifold), and $dx_1 \wedge \dots \wedge dx_k$ provides the $k$-form contribution. The complete basis for $C_{V,\nabla}^k(0)$ is formed by the union of all such states across all critical points $a$ with $\dim(W^s(a))=k$:
\begin{equation}
    \mathscr{B}_k = \big\{ U_{a, j} \mid a \in M, \, V(a)=0, \, \dim(W^s(a))=k, \, j=1, \dots, N \big\}.
\end{equation}
This distinguished basis will be used to construct (non-canonical) isomorphisms with the Thom--Smale complex.
\end{rem}

\begin{rem}
\label{rem:acyclic_case}
If the Morse--Smale vector field $V$ is both {non-singular}, \textit{i.e.}, $c_k(V) = 0$ for all $k$, and {non-aligned}, then, due to Equation \eqref{eq:dim-resonant-space}, $z_0 = 0$ is not a Pollicott--Ruelle resonance. In this setting, the complex $(C^\bullet_{V, \nabla}(0), d^\nabla)$ is trivial, and the spectral projectors vanish identically: $\pi_0^{(k)} \equiv 0$.
\end{rem}

As shown in \cite[Proposition 5.1]{DangRiviere17Topology}, the spaces $C^k_{V, \nabla}(0)$ naturally form a chain complex with respect to the contraction $\iota_V$. Hence, the following result holds true.

\begin{prop}
\label{prop:dimension-homology}
Let $\iota_{V, C}^{(k)} \doteq \iota_V^{(k)}\big|_{C^k_{V, \nabla}(0)} \colon C^{k}_{V, \nabla}(0) \to C^{k-1}_{V, \nabla}(0)$ be the restriction of the contraction operator. Defining the $k$-th homology group of this complex as:
\begin{equation}
    H_k(C_{V,\nabla}^\bullet(0), \iota_{V,C}) \doteq \frac{\text{Ker}\left(\iota_{V, C}^{(k)} \colon C_{V,\nabla}^k(0) \to C_{V,\nabla}^{k-1}(0)\right)}{\text{Im}\left(\iota_{V, C}^{(k+1)} \colon C_{V,\nabla}^{k+1}(0) \to C_{V,\nabla}^k(0)\right)},
\end{equation}
its dimension is identically given by:
\begin{equation}
    \dim(H_k(C_{V,\nabla}^\bullet(0), \iota_{V,C})) = N \cdot c_k(V).
\end{equation}
\end{prop}

As an immediate consequence of Proposition \ref{prop:dimension-of-Ck} and Proposition \ref{prop:dimension-homology}, we obtain the following result.

\begin{prop}
\label{prop:zero_operator_consequence}
If the Morse--Smale vector field $V$ is {non-aligned},
then the contraction operator vanishes identically on the zero-resonant space:
\begin{equation}
    \iota_{V, C}^{(k)} \equiv 0, \quad \forall k.
\end{equation}
\end{prop}

\begin{proof}
Under the stated hypotheses, Proposition \ref{prop:dimension-of-Ck} ensures that $\dim(C_{V,\nabla}^k(0)) = N \cdot c_k(V)$. By definition, the dimension of the $k$-th homology group is:
\begin{equation}
\label{eq:homology_dim_def}
    \dim(H_k) = \dim(\text{Ker}(\iota_{V, C}^{(k)})) - \dim(\text{Im}(\iota_{V, C}^{(k+1)})).
\end{equation}
Substituting the result of Proposition \ref{prop:dimension-homology} into Equation \eqref{eq:homology_dim_def}, we obtain:
\begin{equation*}
    \dim(C_{V,\nabla}^k(0)) = \dim(\text{Ker}(\iota_{V, C}^{(k)})) - \dim(\text{Im}(\iota_{V, C}^{(k+1)})).
\end{equation*}
Since $\text{Ker}(\iota_{V, C}^{(k)})$ is a subspace of $C_{V,\nabla}^k(0)$, we have $\dim(\text{Ker}(\iota_{V, C}^{(k)})) \le \dim(C_{V,\nabla}^k(0))$, and naturally $\dim(\text{Im}(\iota_{V, C}^{(k+1)})) \ge 0$. The equality can therefore only hold if both of the following conditions are met simultaneously:
\begin{enumerate}
    \item $\dim(\text{Ker}(\iota_{V, C}^{(k)})) = \dim (C_{V,\nabla}^k(0))$,
    \item $\dim (\text{Im}(\iota_{V, C}^{(k+1)})) = 0$.
\end{enumerate}
The first condition proves that the kernel of $\iota_{V, C}^{(k)}$ is the entire space $C_{V,\nabla}^k(0)$. Thus, for any state $\psi \in C_{V,\nabla}^k(0)$, its image under $\iota_V^{(k)}$ is zero, meaning $\iota_{V, C}^{(k)}$ is identically zero.
\end{proof}
\section{Algebraic Torsion of the Resonant Complex}
\label{sec:algebraic_torsion}

In this section, we define and compute the algebraic torsion of the complex of zero-resonant states $(C^\bullet_{V, \nabla}(0), d^\nabla)$. For the fundamental notions of algebraic torsion for finite-dimensional cochain complexes and its invariance properties, the reader is referred to Appendix \ref{app:torsion_isomorphism} and the references therein. We first introduce the Thom--Smale complex generated by the critical points of a Morse--Smale vector field $V$ and subsequently  we prove that it is isometrically isomorphic to the resonant complex. This correspondence ensures that the numerical components of their respective torsions coincide, providing a dynamical interpretation of the spectral data.

\subsection{The Thom--Smale Complex and its Torsion}
\label{subsec:thom_smale_complex}

We begin by introducing a geometric finite-dimensional cochain complex associated with the critical points of a Morse--Smale vector field, which incorporates the twist by the flat vector bundle $\mathcal{E} \to M$. This construction yields a generalisation of the classical Morse complex. For a comprehensive exposition of Morse theory and its associated differential complexes, we refer the reader to \cite{Hutchings02, Hutchings02Reidemeister, Dang21}.

By recalling the definition of critical points given in Proposition \ref{prop:dimension-of-Ck}, for any $p \in \operatorname{Crit}_k(V)$ and $q \in \operatorname{Crit}_{k+1}(V)$, the transversality condition of the Morse--Smale flow ensures that the intersection of the unstable manifold of $p$ and the stable manifold of $q$ is a one-dimensional manifold. We define the space of connecting flow lines, modulo the action of $\mathbb{R}$ by time translation, as:
\begin{equation}
\label{eq:moduli-space}
    \mathcal{M}(p,q) \doteq \big(W^u(p) \cap W^s(q)\big) / \mathbb{R}.
\end{equation}
As shown in \cite[Section 2]{Hutchings02}, $\mathcal{M}(p,q)$ is a finite set of elements. By fixing an orientation for the stable manifolds of all critical points, each trajectory $\gamma \in \mathcal{M}(p,q)$ naturally acquires a sign, denoted $\operatorname{sign}(\gamma) \in \{-1, +1\}$.

\begin{definition}
\label{def:Thom--Smale-complex}
For each $k \in \{0, \dots, n\}$, the {$k$-th Thom--Smale vector space} $C^k_{TS}(V, \nabla)$ is defined as the finite-dimensional direct sum of the fibers of $\mathcal{E}$ over the critical points of index $k$:
\begin{equation}
    C^k_{TS}(V, \nabla) \doteq \bigoplus_{p \in \operatorname{Crit}_k(V)} \mathcal{E}_p.
\end{equation}
\end{definition}

The differential structure on these spaces is intrinsically dynamical and relies on the parallel transport along the flow lines connecting critical points of adjacent indices.

\begin{definition}
\label{def:Thom--Smale-differential}
The {Thom--Smale differential} $d_{TS}^k \colon C^k_{TS}(V, \nabla) \to C^{k+1}_{TS}(V, \nabla)$ is the linear map defined on any element $u_p \in \mathcal{E}_p$ by:
\begin{equation}
    d_{TS}^k(u_p) \doteq \sum_{q \in \operatorname{Crit}_{k+1}(V)} \sum_{\gamma \in \mathcal{M}(p,q)} \operatorname{sign}(\gamma) \mathbf{P}_\gamma(u_p),
\end{equation}
where $\mathbf{P}_\gamma \colon \mathcal{E}_p \to \mathcal{E}_q$ denotes the parallel transport map with respect to the flat connection $\nabla$ along the trajectory $\gamma$.
\end{definition}

The flatness of $\nabla$ guarantees that the parallel transport along a trajectory $\gamma$ is unambiguously defined and invariant under local deformations of the path. Eventually, the following result turns the sequence of spaces $C^\bullet_{TS}(V, \nabla)$ into a finite-dimensional cochain complex with cochain map $d_{TS}$. 

\begin{prop}
\label{prop:Thom--Smale-nilpotency}
The Thom--Smale differential is nilpotent, \textit{i.e.},
\begin{equation}
    d_{TS}^{k+1} \circ d_{TS}^k = 0, \quad \forall k \in \{0, \dots, n-1\}.
\end{equation}
\end{prop}

\begin{proof}
The interested reader can find a formal proof in \cite[Theorem 2.6]{Dang21} or \cite[Lemma 2.2]{Hutchings02}.    
\end{proof}

\begin{rem}
\label{rem:basis-dimension-link}
To better understand the structure of the spaces $C^k_{TS}(V, \nabla)$, it is useful to construct a distinguished basis for them. Let $x \in M$ be a generic base point and let $\{u_{x,1}, \dots, u_{x,N}\}$ be a reference basis for the fibre $\mathcal{E}_{x}$, where $N = \operatorname{rk}(\mathcal{E})$. Then, for each critical point $p \in \operatorname{Crit}_k(V)$, we generate a basis for the local fiber $\mathcal{E}_p$ by parallel transporting the reference basis along a chosen path from $x$ to $p$. The union of these local bases over all critical points in $\operatorname{Crit}_k(V)$ provides a distinguished basis for the whole space:
\begin{equation*}
    \mathcal{B}_k = \big\{ u_{p, j} \in \mathcal{E}_p \mid p \in \operatorname{Crit}_k(V), \, j \in \{1, \dots, N\} \big\}.
\end{equation*}
This explicit construction makes it clear that the dimension of $C^k_{TS}(V, \nabla)$ is strictly given by the product of the number of critical points and the rank of the vector bundle. By setting $c_k(V) \doteq \# \{\operatorname{Crit}_k(V)\}$, we have
\begin{equation*}
    \dim(C^k_{TS}(V, \nabla)) = c_k(V) \cdot N.
\end{equation*}
This establishes a direct connection with the Pollicott--Ruelle resonant states; indeed, under the assumption that the vector field is non-aligned, Proposition \ref{prop:dimension-of-Ck} ensures that $\dim(C^k_{V,\nabla}(0))$ shares this exact value. Consequently, for each degree $k$, there exists an isomorphism of complex vector spaces:
\begin{equation*}
    C^k_{TS}(V, \nabla) \cong C^k_{V,\nabla}(0).
\end{equation*}
However, it is essential to emphasise that an isomorphism between spaces of the same degree does not, by itself, constitute an isomorphism of the underlying cochain complexes. Proving that the differentials $d_{TS}$ and $d^\nabla$ rigorously intertwine under this identification will be the subject of Subsection \ref{subsec:isomorphism_torsion}.
\end{rem}

\begin{rem}
By considering $\mathcal{B}_k$ to be the base constructed in Remark \ref{rem:basis-dimension-link} for all $k \in \{0, \dots, n\}$, it is useful for future purposes to make explicit the action of the Thom--Smale differential on its elements:
\begin{equation}
\label{eq: Thom--Smale-differential}
    d_{TS}^k(u_{p, j}) = \sum_{q \in \text{Crit}_{k+1}(V)} \sum_{i=1}^N \left( \sum_{\gamma \in \mathcal{M}(p,q)} \text{sign}(\gamma) [\mathbf{P}_\gamma]_{ij} \right) u_{q, i},
\end{equation}
where $\mathcal{M}(p, q)$ is as per Equation \eqref{eq:moduli-space} and  $[\mathbf{P}_\gamma]_{ij}$ is the $(i,j)$-th entry of the matrix associated with the parallel transport operator $\mathbf{P}_\gamma: \mathcal{E}_p \to \mathcal{E}_q$, expressed with respect to the chosen bases $\mathcal{B}_k =\{u_{p, j}\}_{p \in \text{Crit}_k(V)}^{j=1, \dots, N}$ and $\mathcal{B}_{k+1} =\{u_{q, j}\}_{q \in \text{Crit}_{k+1}(V)}^{i=1, \dots, N}$.    
\end{rem}


\begin{rem}
\label{rem:thom_smale_torsion_definition}
Having established that $(C^\bullet_{TS}(V, \nabla), d_{TS})$ is a finite-dimensional cochain complex, the algebraic machinery recalled in Appendix \ref{app:torsion_isomorphism} applies directly. Given a Hermitian inner product on the fibers $\mathcal{E}_p$, the spaces $C^k_{TS}(V, \nabla)$ inherit a natural metric structure, allowing one to compute the numerical torsion of the Thom–Smale complex, which we denote by $\tau(C^\bullet_{TS})$.
\end{rem}

\subsection{Isomorphism of the Complexes and Equality of Torsions}
\label{subsec:isomorphism_torsion}

Our goal is now to prove that the numerical torsion of the Pollicott--Ruelle resonant complex exactly coincides with the torsion of the Thom--Smale complex. To achieve this, we will establish that $(C^\bullet_{V,\nabla}(0), d^\nabla)$ and $(C^\bullet_{TS}(V, \nabla), d_{TS})$ are isometrically isomorphic, as recalled in Definition \ref{def:cochain_isomorphism}. 


The first result we must present is the explicit action of the twisted de Rham differential on the distinguished basis of $C^\bullet_{V,\nabla}(0)$ presented in Remark \ref{rem:preferred-basis-PR}. The interested reader can find a proof of the following result in \cite[Theorem 2.6]{Dang21}.

\begin{prop}
\label{prop:differential_action_pref_base}
Let $\mathscr{B}_k = \big\{ U_{a, j} \mid a \in \operatorname{Crit}_k(V), \, j=1, \dots, N \big\}$ be the preferred basis of $C^k_{V, \nabla}(0)$ introduced in Remark \ref{rem:preferred-basis-PR}. The twisted de Rham differential $d^\nabla_k \colon C^k_{V, \nabla}(0) \to C^{k+1}_{V, \nabla}(0)$ acts on these basis elements via the formula:
\begin{equation}
\label{eq:differential_action_pref_base}
    d^\nabla_k (U_{a, j}) = \sum_{b \in \operatorname{Crit}_{k+1}(V)} \sum_{i=1}^N \left( \sum_{\gamma \in \mathcal{M}(a,b)} \operatorname{sign}(\gamma) \,[\mathbf{P}_\gamma]_{ij} \right) U_{b, i},
\end{equation}
where $\mathcal{M}(a,b)$ is the moduli space of flow lines connecting $a$ to $b$, $\operatorname{sign}(\gamma)$ is the orientation sign of the trajectory $\gamma$, and $[\mathbf{P}_\gamma]_{ij}$ denotes the $(i,j)$-th entry of the parallel transport matrix $\mathbf{P}_\gamma \colon \mathcal{E}_a \to \mathcal{E}_b$ along $\gamma$, expressed with respect to the chosen local $\nabla$-parallel frames.
\end{prop}

Notice that the matrix elements $[\mathbf{P}_\gamma]_{ij}$ appearing in Equation \eqref{eq:differential_action_pref_base} are identical to those defining the Thom--Smale differential in Definition \ref{def:Thom--Smale-differential}. This observation motivates the construction of a natural mapping between the two spaces. To make this mapping rigorous, we consider a \emph{coherent} choice of bases. Let $\mathcal{B}_k = \{u_{a, j}\}$ denote the canonical basis of $C^k_{TS}(V, \nabla)$, where $u_{a, j}$ is a basis vector of the fibre $\mathcal{E}_a$ over $a \in \operatorname{Crit}_k(V)$. We choose the local $\nabla$-parallel sections $c_{a, j}$ used to construct the resonant basis $\mathscr{B}_k = \{U_{a, j}\}$ such that their evaluation at the critical point precisely yields the corresponding fibre basis vector: $c_{a, j}(a) = u_{a, j} \in \mathcal{E}_a$. 

\begin{definition}
\label{def:complex_isomorphism}
We define $\Phi_k \colon C^k_{TS}(V, \nabla) \to C^k_{V, \nabla}(0)$ as the unique linear map satisfying the assignment on the coherent basis elements:
\begin{equation}
    \Phi_k(u_{a, i}) \coloneqq U_{a, i}, \qquad \forall a \in \operatorname{Crit}_k(V), \, \forall i \in \{1, \dots, N\}.
\end{equation}
\end{definition}

By construction, $\Phi = \{\Phi_k\}_k$ is an isomorphism of graded vector spaces, while the following theorem ensures it is also a chain map.

\begin{theorem}
\label{thm:iso-complex}
The cochain complexes $(C^\bullet_{TS}(V, \nabla), d_{TS})$ and $(C^\bullet_{V,\nabla}(0), d^\nabla)$ are isomorphic as per Definition \ref{def:cochain_isomorphism}. Specifically, the family of maps $\Phi_k$ satisfies:
\begin{equation}
    \Phi_{k+1} \circ d_{TS}^k = d^\nabla_k \circ \Phi_k, \qquad \forall k \in \{0, \dots, n-1\}.
\end{equation}
\end{theorem}

\begin{proof}
Let us verify the commutation relation on an arbitrary basis element $u_{a, j} \in \mathcal{B}_k$. Applying the Thom--Smale differential followed by the isomorphism $\Phi_{k+1}$, we obtain:
\begin{align*}
    \Phi_{k+1} \big( d_{TS}^k(u_{a, j}) \big) 
    &= \Phi_{k+1} \left( \sum_{b \in \operatorname{Crit}_{k+1}(V)} \sum_{i=1}^N \left( \sum_{\gamma \in \mathcal{M}(a,b)} \operatorname{sign}(\gamma) \, [\mathbf{P}_\gamma]_{ij} \right) u_{b, i} \right) \\
    &= \sum_{b \in \operatorname{Crit}_{k+1}(V)} \sum_{i=1}^N \left( \sum_{\gamma \in \mathcal{M}(a,b)} \operatorname{sign}(\gamma) \,[\mathbf{P}_\gamma]_{ij} \right) U_{b, i}.
\end{align*}
Conversely, applying the map $\Phi_k$ followed by the twisted de Rham differential yields:
\begin{equation*}
    d_k^\nabla \big( \Phi_k(u_{a, j}) \big) = d_k^\nabla (U_{a, j}).
\end{equation*}
By Proposition \ref{prop:differential_action_pref_base}, this expression exactly matches the right-hand side of the previous equation. Since the maps agree on a generic element of the basis, the theorem follows.
\end{proof}

The cochain isomorphism $\Phi$ enables a direct comparison between the numerical torsions of the two complexes. By equipping them with auxiliary inner products adapted to their distinguished bases, their numerical torsions coincide exactly.

\begin{theorem}
\label{thm:equality-of-torsion}
Let $C^\bullet_{TS}(V, \nabla)$ and $C^\bullet_{V, \nabla}(0)$ be equipped with the Hermitian inner products that render their bases, $\mathcal{B}_k$ and $\mathscr{B}_k$, orthonormal. Let $\tau(C^\bullet_{V, \nabla})$ and $\tau(C^\bullet_{TS})$ denote the numerical torsions of the Pollicott--Ruelle resonant complex and the Thom--Smale complex computed with respect to these metrics. It holds that:
\begin{equation}
    \tau(C^\bullet_{V, \nabla}) = \tau(C^\bullet_{TS}).
\end{equation}
\end{theorem}

\begin{proof}
By Definition \ref{def:complex_isomorphism}, the fixed cochain complexes isomorphism $\Phi_k$ acts precisely by mapping the basis elements of $\mathcal{B}_k$ to the corresponding elements of $\mathscr{B}_k$. Since it maps an orthonormal basis to an orthonormal basis, $\Phi_k$ is an isometry by construction. We can thus directly apply Theorem \ref{thm:isomorphic_complexes_torsion}, which ensures that an isometric isomorphism between two finite-dimensional cochain complexes implies the strict equality of their respective numerical torsions.

\end{proof}
\section{The Ruelle Zeta Function and the Analytic Torsion as Regularised Determinants}
\label{sec: Ruelle-zeta-and-Torsion-as-determinants}

This section presents the primary mathematical objects of Fried’s conjecture: the Ruelle zeta function and the Ray--Singer analytic torsion. By expressing both invariants as alternating products of regularised determinants, we provide a unified spectral framework for the analysis of their relation. The investigation begins by recalling the main properties of the linearised Poincaré map, which serves as a useful tool for the spectral realisation of the Ruelle zeta function within the framework of Pollicott--Ruelle resonances. On the other hand, we introduce the Ray--Singer torsion, highlighting its formal similarity to the dynamical zeta function.

\subsection{Linearised Poincaré Map and Orbit Stability}
\label{subsec: linearized-poincare}

To analyse the dynamics in the proximity of a closed orbit $\Lambda$ of minimal period $T_\Lambda > 0$, we employ the Poincaré first-return map. This construction reduces the study of the continuous $n$-dimensional flow to a dynamical system on an $(n-1)$-dimensional transverse section, providing a practical tool to characterise the geometric properties of the orbit.

\begin{definition}
\label{transverse-sec-app}
A local transverse section $\Sigma$ at $p \in \Lambda$ is an $(n-1)$-dimensional embedded submanifold of $M$ such that $p \in \Sigma$ and $T_pM = T_p\Sigma \oplus \mathbb{R}V_p$.
\end{definition}

The existence of such a section is guaranteed by the Tubular Flow Theorem \cite[Theorem 1.1]{Palis82}; specifically, $\Sigma$ can be realised as the zero-level set of a smooth function $f$ such that $d_pf(V_p) \neq 0$.

\begin{definition}
\label{Poincaré-map-app}
Let $\Sigma$ be a transverse section at $p \in \Lambda$. The \textbf{Poincaré map} $\mathbf{P}: U \to \Sigma$, defined on a neighbourhood $U \subset \Sigma$ of $p$, is given by $\mathbf{P}(q) = \phi^{\tau(q)}(q)$, where $\tau(q) \coloneqq \inf \{t > 0 : \phi^t(q) \in U\}$ is the first return time. 
\end{definition}

By construction, $p$ is a fixed point of $\mathbf{P}$ with $\tau(p) = T_\Lambda$. Moreover, $\mathbf{P}$ is a local diffeomorphism from a neighbourhood of $p$ in $\Sigma$ to its image \cite[Proposition 1.2]{Palis82}. To study the transverse linear stability, we introduce the linearised return map.

\begin{definition}
\label{linearised-Poincaré-app}
Let $\Lambda$ be a closed orbit of minimal period $T_\Lambda > 0$ and $\Sigma$ a transverse section at $p \in \Lambda \cap \Sigma$. The {linearised Poincaré map} $\mathcal{P}_{\Lambda, p}: T_p\Sigma \to T_p\Sigma$ is defined as
\begin{equation}
\mathcal{P}_{\Lambda, p}(v) = \frac{d}{ds}\Big|_{s=0} \mathbf{P}(c(s)), \quad \forall v \in T_p\Sigma,
\end{equation}
where $c:(-\epsilon, \epsilon) \to \Sigma$ is any smooth curve such that $c(0)=p$ and $\dot{c}(0)=v$.
\end{definition}

The following proposition formalises the relation between the linearised Poincaré map and the differential of the flow.

\begin{prop}
\label{prop:Poincare-flow-relation}
Given a closed orbit $\Lambda$ and a transverse section $\Sigma$ at $p \in \Lambda$ such that $T_p\Sigma = T^u_p \oplus T^s_p$ (see Equation \eqref{eq:tangent-bundle-decomposition-orbits}), the linearised Poincaré map satisfies:
\begin{equation}
\label{eq:equality-Poincaré-differential-app}
\mathcal{P}_{\Lambda, p}(v) = d_p\phi^{T_\Lambda}(v), \quad \forall v \in T_p\Sigma,
\end{equation}
where $d_p \phi^{T_\Lambda}: T_pM \to T_pM$ is the differential of the flow evaluated at the period $T_\Lambda$.
\end{prop}

\begin{proof}
Let $c(s)$ be a curve in $\Sigma$ such that $c(0) = p$ and $\dot{c}(0) = v$. Expanding the definition of the Poincaré map $\mathbf{P}(c(s)) = \phi^{\tau(c(s))}(c(s))$ and applying the chain rule, we obtain:
\begin{align*}
    \mathcal{P}_{\Lambda, p}(v) &= \frac{d}{ds}\Bigg|_{s=0} \phi^{\tau(c(s))}(c(s)) \\
    &= d_p\phi^{T_\Lambda}(v) + \left( \frac{d}{ds}\Big|_{s=0} \tau(c(s)) \right) V(\phi^{T_\Lambda}(p)) \\
    &= d_p\phi^{T_\Lambda}(v) + (d_p\tau(v))V_p,
\end{align*}
where $d_p\tau(v)$ is the directional derivative of the return time at $p$. By definition, $\mathcal{P}_{\Lambda, p}(v) \in T_p\Sigma$. Furthermore, since $T_p\Sigma = T^u_p \oplus T^s_p$ is a $d_p\phi^{T_\Lambda}$-invariant splitting, we have $d_p\phi^{T_\Lambda}(v) \in T_p\Sigma$. 
Given that $V_p$ is transverse to $T_p\Sigma$, \textit{i.e.}, $V_p \notin T_p\Sigma$, the identity $\mathcal{P}_{\Lambda, p}(v) - d_p\phi^{T_\Lambda}(v) = (d_p\tau(v))V_p$ implies that the vector $(d_p\tau(v))V_p$ must also lie in $T_p\Sigma$. This is only possible if $d_p\tau(v) = 0$, which yields the sought equality.
\end{proof}

Morse--Smale flows are characterised by the non-degeneracy of their closed orbits, meaning that $1$ is never an eigenvalue of $\mathcal{P}_{\Lambda, p}$. This property is intrinsically linked to the orientability of the unstable bundle.

\begin{theorem}
\label{theorem:orientability-and-sign-app}
Let $\Lambda$ be a closed orbit for $V$. The unstable bundle $T^u_\Lambda$ is orientable if and only if $\det(\mathcal{P}_{\Lambda, p}|_{T^u_p}) > 0$ for all $p \in \Lambda$.
\end{theorem}

\begin{proof}
An orientation $\mathcal{O}_p$ on $T^u_p$ is globally consistent if its parallel transport along $\Lambda$ via $d_p\phi^t$ returns to $p$ with the same orientation. Since the transport for a full period $T_\Lambda$ is given by $d_p\phi^{T_\Lambda}|_{T^u_p}$, orientability is equivalent to this map being orientation-preserving, \textit{i.e.}, having a positive determinant. The result follows from Proposition \ref{prop:Poincare-flow-relation}.
\end{proof}

\begin{rem}
A rather straightforward consequence is that the sign of the determinant is related to the twist parameter $\Delta(\Lambda)$ defined in Equation \eqref{eq:twist-operator} by $\operatorname{sign}(\det(\mathcal{P}_{\Lambda, p}|_{T^u_p})) = \Delta(\Lambda)$.
\end{rem}

Two technical propositions concerning the spectral properties of $\mathcal{P}_{\Lambda, p}$ conclude the present section.

\begin{prop}
\label{prop:positivity-det}
For any closed orbit $\Lambda$ of a Morse--Smale flow and for all $j \ge 1$:
\begin{equation}
\label{eq:poincaré-eigenvalues-app}
\det\left(\mathbb{I} - (\mathcal{P}_{\Lambda, p}|_{T^u_p})^{-j}\right) > 0 \quad \text{and} \quad \det\left(\mathbb{I} - (\mathcal{P}_{\Lambda, p}|_{T^s_p})^j\right) > 0.
\end{equation}
\end{prop}

\begin{proof}
Looking at the second inequality, the eigenvalues $\lambda$ of $\mathcal{P}_{\Lambda}|_{T^s_\Lambda}$ satisfy $|\lambda| < 1$. Real eigenvalues $\lambda^j$ yield factors $(1-\lambda^j) > 0$. Complex eigenvalues appear in conjugate pairs $\mu, \bar{\mu}$ with $|\mu| < 1$; their contribution $(1-\mu)(1-\bar{\mu}) = 1 - 2\operatorname{Re}(\mu) + |\mu|^2$ is strictly positive since $|\operatorname{Re}(\mu)| \le |\mu| < 1$. Since the determinant of this operator can be seen as the product of all its eigenvalues and since the product of these positive terms is positive, the statement holds. The first inequality follows in a similar way using eigenvalues $|\lambda|^{-j} < 1$.
\end{proof}

The following result concerns the sign of the determinant of $(\mathbb{I} - \mathcal{P}_{\Lambda, p}^j)$.

\begin{prop}
\label{prop:module-determinant-relation-app}
Given $\mathcal{P}_{\Lambda, p}$ as per Definition \ref{linearised-Poincaré-app}, for every $j \ge 1$, it holds that:
\begin{equation}
    \left|\det\left(\mathbb{I} - \mathcal{P}_{\Lambda, p}^j \right) \right| = (-1)^{\operatorname{ind}(\Lambda)}\Delta(\Lambda)^j\det\left(\mathbb{I} - \mathcal{P}_{\Lambda, p}^j \right),
\end{equation}
where $\operatorname{ind}(\Lambda) = \operatorname{rk}(T^u_\Lambda)$ and $\Delta(\Lambda)$ is the twist parameter defined in Equation \eqref{eq:twist-operator}.
\end{prop}

\begin{proof}
For any invertible matrix $A$, the following linear algebra identity holds:
\begin{equation*}
    \det(\mathbb{I} - A) = \det(-A) \det(\mathbb{I} - A^{-1}).
\end{equation*}
By applying this identity to the restriction of the linearised Poincaré map to the unstable subspace $(\mathcal{P}_{\Lambda}|_{T^u_p})^j$, we obtain:
\begin{align*}
    \det\left(\mathbb{I} - (\mathcal{P}_{\Lambda}|_{T^u_p})^j\right) &= \det\left(-(\mathcal{P}_{\Lambda}|_{T^u_p})^j\right) \det\left(\mathbb{I} - (\mathcal{P}_{\Lambda}|_{T^u_p})^{-j}\right) \\
    &= (-1)^{\operatorname{ind}(\Lambda)} \det\left(\mathcal{P}_{\Lambda}|_{T^u_p}\right)^j \det\left(\mathbb{I} - (\mathcal{P}_{\Lambda}|_{T^u_p})^{-j}\right).
\end{align*}
From Theorem \ref{theorem:orientability-and-sign-app}, the sign of $\det(\mathcal{P}_{\Lambda}|_{T^u_p})$ is given by $\Delta(\Lambda)$. Furthermore, Proposition \ref{prop:positivity-det} ensures that the term $\det(\mathbb{I} - (\mathcal{P}_{\Lambda}|_{T^u_p})^{-j})$ is always strictly positive. Consequently, the sign of the right-hand side is determined by the factor $(-1)^{\operatorname{ind}(\Lambda)} \Delta(\Lambda)^j$.

To conclude, we observe that the total determinant factorises according to the flow-invariant splitting $T_p\Sigma = T^u_p \oplus T^s_p$:
\begin{equation*}
    \det\left(\mathbb{I} - \mathcal{P}_{\Lambda, p}^j\right) = \det\left(\mathbb{I} - (\mathcal{P}_{\Lambda}|_{T^u_p})^j\right) \det\left(\mathbb{I} - (\mathcal{P}_{\Lambda}|_{T^s_p})^j\right).
\end{equation*}
Again by Proposition \ref{prop:positivity-det}, the factor $\det(\mathbb{I} - (\mathcal{P}_{\Lambda}|_{T^s_p})^j)$ is strictly positive for all $j \ge 1$. Therefore, the sign of $\det(\mathbb{I} - \mathcal{P}_{\Lambda, p}^j)$ is identical to the sign of its unstable part, which we have shown to be $(-1)^{\operatorname{ind}(\Lambda)} \Delta(\Lambda)^j$. This implies that:
\begin{equation*}
    \operatorname{sign}\left(\det\left(\mathbb{I} - \mathcal{P}_{\Lambda, p}^j\right)\right) = (-1)^{\operatorname{ind}(\Lambda)} \Delta(\Lambda)^j,
\end{equation*}
which is equivalent to the sought identity.

\end{proof}

\subsection{Ruelle Zeta Function and Regularised Determinants}
\label{sec: Ruelle-zeta-det}

In this section, we briefly recall the definition and properties of the Ruelle zeta function and express its value at zero as an alternating product of regularised determinants of the Lie derivative. Our approach follows the microlocal framework of \cite{DangRiviere17Topology}, with further technical foundations provided in \cite[Section 2.2]{Dyatlov16}, \cite{Guillemin77}, and \cite{Fried87}. Many of the steps presented here are also inspired by the treatment for contact Anosov flows in \cite{HadfieldKandelSchiavina20Ruelle}. Notably, the analysis of the Ruelle zeta function is significantly simplified in the case of Morse--Smale flows due to the finiteness of the set of periodic orbits. For the definition of the zeta function, we follow the conventions of \cite{Shen21Survey}.

\begin{definition}
\label{def:Ruelle-zeta}
Let $V \in \mathfrak{X}(M)$ be a $C^\infty$-linearisable Morse--Smale vector field as per Definition \ref{MS-cl_linearisable-app}, and let $\Delta(\Lambda) \in \{\pm 1\}$ be the twist parameter associated with a closed orbit $\Lambda$ as defined in Definition \ref{def:twist-operator}. The Ruelle zeta function is defined as the finite product:
\begin{equation}
\label{eq:Ruelle-zeta-def}
    \mathfrak{R}_{V, \rho}(z) \coloneqq \prod_{\Lambda \text{ closed orbits}}\det\left(\mathbb{I} - \operatorname{\Delta}(\Lambda) \rho(\Lambda)e^{-T_{\Lambda}z}\right)^{(-1)^{\operatorname{ind}(\Lambda)}},
\end{equation}
where $T_{\Lambda} > 0$ denotes the minimal period of the orbit $\Lambda$ and $\operatorname{ind}(\Lambda) \coloneqq \text{rk}(T^u_\Lambda)$ is the rank of the unstable bundle defined in Equation \eqref{eq:tangent-bundle-decomposition-orbits}.
\end{definition}

The regularity of the Ruelle zeta function is a crucial aspect for its spectral interpretation. In the Morse--Smale setting, the following result clarifies its fundamental analytic properties, see \cite[Section 4]{Shen21Survey} for a proof.

\begin{prop}
\label{prop:ruelle-zeta-well-defined-ms}
There exists a constant $\sigma_0 > 0$ such that $\mathfrak{R}_{V,\rho}(z)$ is a non-zero analytic function in the half-plane $\{z \in \mathbb{C} \mid \operatorname{Re}(z) > \sigma_0\}$. Furthermore, $\mathfrak{R}_{V,\rho}(z)$ admits a meromorphic continuation to the entire complex plane $\mathbb{C}$ and is analytic at $z=0$.
\end{prop}

As we will see, the behaviour of $\mathfrak{R}_{V,\rho}(z)$ at the origin is of particular interest in the context of Fried's conjecture. In the following analysis, we shall always consider the meromorphic extension of the zeta function. We emphasise that the assumption that $V$ be {non-aligned} is essential to ensure that the determinants in Equation \eqref{eq:Ruelle-zeta-def} do not vanish at $z = 0$, thus guaranteeing that $\mathfrak{R}_{V,\rho}(0)$ is well-defined and non-zero. To bridge the gap between the closed orbits and the spectral data of the Lie derivative, it is useful to first decompose the Ruelle zeta function into factors corresponding to the degree $k$ of the anisotropic Sobolev spaces.

\begin{prop}
\label{prop:product-of-zetas}
Let $\mathfrak{R}_{V, \rho}(z)$ be the meromorphic extension of the Ruelle zeta function as per Definition \ref{def:Ruelle-zeta}. Then, the following decomposition holds:
\begin{equation}
\label{eq:def-r_k}
\begin{split}
    \mathfrak{R}_{V, \rho}(z) &= \prod_{k=0}^{n-1}\mathfrak{R}_{k}(z)^{(-1)^k},\\
    \mathfrak{R}_k(z) &\coloneqq \exp \Biggl\{ \sum_{\Lambda, j \ge 1} \frac{e^{-zj T_{\Lambda}}}{j} \frac{\operatorname{tr}(\rho(\Lambda)^j) \operatorname{tr}(\wedge^k\mathcal{P}_{\Lambda}^j)}{\left|\det(\mathbb{I} - \mathcal{P}_\Lambda^j) \right|} \Biggr\},
\end{split}
\end{equation}
where $\mathcal{P}_\Lambda$ denotes the linearised Poincaré map associated with the orbit $\Lambda$.
\end{prop}

\begin{proof}
From linear algebra we know that $\det(\mathbb{I}-A) = \sum_{k=0}^m (-1)^k \operatorname{tr}(\wedge ^k A)$. If we apply it to the definition of $\mathfrak{R}_{V, \rho}(z)$, and use the relation between the determinant and the trace of the logarithm, we obtain:
\begin{align*}
    \log\left(\mathfrak{R}_{V, \rho}(z)\right) &= \sum_{\Lambda} (-1)^{\operatorname{ind}(\Lambda)} \sum_{j=1}^{\infty} \frac{1}{j} e^{-zjT_\Lambda} \Delta(\Lambda)^j \operatorname{tr}(\rho(\Lambda)^j) \sum_{k=0}^{n-1} \frac{(-1)^k \operatorname{tr}(\wedge^k\mathcal{P}_\Lambda^j)}{\det(\mathbb{I} - \mathcal{P}_\Lambda^j)}.
\end{align*}
By Proposition \ref{prop:module-determinant-relation-app}, we have $\det(\mathbb{I} - \mathcal{P}_\Lambda^j)^{-1} = (-1)^{\operatorname{ind}(\Lambda)} \Delta(\Lambda)^j |\det(\mathbb{I} - \mathcal{P}_\Lambda^j)|^{-1}$. Substituting this into the sum, the terms $(-1)^{2\operatorname{ind}(\Lambda)}$ and $\Delta(\Lambda)^{2j}$ vanish, yielding:
\begin{equation*}
    \log\left(\mathfrak{R}_{V, \rho}(z)\right) = \sum_{k=0}^{n-1} (-1)^k \sum_{\Lambda, j} \frac{1}{j} e^{-zjT_\Lambda} \frac{\operatorname{tr}(\rho(\Lambda)^j) \operatorname{tr}(\wedge^k\mathcal{P}_\Lambda^j)}{|\det(\mathbb{I} - \mathcal{P}_\Lambda^j)|}.
\end{equation*}
Exponentiating both sides leads to the sought decomposition.
\end{proof}

Let us introduce the twisted Fuller measure, which defines an additive measure on the non-wandering set $\mathbf{NW}(V)$, assigning independent contributions to each critical element of the Morse--Smale flow.

\begin{definition}
\label{fuller-measure}
The {twisted Fuller measure} $\mu_{V, \nabla} \in \mathcal{D}'(\mathbb{R}^*_+)$ is defined as:
\begin{equation}
\begin{split}
\mu_{V,\nabla}(t) \coloneqq & -\frac{1}{t} \sum_{p \in \operatorname{Crit}(V)} \frac{\det(\mathbb{I} - d_p\phi^t)}{|\det(\mathbb{I} - d_p\phi^t)|} \\
& + \frac{1}{t} \sum_{\Lambda \text{ closed orbits}} T_\Lambda \sum_{j = 1}^{\infty} \frac{\det(\mathbb{I} - \mathcal{P}_\Lambda^j)}{|\det(\mathbb{I} - \mathcal{P}_\Lambda^j)|} \operatorname{tr}(\rho(\Lambda)^j) \delta(t - jT_\Lambda).
\end{split}
\end{equation}
\end{definition}

The presence of Dirac's deltas in $\mu_{V, \nabla}$ is essential for capturing the discrete recurrence of the flow.
To relate this measure to the Ruelle zeta function, we consider its regularised Mellin--Laplace transform and, following the framework of \cite{Fried87}, define an auxiliary function linked to it.

\begin{definition}
\label{def: zeta-flat-function}
The function $\zeta^{\flat}_{V, \nabla}: \mathbb{C} \times \mathbb{C} \to \mathbb{C}$ is defined as:
\begin{equation}
\label{eq: zeta-flat-function}
    \zeta^{\flat}_{V, \nabla}(s, z) \coloneqq \frac{1}{\Gamma(s)} \int_0^{\infty}e^{-tz}t\mu_{V, \nabla}(t) t^{s-1}dt,
\end{equation}
where the integral represents the application of the Fuller measure to the test function $f(t; s, z) \coloneqq e^{-tz}t^s$ and $\Gamma(s)$ is the Euler's Gamma function.
\end{definition}

The following result establishes the link between the Fuller measure and the Ruelle zeta function through a regularisation at $s=0$, see \cite[Section 6.5]{DangRiviere17Topology} for details.

\begin{prop}
\label{prop: Fuller-Zeta-function}
The \textbf{Fuller zeta function} $Z_{V,\nabla}(z) \coloneqq \exp ( - \partial_s \zeta^{\flat}_{V,\nabla}(s,z) |_{s=0} )$ satisfies:
\begin{equation}
\label{eq: explicit-formula-Fuller-Zeta}
Z_{V,\nabla}(z) = z^{-\chi(M, \mathcal{E})} \mathfrak{R}_{V, \rho}(z),
\end{equation}
where $\chi(M,\mathcal{E}) = \sum_{k=0}^n (-1)^k b_k(M, \mathcal{E})$ is the Euler characteristic.
\end{prop}

\begin{rem}
\label{rem: regularization-technical}
The regularisation in Proposition \ref{prop: Fuller-Zeta-function} is well-posed as $\zeta^{\flat}_{V, \nabla}(s,z)$ admits a meromorphic continuation to the entire $s$-plane with a single simple pole at $s=1$, making it analytic in a neighbourhood of $s=0$, see \cite[pages 29-30]{DangRiviere17Topology}. This result relies on the analysis of Mellin transforms for Morse--Smale systems as detailed in \cite[Theorem. 12.4-5]{Apostol76} and \cite[Section 3]{Fried87}.
\end{rem}

The relation between the Fuller measure and the Lie derivative is given by the {Guillemin trace formula}, which relates the measure to the flat-regularised trace of the Lie derivative. According to \cite[Theorem. 8, (II, 22)]{Guillemin77} and \cite[Theorem. 2.6]{DangRiviere17Topology}, we have:
\begin{equation}
\label{eq:spectral-guillemin}
     t\mu_{V, \nabla}(t) = \sum_{k=0}^{n} (-1)^{n-k+1} \operatorname{tr}^{\flat}\left(e^{-t\mathcal{L}_{V,\nabla}^{(k)}\big|_{\text{Ker}(\iota_V)}}\right).
\end{equation}
We can now isolate the contribution of the closed orbits degree by degree.

\begin{prop}
\label{prop:k-trace-formula}
Let $\widetilde{\mathcal{L}}_{V, \nabla}^{(k)}$ be the Lie derivative restricted to $\text{Im}\big(\iota_V\big|_{\text{Ker}\, \pi_0^{(k+1)}}\big)$. The following identity holds:
\begin{equation}
\label{eq:Guillemin-trace-formula-final}
    \operatorname{tr}^{\flat}\left(e^{-t\widetilde{\mathcal{L}}_{V, \nabla}^{(k)}}\right) = (-1)^{n+1}\sum_{\Lambda \text{ closed orbits}}\sum_{j=1}^{\infty} T_\Lambda \frac{\operatorname{tr}(\wedge^k\mathcal{P}_\Lambda^j)}{|\det(\mathbb{I} - \mathcal{P}_\Lambda^j)|} \operatorname{tr}(\rho(\Lambda)^j) \delta(t - jT_\Lambda).
\end{equation}
\end{prop}

\begin{proof}
Consider the chain homotopy decomposition of the Sobolev spaces established in Corollary \ref{cor:direct_sum_decomposition}:
\begin{equation*}
    \mathcal{H}^\bullet_m(M, \mathcal{E}) = C^\bullet_{V, \nabla}(0) \, \oplus \, \text{Im}\left(d^{\nabla}\big|_{\text{Ker}\, \pi_0}\right) \, \oplus \, \text{Im}\left(\iota_V\big|_{\text{Ker}\, \pi_0}\right).
\end{equation*}
First, we characterise the kernel of the contraction map $\iota_V$. From Proposition \ref{prop:zero_operator_consequence}, the restriction $\iota_V|_{C^\bullet_{V,\nabla}(0)}$ vanishes. By nilpotency, $\iota_V$ also vanishes on $\text{Im}\left(\iota_V|_{\text{Ker}\, \pi_0}\right)$. Conversely, on the second summand $\text{Im}\left(d^{\nabla}|_{\text{Ker}\, \pi_0}\right)$, the contraction map is non-vanishing; otherwise, a non-zero element would belong to $\text{Ker}(\mathcal{L}_{V, \nabla})$, which is contained in $C^\bullet_{V, \nabla}(0)$, contradicting the direct sum decomposition. Thus, we have that:
\begin{equation*}
\text{Ker}(\iota_V) = C^\bullet_{V, \nabla}(0) \oplus \text{Im}\left(\iota_V\big|_{\text{Ker}\, \pi_0}\right).
\end{equation*}
Substituting this into Equation \eqref{eq:spectral-guillemin} and using the fact that $C^k_{V, \nabla}(0)$ is the space of generalised eigenvalues of $\mathcal{L}_{V, \nabla}^{(k)}$, we obtain:
\begin{equation*}
t\mu_{V,\nabla}(t) = \sum_{k=0}^{n} (-1)^{n-k+1}\left[\operatorname{dim}(C^k_{V, \nabla}(0)) + \operatorname{tr}^{\flat}\left(e^{-t\widetilde{\mathcal{L}}_{V,\nabla}^{(k)}}\right)\right].
\end{equation*}
By comparing this with Definition \ref{fuller-measure} and applying the spectral identification of the fixed points with resonant states at zero, due to additivity of $\mu_{V, \nabla}(t)$, the terms involving $\operatorname{dim}(C^k_{V, \nabla}(0))$ exactly cancel the fixed-point contributions in the measure. The remaining sum over closed orbits corresponds to the restricted traces. Finally, employing the identity $\det(\mathbb{I} - \mathcal{P}_\Lambda^j) = \sum_{k=0}^{n-1} (-1)^k \operatorname{tr}(\wedge^k \mathcal{P}_\Lambda^j)$ and using the additive property of the twisted Fuller measure on the periodic orbits, we can identify the degree-$k$ components. This yields the sought expression.
\end{proof}

We are now in a position to define its associated regularised determinant, exactly as outlined in Subsection \ref{sec:flat_traces}. For further details we refer to \cite[page 51]{Fried87} and \cite{Guillemin90}.

\begin{definition}
\label{def: 1-step-determinant}
We define the auxiliary holomorphic function:
\begin{equation}
    \mathcal{F}_{\widetilde{\mathcal{L}}_{V, \nabla}^{(k)}}(z, s) \coloneqq \frac{1}{\Gamma(s)}\int_0^{\infty}e^{-tz}t^{s-1}\operatorname{tr}^{\flat}\left(e^{-t \widetilde{\mathcal{L}}_{V,\nabla}^{(k)}}\right)dt,
\end{equation}
where the flat-trace is given by the spectral-dynamical identity in Proposition \ref{prop:k-trace-formula}. 
\end{definition}

The flat-regularised determinant is then obtained by differentiating this auxiliary function with respect to the parameter $s$ and evaluating it at zero.

\begin{definition}
\label{def: 2-step-determinant}
The {flat-regularised determinant} of the operator $(\widetilde{\mathcal{L}}_{V, \nabla}^{(k)} + z)$ is defined by:
\begin{equation}
\log\left({\det}^{\flat}(\widetilde{\mathcal{L}}_{V, \nabla}^{(k)} +z)\right) \coloneqq - \partial_s\vert_{s=0}\mathcal{F}_{\widetilde{\mathcal{L}}_{V,\nabla}^{(k)}}(z, s).
\end{equation}
\end{definition}

To evaluate this expression, we observe that $\mathcal{F}_{\widetilde{\mathcal{L}}_{V,\nabla}}$ is well-defined for $\operatorname{Re}(z)$ sufficiently large. Since $\operatorname{tr}^{\flat} (e^{-t\widetilde{\mathcal{L}}_{V, \nabla}})$ is supported on non-zero multiples of the periods $T_\Lambda$, as per Proposition \ref{prop:k-trace-formula}, it is therefore supported away from $t=0$. Applying the Leibniz rule and noting that $1/\Gamma(s) = s + \mathcal{O}(s^2)$ near $s=0$, the derivative yields the following integral representation:
\begin{equation}
\label{eq: log-det-integral}
\log\left({\det}^{\flat}(\widetilde{\mathcal{L}}_{V, \nabla}^{(k)} +z)\right) = - \int_0^{\infty} e^{-tz} t^{-1} \operatorname{tr}^{\flat}\left(e^{-t \widetilde{\mathcal{L}}_{V,\nabla}^{(k)}}\right) \mathrm{d}t.
\end{equation}

\begin{rem}
\label{rem: r_k-(s=0)}
Equation \eqref{eq: log-det-integral} allows for a direct identification between the spectrum of $\widetilde{\mathcal{L}}_{V, \nabla}^{(k)}$ and the $k$-th component of the Ruelle zeta function. By testing the regularised flat-trace of the restricted Lie derivative against $f(t, z) \coloneqq -t^{-1}e^{-tz}$ and substituting the identity from Proposition \ref{prop:k-trace-formula}, we obtain:
\begin{align*}
    \int_0^{\infty} f(t, z) \operatorname{tr}^{\flat}&\left(e^{-t\widetilde{\mathcal{L}}_{V, \nabla}^{(k)}}\right) dt \\
    &= -(-1)^{n+1}\sum_{\Lambda, j \ge 1} \frac{\operatorname{tr}(\wedge^k\mathcal{P}_\Lambda^j) \operatorname{tr}(\rho(\Lambda)^j)}{|\det(\mathbb{I} - \mathcal{P}_\Lambda^j)|} \int_0^{\infty} T_\Lambda e^{-tz} \delta(t-jT_\Lambda) t^{-1} dt \\
    &= -(-1)^{n+1}\sum_{\Lambda, j \ge 1} \frac{\operatorname{tr}(\wedge^k\mathcal{P}_\Lambda^j) \operatorname{tr}(\rho(\Lambda)^j)}{|\det(\mathbb{I} - \mathcal{P}_\Lambda^j)|} \frac{e^{-jT_\Lambda z}}{j}.
\end{align*}

By simplifying the periods $T_\Lambda$ and comparing the resulting expression with the definition of the $k$-th dynamical factor $\mathfrak{R}_k(z)$ given in Equation \eqref{eq:def-r_k}, it follows that:
\begin{equation}
\label{eq: link-trace-Rk}
     \int_0^{\infty} -t^{-1}e^{-tz} \operatorname{tr}^{\flat}\left(e^{-t\widetilde{\mathcal{L}}_{V, \nabla}^{(k)}}\right) dt = (-1)^{n} \operatorname{log}(\mathfrak{R}_k(z)).
\end{equation}
\end{rem}

Crucially, for our purposes, we must ensure that these determinants are well-behaved at the origin. The following result provides the necessary analytic foundation.

\begin{prop}
\label{prop: analytical-at-zero-app}
The function $\operatorname{log}(\det^{\flat}(\widetilde{\mathcal{L}}_{V, \nabla}^{(k)} +z))$ is analytic in a neighbourhood of $z=0$.    
\end{prop}

\begin{proof}
Consider the semi-group $e^{-t(z + \mathcal{L}_{V, \nabla}^{(k)})}$. As established in Subsection \ref{subsec:resonances_and_algebra}, this operator is bounded for $\operatorname{Re}(z)$ large enough and is linked to its resolvent via the integral
\begin{equation*}
\int_0^{+\infty} e^{-t(z + \widetilde{\mathcal{L}}_{V, \nabla}^{(k)})} dt = (\widetilde{\mathcal{L}}_{V, \nabla}^{(k)} + z)^{-1}.
\end{equation*}
Recalling Proposition \ref{prop:Laurent-expansion}, the resolvent of the full Lie derivative admits a meromorphic extension near $z_0=0$ of the form:
\begin{equation*}
\left(\mathcal{L}_{V, \nabla}^{(k)} + z\right)^{-1} = \sum_{l=1}^{m_k(z_0)} \frac{(-1)^{l-1} \left(\mathcal{L}_{V, \nabla}^{(k)} + z_0\right)^{l-1} \pi_{z_0}^{(k)}}{(z-z_0)^l} + F_{0,k}(z),
\end{equation*}
where $F_{0,k}(z)$ is holomorphic. By restricting the operator to $\text{Im}\big(\iota_V\big|_{\text{Ker}\, \pi_0^{(k+1)}}\big)$, which is contained in $\text{Im}(\mathbb{I} - \pi_0^{(k)})$, the terms associated with the poles (which all contain the projector $\pi_0^{(k)}$ at the numerator) vanish identically. It follows that the restricted resolvent $(\widetilde{\mathcal{L}}_{V, \nabla}^{(k)} + z)^{-1}$ is holomorphic at $z=0$. Given the relation in Equation \eqref{eq: log-det-integral}, this ensures the analyticity of the regularised determinant.
\end{proof}

\begin{theorem}
\label{thm: Ruelle-as-prod-det-app}
The value of the Ruelle zeta function at the origin is given by
\begin{equation}
\mathfrak{R}_{V, \rho}(0) = \prod_{k=0}^{n-1}\left[{\det}^{\flat}(\widetilde{\mathcal{L}}_{V, \nabla}^{(k)})\right]^{(-1)^{n+k}},
\end{equation}
where $\det^{\flat}(\widetilde{\mathcal{L}}_{V, \nabla}^{(k)})$ denotes the flat-regularised determinant evaluated at $z=0$.
\end{theorem}

\begin{proof}
The proof proceeds by considering the expression of the regularised determinant obtained in Equation \eqref{eq: log-det-integral}:
\begin{equation*}
\log\left({\det}^{\flat}(\widetilde{\mathcal{L}}_{V, \nabla}^{(k)} +z)\right) = - \int_0^{\infty} e^{-tz} t^{-1} \operatorname{tr}^{\flat}\left(e^{-t \widetilde{\mathcal{L}}_{V,\nabla}^{(k)}}\right) \mathrm{d}t
\end{equation*}
and the identity established in Remark \ref{rem: r_k-(s=0)}, which brings us to
\begin{equation*}
  \mathfrak{R}_k(z) = {\det}^{\flat}(\widetilde{\mathcal{L}}_{V, \nabla}^{(k)} + z)^{(-1)^{n}}.
\end{equation*}
This identity, which is \textit{a priori} valid for $\operatorname{Re}(z)$ large enough, can be substituted in the product formula derived in Proposition \ref{prop:product-of-zetas}, obtaining:
\begin{align*}
    \mathfrak{R}_{V, \rho}(z) &= \prod_{k=0}^{n-1} \mathfrak{R}_k(z)^{(-1)^k} \\
    &= \prod_{k=0}^{n-1} \left[ {\det}^{\flat}(\widetilde{\mathcal{L}}_{V, \nabla}^{(k)} + z)^{(-1)^n} \right]^{(-1)^k} \\
    &= \prod_{k=0}^{n-1} \left[ {\det}^{\flat}(\widetilde{\mathcal{L}}_{V, \nabla}^{(k)} + z) \right]^{(-1)^{n+k}}.
\end{align*}
To evaluate this identity at $z=0$, we rely on the analytic properties of the objects involved. On the right hand side, Proposition \ref{prop: analytical-at-zero-app} ensures that the regularised determinants $\det^{\flat}(\widetilde{\mathcal{L}}_{V, \nabla}^{(k)} + z)$ are analytic in a neighbourhood of the origin. On the left hand side, as stated in Proposition \ref{prop:ruelle-zeta-well-defined-ms}, the Ruelle zeta function admits a meromorphic continuation to the entire complex plane which is analytic at $z=0$, see also \cite[Section 4]{Shen21Survey}. Therefore, the identity remains valid at the origin.
\end{proof}

\subsection{Analytic Torsion}
\label{sec: analytic-torsion}

We now briefly introduce the Ray--Singer analytic torsion. Its definition exhibits a structure formally analogous to that of the Ruelle zeta function as expressed in Theorem \ref{thm: Ruelle-as-prod-det-app}, \textit{i.e.}, an alternating product of regularised determinants. However, a fundamental distinction must be made regarding both the underlying operators and the functional framework in which they operate: while the Ruelle zeta function is related to the spectral data of the Lie derivative $\mathcal{L}_{V, \nabla}^{(k)}$ acting on the anisotropic Sobolev spaces $\mathcal{H}^k_m(M, \mathcal{E})$, the analytic torsion is defined through the twisted Laplacian $\Delta_k$, which is naturally defined on the space of twisted differential forms.

The  latter is a global invariant taking values in the {determinant line} of the cohomology, $\text{Det}(H^\bullet(M, \mathcal{E}))$. We shall not worry about the well-posedness of this space, since, for a compact manifold $M$, twisted de Rham cohomology groups $H^k(M, \mathcal{E})$ are finite-dimensional vector spaces, a result detailed in \cite[Section 2.3]{DangRiviere17Topology}. Let us fix a specific element within this space: the canonical volume element $\eta_H$.
\begin{equation}
    \eta_H \coloneqq \bigotimes_{k=0}^n \left( h_{k,1} \wedge \dots \wedge h_{k, b_k} \right)^{(-1)^k} \in \text{Det}(H^\bullet(M, \mathcal{E})),
\end{equation}
where $\{h_{k,j}\}$ denotes an orthonormal basis for $H^k(M, \mathcal{E})$ induced by the $L^2$-inner product on harmonic forms. In this regard, we recall that for a compact manifold, the space of twisted differential forms admits the Hodge decomposition (see \cite{Mol26} for a concise recall):
\begin{equation}
\label{eq: Hodge-decomposition}
    \Omega^{k}(M, \mathcal{E}) = \mathcal{H}^k(M, \mathcal{E}) \oplus \text{Im}(d^\nabla_{k-1}) \oplus \text{Im}(d_k^{\nabla, \dagger}),
\end{equation}
where $\mathcal{H}^k(M, \mathcal{E}) \doteq \text{Ker}(\Delta_k)$, the space of harmonic forms, is canonically isomorphic to the $k$-th twisted cohomology group. It is important to remark that while $\eta_H$ is defined on the twisted de Rham cohomology, this space is canonically isomorphic to the cohomology of the complex of resonant states at zero, as well as to the total cohomology of the anisotropic Sobolev spaces $\mathcal{H}^\bullet_m(M, \mathcal{E})$, as seen in Equation \eqref{eq:iso-cohomologies}.

As recalled in \cite[Section 8]{Mnev14} and \cite{Ray71}, $\Delta_k$ is an elliptic, self-adjoint operator characterised by a purely discrete non-negative spectrum. This spectral structure allows us to define its determinant by employing the same regularisation procedure introduced in Definition \ref{def:flat_determinant}.

\begin{definition}
\label{def:flat-det-laplacian}
The \textbf{flat-regularised determinant} of the Laplacian $\Delta_k$, restricted to the orthogonal complement of its kernel, is defined as:
\begin{equation}
    \log \left(\det{}^\flat (\Delta_k)\right)
    \coloneqq -\frac{d}{ds}\bigg|_{s=0}
       \left[ \frac{1}{\Gamma(s)}
              \int_0^\infty t^{s-1}
              \operatorname{tr}^\flat\!\left( e^{-t\Delta_k} - \Pi_0 \right) dt
       \right],
\end{equation}
where $\Pi_0$ is the spectral projector onto the space of harmonic forms $\mathcal{H}^k(M, \mathcal{E})$.
\end{definition}

By combining these determinants across all degrees $k$, we can now provide the formal definition of the torsion as an element within the determinant line of the cohomology.

\begin{definition}
\label{def:Ray--Singer-torsion-flat}
The \textbf{Ray--Singer analytic torsion} $\tau(M, \mathcal{E}) \in \text{Det}(H^\bullet(M, \mathcal{E}))$ is defined as:
\begin{equation}
    \tau(M, \mathcal{E}) \coloneqq \left( \prod_{k=0}^{n} \left(\det{}^\flat(\Delta_k)\right)^{(-1)^{k+1} k/2} \right) \cdot \eta_H.
\end{equation}
\end{definition}

The topological nature of this construction is fundamentally supported by the Cheeger--Müller theorem \cite{Cheeger79, Muller78}, which ensures that $\tau(M, \mathcal{E})$ is independent of the choice of the auxiliary Riemannian metric $g$ and the Hermitian metric $h$ on the bundle.

Furthermore, the numerical part in Definition \ref{def:Ray--Singer-torsion-flat} can be rearranged to focus on the contribution of coexact forms (we refer the reader to \cite{Mol26} for a detailed proof of this identity).

\begin{prop}
\label{prop:torsion-rearrangement-flat}
The analytic torsion can be expressed as:
\begin{equation}
    \tau(M, \mathcal{E}) = \left[ \prod_{k=1}^{n} \left(\det{}^\flat(d^{\nabla, \dagger}_{k-1} d^\nabla_{k-1})\right)^{(-1)^{k+1}/2} \right] \cdot \eta_H,
\end{equation}
where $\det{}^\flat(d^{\nabla, \dagger} d^\nabla)$ denotes the regularised determinant of the Laplacian restricted to the image of the adjoint operator.
\end{prop}

\begin{rem}
\label{rem:zeta-equivalence}
In the traditional literature concerning Ray--Singer torsion \cite{Ray71}, the determinant is typically defined via zeta-regularisation. However, the flat-trace approach employed here is completely equivalent to this standard method. A comprehensive, albeit non-exhaustive, treatment of the classical framework can be found in \cite{MathaiWu11Twisted} or \cite{Shen21Survey}.
\end{rem}
\section{Abelian BF Theory in the BV Formalism} 
\label{sec: BF-theory-and-Schwarz}

In this section, we define the twisted Abelian $BF$ theory and formulate it within the Batalin--Vilkovisky framework, following what has been done in \cite{HadfieldKandelSchiavina20Ruelle, Schiavina23, Mnev19}. This setup provides the necessary tools to handle the theory's reducible gauge symmetries and serves as the formal basis for the investigation of its partition function. We begin by recalling the classical formulation of the theory to highlight the gauge redundancies that necessitate the use of the BV formalism.

\begin{definition}[\textbf{Classical $BF$ Theory}]
\label{def: BF-classical-theory}
The space of classical fields for twisted Abelian $BF$ theory is defined as $\mathcal{F}_{cl} \coloneqq \Omega^1(M, \mathcal{E}) \oplus \Omega^{n-2}(M, \mathcal{E})$. The $BF$ action functional $S_{BF}: \mathcal{F}_{cl} \to \mathbb{C}$ is given by:
\begin{equation}
\label{eq:classical_bf_action}
    S_{BF}[A, B] = \int_M B \wedge_h d^\nabla A,
\end{equation}
where $A \in \Omega^1(M, \mathcal{E})$ and $B \in \Omega^{n-2}(M, \mathcal{E})$.
\end{definition}

The Euler-Lagrange equations are obtained by taking the variation of the action with respect to $A$ and $B$, yielding $d^\nabla A = 0$ and $d^\nabla B = 0$. The critical locus $\mathcal{F}_{\mathcal{EL}}$ thus consists of the subspace of closed fields. 

\begin{rem}
A key feature of $S_{BF}$ is its invariance under the gauge transformations:
\begin{equation}
\label{eq:gauge_shift}
    A \to A + d^\nabla \epsilon, \quad B \to B + d^\nabla \xi,
\end{equation}
for $\epsilon \in \Omega^0(M, \mathcal{E})$ and $\xi \in \Omega^{n-3}(M, \mathcal{E})$. Crucially, the gauge parameters themselves possess redundancies (e.g., $\, \xi \to \xi + d^\nabla \omega$ does not affect the shift of $B$). This "symmetry of a symmetry" continues as a finite tower of reducibilities, suggesting the use of the BV formalism and its unified graded framework, as partly shown in Subsection \ref{sec:BV-data}.
\end{rem}

The classical field space is extended to a $\mathbb{Z}$-graded space of inhomogeneous forms, defined as the direct sum of degree-shifted twisted differential forms:
\begin{equation}
\label{eq:bv_field_space}
    \mathcal{F}_{BF} \coloneqq \Omega^{\bullet}(M, \mathcal{E})[1] \oplus \Omega^{\bullet}(M, \mathcal{E})[n-2].
\end{equation}
Hence, a field configuration is a pair $(\mathbb{A}, \mathbb{B})$ of forms of inhomogeneous degrees. The specific degree shifts in Equation \eqref{eq:bv_field_space} ensure the consistency of the underlying BV algebra, see \cite[Remark 26]{HadfieldKandelSchiavina20Ruelle} for details. Consequently, the components $\mathcal{A}^{(k)} \in \Omega^k(M, \mathcal{E})[1]$ and $\mathcal{B}^{(k)} \in \Omega^k(M, \mathcal{E})[n-2]$ possess BV degrees $|\mathcal{A}^{(k)}|_{BV} = k-1$ and $|\mathcal{B}^{(k)}|_{BV} = k-n+2$, respectively.

The space of fields $\mathcal{F}_{BF}$ is naturally endowed with a symplectic form $\Omega_{BF}$ of degree $-1$:
\begin{equation}
\label{eq:BV_form}
    \Omega_{BF}((0, \mathbb{B}), (\mathbb{A}, 0)) \coloneqq \int_M \left[ \mathbb{B} \wedge_h \mathbb{A} \right]^{\text{top}} = \sum_{k=0}^n \int_M \mathcal{B}^{(k)} \wedge_h \mathcal{A}^{(n-k)}.
\end{equation}
One can readily verify that $\Omega_{BF}$ satisfies the requirement $|\Omega_{BF}|_{BV} = -1$. Indeed, for homogeneous elements $\mathbb{B}$ and $\mathbb{A}$ whose geometric degrees sum to the dimension $n$, the combined BV degree is $(k-n+2) + (n-k-1) = 1$, which is compensated by the $-1$ intrinsic degree of the symplectic pairing to yield a degree-0 scalar.

The dynamics and the gauge structure of the theory are captured by the BV action functional $\mathbb{S}_{BF}: \mathcal{F}_{BF} \to \mathbb{C}$ and the associated BV operator $Q_{BF}$. The first is a degree $0$ functional defined as:
\begin{equation}
\label{eq:BV_action}
    \mathbb{S}_{BF}(\mathbb{A}, \mathbb{B}) \coloneqq \int_M \left[\mathbb{B} \wedge_h d^\nabla \mathbb{A} \right]^{\text{top}},
\end{equation}
while the BV operator $Q_{BF}$ is a degree $1$ vector field defined as:
\begin{equation}
\label{eq:BV_operator}
    Q_{BF}(\mathbb{A}, \mathbb{B}) = (d^\nabla \mathbb{A}, d^\nabla \mathbb{B}).
\end{equation}
The compatibility of the previous definitions with the structural requirements of the BV formalism is established by the following result, which identifies $Q_{BF}$ as the Hamiltonian vector field of the theory.

\begin{prop}
\label{prop:BV-consistency}
The operator $Q_{BF}$ is the Hamiltonian vector field associated with the action $\mathbb{S}_{BF}$ with respect to $\Omega_{BF}$, \textit{i.e.}, $\iota_{Q_{BF}} \Omega_{BF} = - \delta \mathbb{S}_{BF}$. Furthermore, $\mathbb{S}_{BF}$ satisfies the Classical Master Equation:
\begin{equation}
    \{\mathbb{S}_{BF}, \mathbb{S}_{BF}\} = 0,
\end{equation}
where $\{\cdot, \cdot\}$ is the Poisson bracket induced by $\Omega_{BF}$.
\end{prop}

\begin{proof}
By considering the variation of the action functional, we obtain $\delta \mathbb{S}_{BF} = \int_M [\delta \mathbb{B} \wedge_h d^\nabla \mathbb{A} + \mathbb{B} \wedge_h \delta (d^\nabla \mathbb{A})]$. Integrating by parts on a manifold without boundary yields $\delta \mathbb{S}_{BF} = \int_M [\delta \mathbb{B} \wedge_h d^\nabla \mathbb{A} - d^\nabla \mathbb{B} \wedge_h \delta \mathbb{A}]^{\text{top}}$, which matches exactly the expression for $-\iota_{Q_{BF}} \Omega_{BF}$ derived from Equation \eqref{eq:BV_form}.
The CME then follows from the identity $\{\mathbb{S}_{BF}, \mathbb{S}_{BF}\} = Q_{BF}(\mathbb{S}_{BF})$. Applying the BV operator to the integrand, we find:
\begin{equation*}
    Q_{BF}(\mathbb{B} \wedge_h d^\nabla \mathbb{A}) = (d^\nabla \mathbb{B}) \wedge_h d^\nabla \mathbb{A} + (-1)^{|\mathbb{B}|} \mathbb{B} \wedge_h (d^\nabla)^2 \mathbb{A} = d^\nabla(\mathbb{B} \wedge_h d^\nabla \mathbb{A}).
\end{equation*}
Since $(d^\nabla)^2 = 0$, the integrand reduces to an exact form whose integral vanishes by Stokes' Theorem, thus establishing that $\{\mathbb{S}_{BF}, \mathbb{S}_{BF}\} = 0$.
\end{proof}

To conclude, we give the following definition.

\begin{definition}
\label{def: BF-data-in-BV-formalism}
The assignment $(M, \mathcal{E}) \rightsquigarrow (\mathcal{F}_{BF}, \Omega_{BF}, \mathbb{S}_{BF}, Q_{BF})$ defines the twisted Abelian $BF$ theory in the Batalin--Vilkovisky formalism.
\end{definition}

\section
{Analytic Torsion from Metric Gauge Fixing} 
\label{sec: metric-gauge}

The relationship between quadratic field theories and topological invariants was pioneered by Schwarz in \cite{Sch78, Sch79}, who demonstrated that the formal path integral of such theories computes the analytic torsion when the twisted de Rham cohomology vanishes, \textit{i.e.}:
\begin{equation}
\label{eq: partition-function-torsion}
    \mathcal{Z}_{Sch} = \prod_{k=1}^{n} \left(\det\nolimits^{\flat} \left(d^{\nabla, \dagger}_k d^\nabla_k\right)\right)^{(-1)^{k}/2}.
\end{equation}
Specifically, this construction restricts the integration to the subspace of co-exact forms $\text{Im}(d^{\nabla, \dagger}_k)$, where the relevant Laplacian-type operators are invertible. This procedure admits a natural and rigorous reformulation within the BV framework when applied to the twisted Abelian $BF$ theory introduced in Definition \ref{def: BF-classical-theory}, as established in \cite{HadfieldKandelSchiavina20Ruelle} for the original acyclic case.

Since we intend to treat the theory introduced in Section \ref{sec: BF-theory-and-Schwarz} in which the field space possesses nontrivial cohomology, we implement the evaluation of the partition function following the general paradigm for infinite-dimensional quadratic action functional established in Subsection \ref{sec:inf-dim-gaussian}. In this context, by choosing an auxiliary Riemannian metric $g$ on the manifold, we rely on the Hodge decomposition to partition the field space into a finite-dimensional harmonic sector and an infinite-dimensional fluctuating complement. Within the latter, we then single out a gauge-fixing Lagrangian subspace to implement the BV pushforward paradigm.

Central to the discussion in this and the next section is the interplay between the Ray–Singer and Milnor metrics, whose definitions follow the treatment in \cite{Shen21}. We invite the reader to consult that work and the review \cite{Mol26} for an exhaustive account of these topics.

\vspace{\baselineskip}

To rigorously implement this metric gauge-fixing and the associated Hodge decomposition, we extend the classical field space $\mathcal{F}_{cl}$ to its $L^2$ Hilbert space completion, which we denote by $\mathcal{F}_{BF}^{L^2}$ (we refer to \cite[Lemma 2.1 and Corollary 2.5]{BruningLesch1992} for the rigorous Hilbert complex formulation). Note that the finite-dimensional notion of a Lagrangian subspace seen in Definition \ref{def:lagrangian_subspace} naturally extends to this infinite-dimensional Hilbert space setting, simply requiring a direct sum decomposition into two isotropic subspaces.

\begin{definition}
\label{def: decomposition-F-BF}
    We define the cohomological and fluctuating subspaces of $\mathcal{F}_{BF}^{L^2}$ respectively as:
\begin{equation}
\begin{split}
    \mathcal{F}_{BF}' &\coloneqq \mathcal{H}^{\bullet}(M, \mathcal{E})[1] \bigoplus \mathcal{H}^{\bullet}(M, \mathcal{E})[n-2], \\
    \mathcal{F}_{BF}'' &\coloneqq \left(\text{Im}\left(d^\nabla_{\bullet-1}\right) \oplus \text{Im}\left(d^{\nabla, \dagger}_{\bullet}\right)\right)[1] \bigoplus \left(\text{Im}\left(d^\nabla_{\bullet-1}\right) \oplus \text{Im}\left(d^{\nabla, \dagger}_{\bullet}\right)\right)[n-2],
\end{split}
\end{equation}
where the images of the operators are understood as their closures in the $L^2$ topology.
\end{definition}

As a direct consequence of the $L^2$ Hodge decomposition, we obtain an orthogonal splitting of the BV field space:
\begin{equation}
\label{eq:field_splitting}
\mathcal{F}_{BF}^{L^2} = \mathcal{F}_{BF}' \oplus \mathcal{F}_{BF}''.
\end{equation}

The finite-dimensionality of the residual sector $\mathcal{F}_{BF}'$ ensures that the associated partition function is well-defined. Consequently, the divergences inherent to the path integral are localised within the infinite-dimensional complement $\mathcal{F}_{BF}''$, which serves as the concrete realisation of the fluctuation sector introduced in the general framework of Subsection \ref{sec:inf-dim-gaussian}. Following the paradigm established therein, we implement the gauge-fixing procedure by selecting a Lagrangian subspace $\mathbb{L}_g$ of $\mathcal{F}_{BF}''$ alone, thereby leaving the cohomological degrees of freedom untouched.

\begin{definition}
\label{def: Lagrangian-metric}
    The \textbf{metric gauge} subspace $\mathbb{L}_g \subset \mathcal{F}_{BF}''$ is defined as:
    \begin{equation}
        \mathbb{L}_g \coloneqq \text{Im}\left(d^{\nabla, \dagger}_{\bullet}\right)[1] \bigoplus \text{Im}\left(d^{\nabla, \dagger}_{\bullet}\right)[n-2].
    \end{equation}
\end{definition}

The next step is to prove that $\mathbb{L}_g$ indeed satisfies the requirements to be a valid gauge-fixing subspace.

\begin{prop}
\label{prop: L_g-lagrangian-subspace}
    The subspace $\mathbb{L}_g$ is a Lagrangian subspace of $\mathcal{F}_{BF}''$ with respect to the symplectic form $\Omega_{BF}$.
\end{prop}

\begin{proof}
    Let us consider the complement $\mathbb{L}_g^\perp \coloneqq \text{Im}(d^{\nabla}_{\bullet-1})[1] \oplus \text{Im}(d^{\nabla}_{\bullet-1})[n-2]$. From Definition \ref{def: decomposition-F-BF} and Hodge Theorem, it follows that $\mathcal{F}_{BF}'' = \mathbb{L}_g \oplus \mathbb{L}_g^\perp$. Eventually, we check the isotropy of both subspaces.
    
    \noindent \textbf{Isotropy of $\mathbb{L}_g$}: Let $\mathbb{A} \in \text{Im}(d^{\nabla, \dagger}_{k})[1]$ and $\mathbb{B} \in \text{Im}(d^{\nabla, \dagger}_{n-k})[n-2]$ be two homogeneous elements, so that $\mathbb{A} = d^{\nabla, \dagger}_k \alpha$ and $\mathbb{B} = d^{\nabla, \dagger}_{n-k} \beta$. Using the identity $d^{\nabla, \dagger}_k = (-1)^{nk+1} *_h d^\nabla_{n-k-1} *_h$, see \cite[Definition 6.1]{Warner1983}, the pairing is evaluated as:
    \begin{align*}
        \Omega_{BF}((0, \mathbb{B}), (\mathbb{A}, 0)) &= \int_M d^{\nabla, \dagger}_{n-k}\beta \wedge_h d^{\nabla, \dagger}_k \alpha \\
        &= (-1)^{nk+1} \int_M d^{\nabla, \dagger}_{n-k}\beta \wedge_h *_h (d^{\nabla}_{n-k-1} (*_h \alpha)).
    \end{align*}
    By definition of the $L^2$ inner product on twisted differential forms and the adjoint property of the codifferential, we have:
    \begin{align*}
        \Omega_{BF}((0, \mathbb{B}), (\mathbb{A}, 0)) &= (-1)^{nk+1} \left(d^{\nabla, \dagger}_{n-k}\beta, d^{\nabla}_{n-k-1} (*_h\alpha)\right)_{L^2} \\
        &= (-1)^{nk+1} \left(\beta, d^{\nabla}_{n-k} \circ d^{\nabla}_{n-k-1} (*_h\alpha)\right)_{L^2} = 0,
    \end{align*}
    vanishing due to the nilpotence of the de Rham differential.

    \noindent \textbf{Isotropy of $\mathbb{L}_g^\perp$}: For $\mathbb{A} = d^{\nabla}_{k-1} \alpha$ and $\mathbb{B} = d^{\nabla}_{n-k-1} \beta$, the Leibniz rule and the nilpotence of $d^\nabla$ yield:
    \begin{align*}
        \Omega_{BF}((0, \mathbb{B}), (\mathbb{A}, 0)) &= \int_M d^{\nabla}_{n-k-1}\beta \wedge_h d^{\nabla}_{k-1} \alpha \\
        &= \int_M d(\beta \wedge_h d^{\nabla}_{k-1} \alpha) = 0,
    \end{align*}
    where the integral vanishes by Stokes' Theorem on the boundaryless manifold $M$. Thus, $\mathbb{L}_g$ is a Lagrangian subspace.
\end{proof}

\begin{rem}
\label{rem: BF-action-restricted-metric}
It is important to note that the BV action functional depends exclusively on the non-harmonic sector of the field space. Indeed, by considering the decomposition $\mathbb{A} = \mathbb{A}_H + \mathbb{A}_{CE}$ and $\mathbb{B} = \mathbb{B}_H + \mathbb{B}_{CE}$ established in Equation \eqref{eq:field_splitting}\footnote{The subscript $CE$ refers to the components of the fields belonging to the non-harmonic sector, standing for the direct sum of the {Exact} and {Co-exact} parts of the Hodge decomposition.}, the action functional reads:
\begin{equation*}
    \mathbb{S}_{BF}(\mathbb{A}, \mathbb{B}) = \int_M \left[ (\mathbb{B}_H + \mathbb{B}_{CE}) \wedge_h d^\nabla (\mathbb{A}_H + \mathbb{A}_{CE}) \right]^{\text{top}}.
\end{equation*}
Since harmonic forms are $d^\nabla$-closed (see \cite[Proposition 6.3]{Warner1983}), the terms involving $d^\nabla \mathbb{A}_H$ vanish identically. Furthermore, the orthogonality between harmonic forms and exact forms ensures that the cross-term $\int_M \mathbb{B}_H \wedge_h d^\nabla \mathbb{A}_{CE}$ also vanishes. Consequently, the action reduces to:
\begin{equation}\label{e:metricEffact}
    \mathbb{S}_{BF}(\mathbb{A}, \mathbb{B}) = \int_M \left[\mathbb{B}_{CE} \wedge_h d^\nabla \mathbb{A}_{CE} \right]^{\text{top}}.
\end{equation}
This derivation implies that the effective action $\mathbb{S}_{BF}'$ on the cohomological sector $\mathcal{F}_{BF}'$ is identically zero, a result that will be central to the definition of the gauge-fixed partition function.
\end{rem}

Before proceeding, we provide a formal meaning for the determinant of the twisted de Rham differential. By analogy with the finite-dimensional case of non-square matrices, we define the regularised determinant of the differential as:
\begin{equation}
\label{eq:sdet_of_d}
    {\det}^\flat (d^\nabla_k) \equiv {\det}^\flat (d^{\nabla, \dagger}_k) \coloneqq \left|{\det}^\flat (d^{\nabla, \dagger}_k  d^\nabla_k)\right|^{1/2},
\end{equation}
where $\det^\flat (d^{\nabla, \dagger}_k d^\nabla_k)$ is the flat-trace regularised determinant of the Laplacian restricted to the space of co-exact forms, as introduced in Definition \ref{def:flat-det-laplacian}. This notion is crucial for the consistent treatment of the Jacobian factors in the path integral.

\vspace{\baselineskip}

We note that the action functional restricted to $\mathbb{L}_g$ assumes a quadratic form.

\begin{prop}
\label{prop: action-restricted-metric}
 The restricted BV action functional satisfies:
\begin{equation}
\label{eq:def-det-differential}
    \mathbb{S}_{BF}|_{\mathbb{L}_g} = \sum_{k=0}^{n-1} \left(\tau^{(k)}, d_k^{\nabla, \dagger} d_k^\nabla \mathcal{A}^{(k)} \right),
\end{equation}
where $\mathcal{A}^{(k)}$ is the $k$-th component of the field $\mathbb{A}$ and the forms $\tau^{(k)}$ parametrise the co-exact components of $\mathbb{B}$.
\end{prop}

\begin{proof}
By decomposing the fields into homogeneous components, the action reads \begin{equation*}
\mathbb{S}_{BF} = \sum \int_M \mathcal{B}^{(n-k-1)} \wedge_h d^\nabla_k \mathcal{A}^{(k)}.
\end{equation*}
On $\mathbb{L}_g$ the components of $\mathbb{B}$ are co-exact and can be parametrised as $\mathcal{B}^{(n-k-1)} = *_h d_k^{\nabla} \tau^{(k)}$, as a consequence of \cite[Definition 6.1]{Warner1983}. Substituting this into the integral and using the $L^2$ inner product definition, we obtain:
\begin{equation*}
    \mathbb{S}_{BF}|_{\mathbb{L}_g} = \sum_{k=0}^{n-1} \left(d_k^{\nabla} \tau^{(k)}, d^\nabla_k \mathcal{A}^{(k)}\right) = \sum_{k=0}^{n-1} \left(\tau^{(k)}, d_k^{\nabla, \dagger} d_k^\nabla \mathcal{A}^{(k)}\right),
\end{equation*}
where we applied the definition of codifferential.
\end{proof}

The combination of the restricted action with the rules of Berezinian integration on shifted spaces (see \cite[Section 3.8]{Mnev19}) provides a rigorous definition for the partition function.

\begin{definition}
\label{def:partition_function_metric_gauge}
The partition function of $BF$ theory in the metric gauge $\mathbb{L}_g$ is defined as:
\begin{equation}
\label{eq:partition_function_metric_gauge}
\begin{split}
    \mathcal{Z}\, '(\mathbb{S}_{BF}, \mathbb{L}_g) 
    &\equiv \mu' \cdot \left| \sdet^\flat(d^{\nabla, \dagger})\right|^{-1} \cdot \int_{\mathbb{L}_g} \exp \left( i \sum_{k=0}^{n-1} \left(\tau^{(k)}, d_k^{\nabla, \dagger} d^\nabla_k \mathcal{A}^{(k)}\right)\right) \\
    &\coloneqq \mu' \cdot \left| \sdet^\flat(d^{\nabla, \dagger})\right|^{-1} \cdot \left| \sdet^\flat(d^{\nabla, \dagger} d^\nabla) \right|,
\end{split}
\end{equation}
where $\mu' \in \text{Det}(\mathcal{F}_{BF}')$ is the volume element on the cohomology part of the field space.
\end{definition}

\begin{rem}
\label{rem:metric-partition-notation}
Note the slight change of notation with respect to Equation~\eqref{eq:residual_partition_function}. Here we
use a more precise notation and write $\mathcal{Z}'(\mathbb{S}_{BF},\mathbb{L}_g)$, stressing the choice of a particular field theory and the dependence on the gauge-fixing subspace, in place of $\mathcal{Z}'(A)$, which highlighted instead the operator appearing in the quadratic action. The two notations are
reconciled by $\mathcal{Z}'(\mathbb{S}_{BF},\mathbb{L}_g) = \mathcal{Z}'(\Delta)$,
where $d^{\nabla,\dagger}d^{\nabla}$ is the generic operator $A$, which restricted to the co-exact sector singled out by $\mathcal{L}_g$ reduces to $A''=d^{\nabla,\dagger}d^{\nabla}$.
\end{rem}

\begin{rem}
\label{rem: rationale-choice-metric}
The expression in Equation \eqref{eq:partition_function_metric_gauge} accounts for both the Jacobian determinant arising from the change of variables and the evaluation of the resulting infinite-dimensional Berezinian integral. As discussed in \cite[Remark 26]{HadfieldKandelSchiavina20Ruelle}, the grading shifts inherent to the BV formalism lead to an inversion of the expected powers: the Jacobian determinant appears with a negative exponent, while the Gaussian integration yields the determinant directly. In addition, the vanishing of the effective action on the harmonic sector, as established in Remark \ref{rem: BF-action-restricted-metric}, explains why the partition function does not present any remaining exponential terms.
\end{rem}

The link between the partition function and the analytic torsion introduced in Definition \ref{def:Ray--Singer-torsion-flat} is straightforward.

\begin{cor}
\label{cor:partition-torsion-link}
The metric gauge-fixed partition function equals the Ray--Singer torsion:
\begin{equation}
    \mathcal{Z}\, '(\mathbb{S}_{BF}, \mathbb{L}_g) = \left( \prod_{k=0}^{n} ({\det}^\flat \Delta_k)^{(-1)^{k+1} k/2} \right) \cdot \mu'.
\end{equation}
\end{cor}

\begin{proof}
Substituting $\sdet^\flat (d^{\nabla, \dagger}) = |\sdet^\flat (d^{\nabla, \dagger}d^\nabla)|^{1/2}$ from Equation \eqref{eq:def-det-differential} into Definition \ref{def:partition_function_metric_gauge}, we obtain $\mathcal{Z}' =|\sdet^\flat(d^{\nabla, \dagger} d^\nabla)|^{1/2} \cdot \mu'$. Expressing this in terms of the degree components and using Proposition \ref{prop:torsion-rearrangement-flat} gets us to the sought statement.
\end{proof}

To conclude, we recall that within the BV pushforward paradigm, the partition function is formally defined as an element of the space of half-densities over the residual fields, $\mathrm{Dens}^{1/2}(\mathcal{F}_{BF}')$ (see Definition \ref{def:BV_pushforward}). Given that the residual space $\mathcal{F}_{BF}'$ is finite-dimensional, this half-density is canonically identified with an element of the determinant line $\text{Det}(\mathcal{F}_{BF}')$ through the isomorphism described in \cite[Equation 30]{Cattaneo20}, justifying the presence of the volume element $\mu'$. Ultimately, through the canonical isomorphism $\mathcal{H}^\bullet(M, \mathcal{E}) \cong H^\bullet(M, \mathcal{E})$ arising from the Hodge decomposition, the metric partition function yields an element of the determinant line of the cohomology $\text{Det}(H^\bullet(M, \mathcal{E}))$, whose Ray--Singer norm can be taken as per \cite[Definition 3.9]{Shen21}.
\section{Ruelle Zeta Function from Axial Gauge Fixing} 
\label{sec: axial-gauge}

An alternative gauge-fixing, hereafter referred to as the {axial gauge}, is provided by the choice of a Morse--Smale vector field $V$. While the $L^2$ completion $\mathcal{F}_{BF}^{L^2}$ was perfectly suited for the metric gauge, it is analytically inadequate here because of the non-discrete spectrum of the Lie derivative acting on it. To provide a consistent mathematical framework, the classical BV field space is instead extended to the anisotropic Sobolev spaces $\mathcal{H}_m^\bullet(M, \mathcal{E})$ detailed in Appendix \ref{app:Sobolev}. 

As established in Proposition \ref{prop:anisotropic_spaces} and \cite{DangRiviere17Topology}, these anisotropic spaces are topologically isomorphic to $L^2$ as Hilbert spaces. However, the separation between zero-modes and fluctuations here is not driven by the Hodge theorem, but rather by the dynamically-induced chain homotopy equation. 

\begin{rem}
\label{rem: BV-data-Sob}
We define the extended graded space of $BF$ fields as:
\begin{equation}
\label{eq:bv_field_space_Sob}
    \widetilde{\mathcal{F}}_{BF} \coloneqq \mathcal{H}_m^\bullet(M, \mathcal{E})[1] \oplus \mathcal{H}_m^\bullet(M, \mathcal{E})[n-2].
\end{equation}
Crucially, we have a topological isomorphism of the total field spaces $\widetilde{\mathcal{F}}_{BF} \simeq \mathcal{F}_{BF}^{L^2}$. The action functional $\widetilde{\mathbb{S}}_{BF}$ and the symplectic $2$-form $\widetilde{\Omega}_{BF}$ are defined as the continuous extensions of the classical functionals to $\widetilde{\mathcal{F}}_{BF}$.
\end{rem}

To implement the BV pushforward, we decompose $\widetilde{\mathcal{F}}_{BF}$.

\begin{definition}
\label{def: decomposition-F-BF-Sob}
    Let $\widetilde{\mathcal{F}}_{BF}'$ and $
\widetilde{\mathcal{F}}_{BF}''$ be defined as:
\begin{equation}
\begin{split}
    \widetilde{\mathcal{F}}_{BF}' &\coloneqq C^{\bullet}_{V, \nabla}(0)[1] \oplus C^{\bullet}_{V, \nabla}(0)[n-2], \\
    \widetilde{\mathcal{F}}_{BF}'' &\coloneqq \left(\operatorname{Im}\big(d^{\nabla}\big|_{\operatorname{Ker} \pi_0}\big) \oplus \operatorname{Im}\big(\iota_V\big|_{\operatorname{Ker}\pi_0}\big)\right)[1] \\
    &\qquad \oplus \left(\operatorname{Im}\big(d^{\nabla}\big|_{\operatorname{Ker} \pi_0}\big) \oplus \operatorname{Im}\big(\iota_V\big|_{\operatorname{Ker}\pi_0}\big)\right)[n-2].
\end{split}
\end{equation}
\end{definition}

As a fundamental consequence of Equation \eqref{eq:chain-homotopy-decomp}, the extended space of fields $\widetilde{\mathcal{F}}_{BF}$ admits the following splitting:
\begin{equation}
\label{eq:bv_splitting_Sob}
    \widetilde{\mathcal{F}}_{BF} = \widetilde{\mathcal{F}}_{BF}' \oplus \widetilde{\mathcal{F}}_{BF}''.
\end{equation}
A key feature of this decomposition is that the subspace $\widetilde{\mathcal{F}}_{BF}'$, being generated by the Pollicott--Ruelle resonant states at zero, is finite-dimensional, under the hypotheses presented in Proposition \ref{prop:dimension-of-Ck}.
The axial gauge is thus defined by choosing a Lagrangian subspace within $\widetilde{\mathcal{F}}_{BF}''$.

\begin{prop}
\label{prop: axial_gauge_lagrangian}
The subspace $\mathbb{L}_V \subset \widetilde{\mathcal{F}}_{BF}''$, defined as:
\begin{equation}
\label{eq:axial_gauge_subspace}
    \mathbb{L}_V \coloneqq \text{Im}\Big(\iota_V\big|_{\text{Ker}\, \pi_0}\Big)[1] \oplus \text{Im}\Big(\iota_V\big|_{\text{Ker}\, \pi_0}\Big)[n-2],
\end{equation}
is a Lagrangian subspace of $\widetilde{\mathcal{F}}_{BF}''$ with respect to the symplectic structure $\widetilde{\Omega}_{BF}$.
\end{prop}

\begin{proof}
    To verify the Lagrangian property, we consider the complementary subspace 
    \begin{equation*}
    \mathbb{L}_V^\perp \coloneqq \text{Im}\left(d^{\nabla}\big|_{\text{Ker} \, \pi_0}\right)[1] \oplus \text{Im}\left(d^{\nabla}\big|_{\text{Ker} \, \pi_0}\right)[n-2].
    \end{equation*}
    The direct sum decomposition $\widetilde{\mathcal{F}}_{BF}'' = \mathbb{L}_V \oplus \mathbb{L}_V^\perp$ is a manifest consequence of the definition of $\widetilde{\mathcal{F}}_{BF}''$. We then proceed to show the isotropy of both subspaces.
    
    The isotropy of $\mathbb{L}_V$ follows from the Leibniz property and the nilpotency of the interior product. For any homogeneous components $\mathbb{B} = \iota_V \beta$ and $\mathbb{A} = \iota_V \alpha$ in $\mathbb{L}_V$, the pairing evaluates to:
    \begin{equation*}
        \widetilde{\Omega}_{BF}((0, \mathbb{B}), (\mathbb{A}, 0)) = \int_M \iota_V \beta \wedge_h \iota_V \alpha = \int_M \iota_V(\beta \wedge_h \iota_V \alpha).
    \end{equation*}
    Since $\beta \wedge_h \iota_V \alpha$ is a differential form of degree $n+1$ on an $n$-dimensional manifold, it is identically zero, hence the integral vanishes. Similarly, the isotropy of $\mathbb{L}_V^\perp$ is established by homogeneous elements $\mathbb{B} = d^\nabla \beta$ and $\mathbb{A} = d^\nabla \alpha$. The Leibniz rule for the twisted differential yields $\int_M d^\nabla \beta \wedge_h d^\nabla \alpha = \int_M d^\nabla(\beta \wedge_h d^\nabla \alpha)$, which vanishes by Stokes' theorem on the boundaryless manifold $M$.
\end{proof}

The action functional decouples between the resonant part of the fields and the fluctuations. This property is formalised in the following result.

\begin{prop}
\label{prop: action-splitting-axial}
Let $(\mathbb{A}, \mathbb{B}) \in \widetilde{\mathcal{F}}_{BF}$ be a field configuration. The action functional $\widetilde{\mathbb{S}}_{BF}$ splits into:
\begin{equation}
\label{eq:action_splitting}
    \widetilde{\mathbb{S}}_{BF}(\mathbb{A}, \mathbb{B}) = \int_M \left[B_0 \wedge_h d^\nabla A_0\right]^{\text{top}} + \int_M \left[B'' \wedge_h d^\nabla A''\right]^{\text{top}},
\end{equation}
where $({A}_0, {B}_0) \in \widetilde{\mathcal{F}}_{BF}'$ and $({A}'', {B}'') \in \widetilde{\mathcal{F}}_{BF}''$. 
\end{prop}

\begin{proof}
By writing $\mathbb{A} = A_0 + A''$ and $\mathbb{B} = B_0 + B''$ and substituting them into the definition of $\widetilde{\mathbb{S}}_{BF}$, we obtain two diagonal terms and two mixed terms. We prove that the mixed terms, $\int_M B_0 \wedge_h d^\nabla A''$ and $\int_M B'' \wedge_h d^\nabla A_0$, vanish identically.

First, recall that the elements of $\widetilde{\mathcal{F}}_{BF}'$ are generalised eigenvectors of $\mathcal{L}_{V, \nabla}$ relative to the zero eigenvalue. Thus, for $B_0 \in C^\bullet_{V,\nabla}(0)$, there exists a minimal integer $m \ge 1$ such that $(\mathcal{L}_{V, \nabla})^m B_0 = 0$. Conversely, the Lie derivative is invertible on $\widetilde{\mathcal{F}}_{BF}''$, see Item~(III) in Proposition~\ref{prop:operator_toolbox}. 

Consider the term $\int_M B_0 \wedge_h d^\nabla A''$. Since $A''$ belongs to the fluctuation sector, we can write $A'' = \mathcal{L}_{V, \nabla} \theta$ for some $\theta \in \text{Im}(\mathbb{I} - \pi_0)$. Recalling that $d^\nabla$ and $\mathcal{L}_{V, \nabla}$ commute, see Item~(I) in Proposition~\ref{prop:operator_toolbox}, we have:
\begin{align*}
    \int_M B_0 \wedge_h d^\nabla A'' &= \int_M B_0 \wedge_h \mathcal{L}_{V, \nabla} (d^\nabla \theta). \\
    \intertext{Applying the Leibniz rule for the Lie derivative, we obtain:}
    &= \int_M \mathcal{L}_{V, \nabla} (B_0 \wedge_h d^\nabla \theta) - \int_M (\mathcal{L}_{V, \nabla} B_0) \wedge_h d^\nabla \theta.
\end{align*}
The first integral on the right-hand side vanishes by Stokes' theorem, as $M$ is a manifold without boundaries. The remaining term involves $\mathcal{L}_{V, \nabla} B_0$. Since $B_0$ is a generalised eigenvector, we can iteratively repeat this integration by parts. After $m$ iterations, the integrand will contain the factor $(\mathcal{L}_{V, \nabla})^m B_0$, which is identically zero by definition. 

An identical argument applies to the second mixed term $\int_M B'' \wedge_h d^\nabla A_0$. By writing $B'' = \mathcal{L}_{V, \nabla} \eta$ and using the nilpotency of the Lie derivative on $A_0 \in C^\bullet_{V,\nabla}(0)$, the integral vanishes.
\end{proof}

The previous result confirms that the effective action $\widetilde{\mathbb{S}}'_{BF}$ on the finite-dimensional resonant sector is non-trivial, unlike in the metric gauge (c.f. Equation \eqref{e:metricEffact}). To evaluate the contribution of the fluctuations to the partition function, we now provide an explicit quadratic representation of the action functional restricted to $\mathbb{L}_V$.

\begin{prop}
\label{prop: action-restricted-axial}
Let $\widetilde{\mathcal{L}}_{V, \nabla}^{(k)}$ be the Lie derivative restricted to the subspace $\text{Im}\big(\iota_V\big|_{\text{Ker}\, \pi_0^{(k+1)}}\big)$. On the axial Lagrangian $\mathbb{L}_V$, the action functional admits the quadratic form:
\begin{equation}
\widetilde{\mathbb{S}}_{BF}|_{\mathbb{L}_V} = \sum_{k=0}^{n-1} \int_M \left(\mathcal{B}^{(n-k)}\wedge_h \widetilde{\mathcal{L}}_{V, \nabla}^{(k)} \eta^{(k)} \right),
\end{equation}
where $\mathcal{B}^{(n-k)}$ is the $(n-k)$-th component of $\mathbb{B} \in \widetilde{\mathcal{F}}_{BF}''$ and $\eta^{(k)}$ parametrises the field $\mathbb{A}$ restricted to the axial gauge.
\end{prop}

\begin{proof}
By definition of $\mathbb{L}_V$, each field component satisfies $\mathcal{A}^{(k)} = \iota_V \eta^{(k+1)}$ for some $\eta^{(k+1)} \in \text{Im}(\mathbb{I} - \pi_0)$. Substituting this into the action $\widetilde{\mathbb{S}}_{BF} = \sum \int_M \mathcal{B}^{(n-k-1)} \wedge_h d^\nabla \mathcal{A}^{(k)}$, we focus on the term $d^\nabla \iota_V \eta^{(k+1)}$. Using Theorem \ref{thm:chain_homotopy}, and noting that $\eta^{(k+1)}$ can be chosen such that $d^\nabla \eta^{(k+1)} = 0$, it follows that $d^\nabla \iota_V \eta^{(k+1)} = \mathcal{L}_{V,\nabla} \eta^{(k+1)}$. The identity is then obtained by substituting this expression back into the integral and re-indexing the summation over the form degrees.
\end{proof}

By combining the quadratic representation of the action on the axial gauge-fixing with the formal rules of Berezinian integration on degree-shifted spaces, we can provide a rigorous definition for the partition function. Analogously to the metric gauge case, this definition incorporates the Jacobian factors arising from the change of coordinates and the evaluation of the infinite-dimensional Gaussian integral over the fluctuation sector.

\begin{definition}
\label{def:partition_function_axial_gauge}
The partition function of $BF$ theory in the axial gauge $\mathbb{L}_V$ is defined as:
\begin{equation}
\label{eq:axial-partition-result}
    \mathcal{Z}'(\widetilde{\mathbb{S}}_{BF}, \mathbb{L}_V) \coloneqq e^{i \widetilde{\mathbb{S}}_{BF}'} \, \nu' \cdot \left| \det\nolimits'(\iota_V)\right|^{-1} \cdot \prod_{k=0}^{n-1} \left|\det\nolimits^{\flat}(\widetilde{\mathcal{L}}_{V, \nabla}^{(k)})^{(-1)^{k}}\right|,
\end{equation}
where $\nu' \in \mathrm{Dens}^{1/2}(\widetilde{\mathcal{F}}_{BF}')$ is a half-density on the residual sector and $\det^\flat$ denotes the flat-regularised determinant as per Definition \ref{def: 2-step-determinant}. Furthermore, as established in Proposition \ref{prop: action-splitting-axial}, the term $\widetilde{\mathbb{S}}_{BF}'$ denotes the effective action on the residual sector:
\begin{equation}
    \widetilde{\mathbb{S}}_{BF}' = \int_M \left[ B_0 \wedge_h d^\nabla A_0 \right]^{\text{top}}.
\end{equation}
\end{definition}

\begin{rem}
As in the metric gauge case (cf.\ Remark~\ref{rem:metric-partition-notation}), we
use the notation $\mathcal{Z}'(\widetilde{\mathbb{S}}_{BF}, \mathbb{L}_V)$ rather
than $\mathcal{Z}'(A)$ from Equation~\eqref{eq:residual_partition_function}, with
$A = \widetilde{\mathcal{L}}_{V,\nabla}$, while $A''$ denotes the operator restricted to
$\mathbb{L}_V$.
\end{rem}
It remains to address the contribution of the interior product operator to the partition function. We will provide an explicit characterisation for its regularised determinant in Subsection \ref{subsec:determinant_iota}, where it will be shown why $\det'(\iota_V)$ can be consistently set to unity.

It is now straightforward to link this object with the dynamical Ruelle zeta function.

\begin{cor}
\label{cor:axial-ruelle-link}
The axial partition function is related to the Ruelle zeta function at zero by the following identity:
\begin{equation}
    \mathcal{Z}'(\widetilde{\mathbb{S}}_{BF}, \mathbb{L}_V) = e^{i \widetilde{\mathbb{S}}_{BF}'} \, \nu' \cdot \left| \det\nolimits'(\iota_V)\right|^{-1} \cdot \mathfrak{R}_{V, \rho}(0)^{(-1)^n}.
\end{equation}
\end{cor}

\begin{proof}
The identity follows by substituting the spectral realisation of the Ruelle zeta function at the origin, established in Theorem \ref{thm: Ruelle-as-prod-det-app}, into Equation \eqref{eq:axial-partition-result}.
\end{proof}

Corollary \ref{cor:axial-ruelle-link} yields an effective partition function on the finite-dimensional space of zero-resonant states $\widetilde{\mathcal{F}}_{BF}'$. Unlike the $L^2$ framework of the metric gauge, which directly isolates harmonic forms, the anisotropic completion yields a residual space spanned by Pollicott--Ruelle resonant states. Consequently, while the total spaces $\widetilde{\mathcal{F}}_{BF}$ and $\mathcal{F}_{BF}^{L^2}$ are topologically isomorphic, their respective residual field spaces $\widetilde{\mathcal{F}}_{BF}' $ and $ \mathcal{F}_{BF}'$ are not.

Nevertheless, Theorem \ref{thm:spectral_realization_cohomology} ensures a canonical isomorphism at the cohomological level: $H^\bullet(C^\bullet_{V,\nabla}(0), d^\nabla) \cong H^\bullet(M, \mathcal{E})$. To restrict the theory to these purely cohomological degrees of freedom and integrate out the remaining resonant fluctuations (thereby handling the non-trivial effective action $\widetilde{\mathbb{S}}_{BF}'$), we perform a second BV pushforward. 

The following commutative diagram illustrates how the two distinct analytical completions, despite their different intermediate splittings, ultimately reduce to the same topological invariant:


\[
\xymatrix@C=1em{
\ar []+<-7ex,-2ex> ; [dd]+<-18ex,2ex> \mathcal{F}_{BF}'\oplus \mathcal{F}_{BF}'' \simeq \widetilde{\mathcal{F}}_{BF}'\oplus \widetilde{\mathcal{F}}_{BF}'' \ar[dr] & \\
 & \widetilde{\mathcal{F}}_{BF}' \simeq \widetilde{\mathcal{F}}_{BF}^{(1)} \oplus \widetilde{\mathcal{F}}_{BF}^{(2)}\ar[d]\\
\mathcal{F}_{BF}' = H^\bullet(M,\mathcal{E})[1]\oplus H^\bullet(M,\mathcal{E})[n-2] \ar[r] & \text{Ker}(\widetilde{\Delta}_\bullet)[1] \oplus \text{Ker}(\widetilde{\Delta}_\bullet)[n-2]= \widetilde{\mathcal{F}}^{(1)}_{BF},
}
\]

As the diagram anticipates, while the two analytical completions partition the fields differently, they ultimately isolate the same topological content. The space $\widetilde{\mathcal{F}}^{(1)}_{BF}$ appearing in the diagram represents the purely cohomological sector of the resonant complex, which is canonically isomorphic to the harmonic sector $\mathcal{F}_{BF}'$. The precise algebraic construction yielding this subspace and the execution of the associated second BV pushforward are detailed in what follows. Specifically, isolating this purely cohomological subspace relies on establishing a Hodge decomposition of $C^k_{V, \nabla}(0)$. We achieve this by equipping the finite-dimensional complex with an auxiliary metric $h_{aux}$, which naturally defines the algebraic adjoint $\tilde{d}^{\nabla, \dagger}_k \colon C^k_{V, \nabla}(0) \to C^{k-1}_{V, \nabla}(0)$ through the relation:
\begin{equation}
    h_{aux}(d^\nabla \alpha, \beta) = h_{aux}(\alpha, \tilde{d}^{\nabla, \dagger} \beta), \quad \forall \alpha, \beta \in C^\bullet_{V, \nabla}(0).
\end{equation}
By introducing $h_{aux}$ we can define a Laplacian $\widetilde{\Delta}_k \coloneqq \tilde{d}^{\nabla, \dagger}_k d^\nabla_k + d^\nabla_{k-1} \tilde{d}^{\nabla, \dagger}_{k-1}$ on $C^k_{V, \nabla}(0)$, which provides the following decomposition:
\begin{equation}
\label{eq:algebraic_hodge_axial}
    C^k_{V, \nabla}(0) = \text{Im}(d^\nabla_{k-1}) \oplus \text{Im}(\tilde{d}^{\nabla, \dagger}_{k}) \oplus \text{Ker}(\widetilde{\Delta}_k),
\end{equation}
where $\text{Ker}(\widetilde{\Delta}_k)$ is canonically isomorphic to the $k$-th cohomology group of the resonant complex.

A critical requirement for applying the BV formalism is to verify that this decomposition induces Lagrangian subspaces. While the exact summand $\text{Im}(d^\nabla)$ is isotropic due to the nilpotency of $d^\nabla$ and Stokes' Theorem, the isotropy of $\text{Im}(\tilde{d}^{\nabla, \dagger})$ is established by exploiting the duality between the auxiliary metric $h_{aux}$ and the symplectic form $\widetilde{\Omega}_{BF}$. Since the pairing induced by $\widetilde{\Omega}_{BF}$ between $C^k_{V, \nabla}(0)$ and $C^{n-k}_{V, \nabla}(0)$ is non-degenerate, {any} arbitrary choice of the metric $h_{aux}$ uniquely induces a linear isomorphism $\star_{alg}: C^k_{V, \nabla}(0) \to C^{n-k}_{V, \nabla}(0)$, defined by the compatibility relation:
\begin{equation}
\label{eq:algebraic_star_compatibility}
    \widetilde{\Omega}_{BF}((0, \beta), (\alpha, 0)) \coloneqq \int_M \beta \wedge_h \alpha = h_{aux}(\beta, \star_{alg} \alpha),
\end{equation}
for all $\alpha \in C^k_{V, \nabla}(0)$ and $\beta \in C^{n-k}_{V, \nabla}(0)$. Because of the defining relation \eqref{eq:algebraic_star_compatibility}, the Hodge decomposition in Equation \eqref{eq:algebraic_hodge_axial} is automatically orthogonal with respect to the symplectic structure. This guarantees the isotropy of the co-exact subspace without imposing any actual restriction on $h_{aux}$. Consequently, the resulting algebraic Hodge decomposition cleanly splits $\widetilde{\mathcal{F}}_{BF}'$ into a cohomological and a fluctuation sector, naturally yielding a suitable gauge-fixing Lagrangian for the secondary BV pushforward.

\begin{prop}
\label{prop:residual_splitting_lagrangian}
Let the residual field space be decomposed as $\widetilde{\mathcal{F}}_{BF}' \simeq \widetilde{\mathcal{F}}_{BF}^{(1)} \oplus \widetilde{\mathcal{F}}_{BF}^{(2)}$, where the summands are defined by:
\begin{equation}
\label{eq:F_BF_1_2}
\begin{split}
    \widetilde{\mathcal{F}}_{BF}^{(1)} &\coloneqq \text{Ker}(\widetilde{\Delta}_\bullet)[1] \oplus \text{Ker}(\widetilde{\Delta}_\bullet)[n-2], \\
    \widetilde{\mathcal{F}}_{BF}^{(2)} &\coloneqq \left(\text{Im}(d^\nabla_{\bullet-1}) \oplus \text{Im}(\tilde{d}_{\bullet}^{\nabla, \dagger})\right)[1] \oplus \left(\text{Im}(d^\nabla_{\bullet-1}) \oplus \text{Im}(\tilde{d}_{\bullet}^{\nabla, \dagger})\right)[n-2].
\end{split}
\end{equation}
The subspace $\mathbb{L}^{(2)} \subset \widetilde{\mathcal{F}}_{BF}^{(2)}$ defined as:
\begin{equation}
\label{eq:lagrangian_2_sub}
    \mathbb{L}^{(2)} \coloneqq \text{Im}(\tilde{d}_{\bullet}^{\nabla, \dagger})[1] \oplus \text{Im}(\tilde{d}_{\bullet}^{\nabla, \dagger})[n-2],
\end{equation}
is a Lagrangian subspace of $\widetilde{\mathcal{F}}_{BF}^{(2)}$ with respect to the restricted symplectic structure. 
\end{prop}

With this gauge-fixing Lagrangian at our disposal, we are now in a position to establish the following theorem regarding the evaluation of the secondary BV pushforward.

\begin{theorem}
\label{thm:main_pushforward_axial}
Let $\mathcal{Z}'(\widetilde{\mathbb{S}}_{BF}, \mathbb{L}_V)$ be the axial partition function of Corollary \ref{cor:axial-ruelle-link}. Given an arbitrary auxiliary metric $h_{aux}$ on $C_{V, \nabla}(0)$, the BV pushforward onto $\widetilde{\mathcal{F}}_{BF}^{(1)}$, relative to $\mathbb{L}^{(2)} \subset \widetilde{\mathcal{F}}_{BF}^{(2)}$ defined in Equation \eqref{eq:lagrangian_2_sub}, equals to:
\begin{equation}
\label{eq:final-partition-function-axial}
    \mathcal{Z}^{(1)}(\mathbb{S}_{BF}',\mathbb{L}^{(2)}) \coloneqq \int_{\mathbb{L}^{(2)}} \mathcal{Z}'(\widetilde{\mathbb{S}}_{BF}, \mathbb{L}_V) = \mathfrak{R}_{V, \rho}(0)^{(-1)^n} \cdot \prod_{k=0}^{n} ({\det} ' \widetilde{\Delta}_k)^{(-1)^{k+1} k/2} \cdot \mu^{(1)},
\end{equation}
where $\mu^{(1)} \in \mathrm{Dens}^{1/2}(\widetilde{\mathcal{F}}_{BF}^{(1)})$ is the residual half-density.
\end{theorem}

\begin{proof}
By the direct sum decomposition in Equation \eqref{eq:F_BF_1_2}, the half-density $\nu'$ factorises as $\mu^{(1)} \otimes \mu^{(2)} \in \mathrm{Dens}^{1/2}(\widetilde{\mathcal{F}}_{BF}^{(1)}) \otimes \mathrm{Dens}^{1/2}(\widetilde{\mathcal{F}}_{BF}^{(2)})$. We perform an integration over $\mathbb{L}^{(2)}$ for the term $(e^{i \widetilde{\mathbb{S}}_{BF}'} \cdot \nu')$. 

Following what has been done in Remark \ref{rem: BF-action-restricted-metric}, the action functional restricted to $\mathbb{L}^{(2)}$ depends only on non-harmonic components and takes the quadratic form:
\begin{equation}
    \widetilde{\mathbb{S}}_{BF}' \big|_{\mathbb{L}^{(2)}} = \sum_{k=0}^{n-1} \left(\tau^{(k)}_0, \tilde{d}_k^{\nabla, \dagger} d_k^\nabla \mathcal{A}^{(k)}_0 \right)_{h_{aux}},
\end{equation}
where $\tau^{(k)}_0$ parametrises the co-exact sector of $\mathcal{B}_0 \in \text{Im}(\tilde{d}_{n-k-1}^{\nabla, \dagger})[n-2]$ while $\mathcal{A}^{(k)}_0 \in \text{Im}(\tilde{d}_k^{\nabla, \dagger})[1]$. The integration involves two factors: the Jacobian $|\det'(\tilde{d}^{\nabla, \dagger} \tilde{d}^\nabla)|^{-1/2}$ arising from the change of variables to $\tau_0$, and the result of the Gaussian integral over the shifted space, which yields the determinant of the operator directly in the numerator (see Definition \ref{def:partition_function_flat}):
\begin{align*}
    \int_{\mathbb{L}^{(2)}} e^{i \widetilde{\mathbb{S}}_{BF}'} \mu^{(2)} &= \left| \sdet\nolimits'(\tilde{d}^{\nabla, \dagger} d^\nabla) \right|^{-1/2} \cdot \left| \sdet\nolimits'(\tilde{d}^{\nabla, \dagger} d^\nabla) \right| \cdot \mu^{(1)}\\
    &= \left| \sdet\nolimits'(\tilde{d}^{\nabla, \dagger} d^\nabla) \right|^{1/2} \cdot \mu^{(1)}.
\end{align*}
By explicitly writing the superdeterminant as $\prod_{k=1}^n (\det'(\tilde{d}_{k-1}^{\nabla, \dagger} d_{k-1}^\nabla))^{(-1)^{k+1}/2}$ and applying the same re-indexing used in Proposition \ref{prop:torsion-rearrangement-flat}, we recover the alternating product of the Laplacian determinants:
\begin{equation*}
    \int_{\mathbb{L}^{(2)}} e^{i \widetilde{\mathbb{S}}_{BF}'} \nu' = \prod_{k=0}^n (\det' \nolimits \widetilde{\Delta}_k)^{(-1)^{k+1} k/2} \cdot \mu^{(1)}.
\end{equation*}

Combining this result with the prefactors from the first pushforward, we recover the partition function up to the Jacobian factor $\left| \det\nolimits'(\iota_V)\right|^{-1}$. As we will rigorously justify in Subsection \ref{subsec:determinant_iota}, a consistent normalisation of the vector field allows us to evaluate $\det'(\iota_V) = 1$, yielding the sought expression.

\end{proof}

\begin{rem}
\label{rem:notation_axial_partition_function}
For the sake of readability and to emphasise the fact that both stages of the BV-pushforward, encoded respectively by the choice of the axial Lagrangian $\mathbb{L}_V$ and by the gauge-fixing Lagrangian $\mathbb{L}^{(2)}$, ultimately depend on the choice of a Morse--Smale vector field, we adopt the lighter notation with superscript $V$:
\begin{equation}
\label{eq:zeta_bf_rename}
    \mathfrak{Z}^V_{BF} \doteq \mathcal{Z}^{(1)}(\mathbb{S}_{BF}', \mathbb{L}^{(2)}) \doteq \int_{\mathbb{L}^{(2)}} \mathcal{Z}'(\widetilde{\mathbb{S}}_{BF}, \mathbb{L}_V)
\end{equation}
for the partition function obtained in Theorem \ref{thm:main_pushforward_axial}.
\end{rem}

\begin{rem}
\label{rem:cohomological_reduction}
Theorem \ref{thm:main_pushforward_axial} completes the field-theoretic reduction to the cohomological sector. The residual half-density $\mu^{(1)} \in \mathrm{Dens}^{1/2}(\widetilde{\mathcal{F}}_{BF}^{(1)})$ is canonically identified with an element of the determinant line $\text{Det}(H^\bullet(M, \mathcal{E}))$, an identification supported by the sequence of isomorphisms $\text{Ker}(\widetilde{\Delta}_\bullet) \cong H^\bullet(C^\bullet_{V, \nabla}(0)) \cong H^\bullet(M, \mathcal{E})$ provided by the algebraic Hodge decomposition and Theorem \ref{thm:spectral_realization_cohomology}. 
\end{rem}

\begin{rem}
\label{rem:thom_smale_gauge_choice}
As anticipated, we aim to establish a direct link between the numerical factor arising from the second integration and the Thom--Smale torsion.
Up to this point, we used a completely arbitrary auxiliary metric $h_{aux}$ to define the gauge-fixing Lagrangian $\mathbb{L}^{(2)}$. However, we now evaluate the partition function by explicitly choosing $h_{aux}$ to be the inner product that promotes the canonical isomorphism $\Phi$, introduced in Definition \ref{def:complex_isomorphism}, to an isometry. Consequently, by applying Theorem \ref{thm:equality-of-torsion}, the alternating product of the determinants of $\widetilde{\Delta}_\bullet$ evaluated in Theorem \ref{thm:main_pushforward_axial} coincides with the combinatorial torsion of the Thom--Smale complex. This provides a field-theoretic derivation of the fixed-point contribution.
\end{rem}

In order to identify the physical object computed by the axial partition function, we recall the definition of the Milnor metric, which naturally decouples the dynamics into two independent sectors, as established in \cite[Section 2]{Shen21}. For an exhaustive analysis of the factorisation of the Milnor metric into its independent fixed-point and closed-orbit components, we refer the reader to the review in \cite{Mol26}.

\begin{definition}
\label{def:milnor_metric}
The {Milnor metric} $\| \cdot \|_{M, V}$ on the determinant line $\text{Det}(H^\bullet(M, \mathcal{E}))$ is defined by its evaluation on a unitary volume element $\mu_H$ as:
\begin{equation}
    \| \mu_H \|_{M, V} = |\tau(C^\bullet_{\mathrm{TS}})| \cdot |\mathfrak{R}_{V, \rho}(0)|^{-1},
\end{equation}
where $\tau(C^\bullet_{\mathrm{TS}})$ is the torsion of the Thom--Smale complex.
\end{definition}

Comparing this definition with the explicit evaluation of the BV pushforward yields our central result.

\begin{cor}
\label{cor:partition_function_milnor}
Under the canonical identification of the residual half-density $\mu^{(1)}$ of Theorem \ref{thm:main_pushforward_axial} with the unitary volume element $\mu_H$, the absolute value of the axial partition function exactly recovers the Milnor metric:
\begin{equation}
    | \mathfrak{Z}^V_{BF} | = \| \mu_H \|_{M, V}.
\end{equation}
\end{cor}

\begin{proof}
By Theorem \ref{thm:main_pushforward_axial}, the axial partition function $\mathfrak{Z}^V_{BF}$ factors into the dynamical contribution of the Ruelle zeta function and the alternating product of the regularised determinants of $\widetilde{\Delta}_k$. As detailed in Remark \ref{rem:thom_smale_gauge_choice}, by selecting the auxiliary metric $h_{aux}$ that renders the isomorphism $\Phi$ an isometry, this spectral factor coincides exactly with the combinatorial torsion of the Thom--Smale complex, $\tau(C^\bullet_{\mathrm{TS}})$. Therefore, the partition function reduces to the product of $|\tau(C^\bullet_{\mathrm{TS}})|$ and $\mathfrak{R}_{V, \rho}(0)$. By Definition \ref{def:milnor_metric}, this combination precisely yields the Milnor metric $\| \mu_H \|_{M, V}$.
\end{proof}

\begin{rem}
\label{rem:combined_lagrangian}
The two-stage construction of $\mathfrak{Z}^V_{BF}$ as per Equation \eqref{eq:zeta_bf_rename} admits an equivalent description as a single BV integration over the lagrangian subspace $\mathbb{L}_V \oplus \mathbb{L}^{(2)}$. Indeed, $\mathbb{L}_V \subset \widetilde{\mathcal{F}}_{BF}''$ and $\mathbb{L}^{(2)} \subset \widetilde{\mathcal{F}}_{BF}^{(2)}$ are Lagrangian subspaces with respect to the symplectic forms induced on $\widetilde{\mathcal{F}}_{BF}''$ and $\widetilde{\mathcal{F}}_{BF}^{(2)}$ respectively, by Proposition \ref{prop: axial_gauge_lagrangian} and Proposition \ref{prop:residual_splitting_lagrangian}. Since $\widetilde{\mathcal{F}}_{BF}''$ and $\widetilde{\mathcal{F}}_{BF}^{(2)}$ are themselves symplectic subspaces within $\widetilde{\mathcal{F}}_{BF}$, via the decompositions provided in Equations \eqref{eq:bv_splitting_Sob} and \eqref{eq:F_BF_1_2}, it follows that $\mathbb{L}_V \oplus \mathbb{L}^{(2)}$ is itself a Lagrangian subspace of $\widetilde{\mathcal{F}}_{BF}$.
\end{rem}

\begin{rem}[Fried's Conjecture]
    To summarise, the evaluation of the partition function for Abelian $BF$ theory has been performed under two distinct gauge-fixing conditions. As proven in Section \ref{sec: metric-gauge}, the metric gauge-fixing leads to the Ray--Singer metric. By contrast, the axial gauge-fixing investigated here recovers the Milnor metric. Given that the equality between the Ray--Singer and Milnor metrics is established by Fried's conjecture, which in the specific context of Morse--Smale flows constitutes a proven theorem, our findings suggest that this correspondence can be reinterpreted as the requirement for the partition function of the theory to be independent of the choice of gauge-fixing condition. 
\end{rem}

\subsection{The Determinant of the Interior Product}
\label{subsec:determinant_iota}

In the evaluation of the axial partition function, we have set the Jacobian determinant $\det'(\iota_V)$ to $1$. We now provide a complete justification for this choice. Let $g$ be a Riemannian metric on the manifold $M$, and denote by $\|V\|_g(p) \coloneqq \sqrt{g_p(V,V)}$ the point-wise norm of the Morse--Smale vector field $V$. By Definition \ref{def: MS-flow}, a Morse--Smale vector field $V$ vanishes at a finite number of hyperbolic fixed points.

We recall the quadratic form of the action restricted to the axial Lagrangian $\mathbb{L}_V$ as derived in Proposition \ref{prop: action-restricted-axial}:
\begin{equation}
\label{eq: axial-action-app}
    \widetilde{\mathbb{S}}_{BF}'\big|_{\mathbb{L}_V} = \sum_{k=0}^{n-1} \int_M \left(\mathcal{B}^{(n-k)}\wedge_h \widetilde{\mathcal{L}}_{V, \nabla}^{(k)} \eta^{(k)} \right),
\end{equation}
where $\widetilde{\mathcal{L}}_{V, \nabla}^{(k)}$ is the Lie derivative restricted to $\text{Im}\big(\iota_V\big|_{\text{Ker}\, \pi_0^{(k+1)}}\big)$, and the forms $\eta^{(k)}$ denote the $\iota_V$-primitives of the field $\mathcal{A}$, \textit{i.e.}, $\mathcal{A}^{(k)}=\iota_V\eta^{(k)}$.

We introduce the normalised vector field $\widehat{V} \coloneqq V / \|V\|_g$ and, correspondingly, we define the normalised restricted Lie derivative:
\begin{equation}
\label{eq: rescaling-definitions}
    \mathbf{L}_{V, \nabla}^{(k)} \coloneqq \frac{1}{\|V\|_g} \widetilde{\mathcal{L}}_{V, \nabla}^{(k)}.
\end{equation}
It is crucial to observe that $\mathbf{L}_{V, \nabla}^{(k)}$ remains well-defined even at the fixed points of the flow: the zeros of the vector field $V$ and its norm $\|V\|_g$ are of the same order, resulting in a removable singularity. Moreover, on the restricted subspace, $\mathbf{L}_{V, \nabla}^{(k)}$ precisely coincides with the Lie derivative along the normalised field $\widehat{V}$, \textit{i.e.}, $\mathbf{L}_{V, \nabla}^{(k)} \equiv \mathcal{L}_{\widehat{V}, \nabla}^{(k)}$.

By substituting this normalised operator and redefining the integration variables as $\widetilde{\eta}^{(k)} \coloneqq \|V\|_g \eta^{(k)}$, the action functional can be identically rewritten as:
\begin{equation}
\label{eq: axial-action-rescaled}
    \widetilde{\mathbb{S}}_{BF}'\big|_{\mathbb{L}_V} = \sum_{k=0}^{n-1} \int_M \left(\mathcal{B}^{(n-k)}\wedge_h \mathcal{L}_{\widehat{V}, \nabla}^{(k)} \widetilde{\eta}^{(k)} \right).
\end{equation}

This algebraic manipulation highlights that evaluating the theory on the axial gauge $\mathbb{L}_V$ generated by $V$ is completely equivalent to evaluating it on the gauge generated by the normalised field $\widehat{V}$. Since the formal expression of the partition function must remain invariant under this reparametrisation, it can be equivalently evaluated by replacing the interior product $\iota_V$ with $\iota_{\widehat{V}}$ and the restricted Lie derivative $\widetilde{\mathcal{L}}_{V, \nabla}^{(k)}$ with $\mathcal{L}_{\widehat{V}, \nabla}^{(k)}$.

We begin by examining how the rescaling of a vector field affects its integral curves, their associated travel times, and ultimately the regularised determinants.

\begin{lemd}
\label{lem:reparametrization-general}
Let $X$ be a smooth vector field on a manifold $M$ and let $f \in C^\infty(M, \mathbb{R}^+)$ be a strictly positive smooth function. Let $Y = fX$ be the rescaled vector field. Then $X$ and $Y$ possess the same integral curves, differing only by a time-reparametrisation.
\end{lemd}

\begin{proof}
Let $\gamma(t)$ be an integral curve of $X$ such that $\gamma(0) = p$, satisfying $\frac{d\gamma}{dt} = X_{\gamma(t)}$. We seek an integral curve $\eta(\tau)$ for $Y$ starting at $p$ which can be expressed as $\eta(\tau) = \gamma(h(\tau))$ for some monotonic and invertible function $h$. By the chain rule:
\begin{equation*}
    \frac{d \eta(\tau)}{d\tau} = \frac{d\gamma}{dt}(h(\tau)) \cdot h'(\tau) = X_{\gamma(h(\tau))} \cdot h'(\tau) = X_{\eta(\tau)} \cdot h'(\tau).
\end{equation*}
For $\eta(\tau)$ to be an integral curve of $Y$, we must also satisfy the defining equation:
\begin{equation*}
    \frac{d \eta(\tau)}{d\tau} = Y_{\eta(\tau)} = f(\eta(\tau)) \cdot X_{\eta(\tau)}.
\end{equation*}
Comparing the two expressions, we obtain the differential equation governing the time reparametrisation:
\begin{equation*}
    h'(\tau) = f(\gamma(h(\tau))) \implies \frac{dt}{d\tau} = f(\gamma(t)).
\end{equation*}
Since $f$ is strictly positive, the reparametrisation is globally solvable, ensuring that $X$ and $Y$ share the exact same integral curves.
\end{proof}

We apply Lemma \ref{lem:reparametrization-general} to the Morse--Smale vector field $V$ and its normalised counterpart $\widehat{V} = V / \|V\|_g$. Under this reparametrisation, any closed orbit $\Lambda$ of $V$ is preserved as a periodic trajectory for $\widehat{V}$. While its minimal period is modified to a new value denoted by $\widehat{T}_\Lambda$, the holonomy representation $\rho(\Lambda)$ and the linearised Poincaré map $\mathcal{P}_\Lambda$ are geometric invariants of the orbit and thus remain strictly invariant. 

\begin{lemd}
\label{lem: equality-lie-dets}
Let $\widetilde{\mathcal{L}}_{V, \nabla}^{(k)}$ be the restricted Lie derivative along the Morse--Smale vector field $V$, and $\mathbf{L}_{V, \nabla}^{(k)}$ be the normalised restricted Lie derivative. Their flat-regularised determinants coincide:
\begin{equation}
    \det\nolimits^\flat \big(\widetilde{\mathcal{L}}_{V, \nabla}^{(k)}\big) = \det\nolimits^\flat \big(\mathbf{L}_{V, \nabla}^{(k)}\big).
\end{equation}
\end{lemd}

\begin{proof}
To isolate the invariant data of the orbits from their period dependence, we introduce the following weight factor:
\begin{equation}
\label{eq: weight-definition}
    W_{\Lambda, j}^{(k)} \coloneqq (-1)^{n+1} \frac{\operatorname{tr}(\wedge^k\mathcal{P}_\Lambda^j) \operatorname{tr}(\rho(\Lambda)^j)}{|\det(\mathbb{I} - \mathcal{P}_\Lambda^j)|}.
\end{equation}
By means of Proposition \ref{prop:k-trace-formula} and the weight just defined, the flat trace of the original restricted Lie derivative can be rewritten as:
\begin{equation}
\label{eq: trace-V-compact}
    \operatorname{tr}^{\flat}\left(e^{-t\widetilde{\mathcal{L}}_{V, \nabla}^{(k)}}\right) = \sum_{\Lambda, j \ge 1} T_\Lambda W_{\Lambda, j}^{(k)} \delta(t - jT_\Lambda).
\end{equation}
Analogously, the flat trace for the normalised operator $\mathbf{L}_{V, \nabla}^{(k)}$ admits a structurally identical expression, where the original periods are replaced by their unit-speed counterparts:
\begin{equation}
\label{eq: trace-Vhat-compact}
    \operatorname{tr}^{\flat}\left(e^{-t\mathbf{L}_{V, \nabla}^{(k)}}\right) = \sum_{\Lambda, j \ge 1} \widehat{T}_\Lambda W_{\Lambda, j}^{(k)} \delta(t - j\widehat{T}_\Lambda).
\end{equation}

We compute the regularised determinant of $\mathbf{L}_{V, \nabla}^{(k)}$ following the paradigm introduced in Subsection \ref{sec:flat_traces}. Applying the Laplace transform to the trace formula yields the auxiliary function:
\begin{align*}
\mathcal{F}_{\mathbf{L}}(z, s) &\coloneqq \frac{1}{\Gamma(s)}\int_0^{\infty}e^{-tz}t^{s-1}\operatorname{tr}^{\flat}\left(e^{-t \mathbf{L}_{V,\nabla}^{(k)}}\right)dt \\
&= \frac{1}{\Gamma(s)} \int_0^{\infty} e^{-tz} t^{s-1} \left[ \sum_{\Lambda, j \ge 1} \widehat{T}_\Lambda W_{\Lambda, j}^{(k)} \delta(t - j\widehat{T}_\Lambda) \right] dt.
\end{align*}
Applying the Dirac distribution, we obtain:
\begin{align*}
\mathcal{F}_{\mathbf{L}}(z, s) &= \frac{1}{\Gamma(s)} \sum_{\Lambda, j \ge 1} W_{\Lambda, j}^{(k)} \widehat{T}_\Lambda (j \widehat{T}_\Lambda)^{s-1} e^{-z j \widehat{T}_\Lambda} \\
&= \frac{1}{\Gamma(s)} \sum_{\Lambda, j \ge 1} W_{\Lambda, j}^{(k)} (\widehat{T}_\Lambda)^s j^{s-1} e^{-z j \widehat{T}_\Lambda}.
\end{align*}

By definition, the logarithm of the determinant is obtained by taking the negative derivative with respect to $s$ evaluated at $s=0$. Using the Leibniz rule and the expansion $1/\Gamma(s) = s + \mathcal{O}(s^2)$ around the origin, we find:
\begin{align*}
\log\left(\det\nolimits^{\flat}(\mathbf{L}_{V, \nabla}^{(k)} + z)\right) &\coloneqq - \partial_s\vert_{s=0}\mathcal{F}_{\mathbf{L}}(z, s) \\
&= - \left[ \left( \partial_s \frac{1}{\Gamma(s)} \right)\Big|_{s=0} \cdot \sum_{\Lambda, j \ge 1} W_{\Lambda, j}^{(k)} (\widehat{T}_\Lambda)^0 j^{-1} e^{-z j \widehat{T}_\Lambda} \right] \\
&= - \sum_{\Lambda, j \ge 1} \frac{W_{\Lambda, j}^{(k)}}{j} e^{-z j \widehat{T}_\Lambda}.
\end{align*}

Following an identical derivation for the original operator $\widetilde{\mathcal{L}}_{V, \nabla}^{(k)}$, we obtain:
\begin{equation*}
\log\left(\det\nolimits^{\flat}(\widetilde{\mathcal{L}}_{V, \nabla}^{(k)} + z)\right) = - \sum_{\Lambda, j \ge 1} \frac{W_{\Lambda, j}^{(k)}}{j} e^{-z j T_\Lambda}.
\end{equation*}

Evaluating both expressions at $z=0$, the exponential damping factors evaluate to $1$. The dependence on the respective minimal periods vanishes completely, yielding:
\begin{equation}
\label{eq: equality-lie-dets}
\log\left(\det\nolimits^{\flat}(\widetilde{\mathcal{L}}_{V, \nabla}^{(k)})\right) = - \sum_{\Lambda, j \ge 1} \frac{W_{\Lambda, j}^{(k)}}{j} = \log\left(\det\nolimits^{\flat}(\mathbf{L}_{V, \nabla}^{(k)})\right).
\end{equation}
Exponentiating this identity concludes the proof.
\end{proof}

Given the invariance of the regularised Lie derivative determinants established in Lemma \ref{lem: equality-lie-dets}, the comparison between the two expressions of $\mathcal{Z}'$ obtained from the use of $V$ (as per Definition \ref{def:partition_function_axial_gauge}) and its normalised counterpart $\widehat{V}$, requires that:
\begin{equation}
\label{eq: equality-iota-dets}
\left| \det\nolimits'(\iota_V)\right| = \left| \det\nolimits'(\iota_{\widehat{V}})\right|.
\end{equation}
This crucial identity reduces the problem of evaluating $\det'(\iota_V)$ to the simpler task of characterising the determinant of the interior product operator along the normalised field $\widehat{V}$. We can define this determinant in a way that is formally analogous to the regularised determinant of the de Rham differential.

\begin{definition}
\label{def: determinant-iota-def}
Let $\flat: \mathfrak{X}(M) \to \Omega^1(M)$ be the musical isomorphism induced by the metric $g$, and let $\widehat{V}^\flat \coloneqq V^\flat / \|V\|_g$ be the normalised $1$-form associated with $V$. The regularised determinant of the interior product operator $\iota_{\widehat{V}}$ is defined via its composition with its formal adjoint, the exterior multiplication operator $\widehat{V}^\flat \wedge: \mathcal{H}_m^k(M, \mathcal{E}) \to \mathcal{H}_m^{k+1}(M, \mathcal{E})$:
\begin{equation}
    \det\nolimits'(\iota_{\widehat{V}}) \coloneqq \left| \det\nolimits'(\iota_{\widehat{V}} \circ (\widehat{V}^\flat \wedge)) \right|^{1/2}.
\end{equation}
\end{definition}

\begin{prop}
\label{prop:det_iota_is_one}
The determinant of the interior product operator along the normalised vector field $\widehat{V}$ is:
\begin{equation}
    \det\nolimits'(\iota_{\widehat{V}}) = 1.
\end{equation}
\end{prop}

\begin{proof}
We recall that the following fundamental identity holds:
\begin{equation*}
    \iota_{\widehat{V}}(\widehat{V}^\flat \wedge \omega) + \widehat{V}^\flat \wedge (\iota_{\widehat{V}} \omega) = g(\widehat{V}, \widehat{V}) \omega, \quad \forall \omega \in \mathcal{H}_m^\bullet(M, \mathcal{E}).
\end{equation*}
In the context of the axial gauge-fixing, the integration is performed over $\mathbb{L}_V$, whose elements satisfy the condition $\iota_{\widehat{V}} \omega = 0$. Therefore, the second term of the left hand side vanishes identically. Furthermore, since the normalised vector field $\widehat{V}$ satisfies $g(\widehat{V}, \widehat{V}) = 1$ everywhere on $M$, the identity reduces to:
\begin{equation}
    \iota_{\widehat{V}} \circ (\widehat{V}^\flat \wedge) \Big|_{\text{Im}(\iota_{\widehat{V}})} = \mathbb{I}.
\end{equation}
The operator composition defining the determinant therefore acts as the identity on the relevant subspace. It follows immediately that $\det'(\iota_{\widehat{V}}) = 1$.

\end{proof}

Due to the invariance established in Equation \eqref{eq: equality-iota-dets}, Proposition \ref{prop:det_iota_is_one} implies that the original Jacobian factor must also be unity. This provides the final and complete justification for setting $\det'(\iota_V) = 1$ in the expression presented in Theorem \ref{thm:main_pushforward_axial}.

\appendix 
\section{Morse--Smale Vector Fields}
\label{app:Morse--Smale}

Since the definition of Morse--Smale flows and their properties are central to our analysis, we collect here the fundamental definitions and structural results concerning this kind of vector fields. The purpose of this appendix is to provide a self-contained reference for the geometric and dynamical assumptions employed throughout the paper. For a more detailed and comprehensive treatment, we refer the reader to \cite{DangRiviere17Topology, Shen21Survey} and \cite[Appendix A]{DangRiviere17Anisotropic}.

The way to characterise this kind of dynamical system is through their recurrent behaviour, which is encapsulated by the {non-wandering set}. Heuristically, this set consists of points whose trajectories return to, or pass arbitrarily close to, after an arbitrarily long time.

\begin{definition}
\label{nw-appendix}
The non-wandering set of a vector field $V \in \mathfrak{X}(M)$, denoted by $\mathbf{NW}(V)$, is the set of points $p \in M$ such that for any neighbourhood $U$ of $p$ and any $T > 0$, there exists a time $t \ge T$ for which $U \cap \phi^t(U) \neq \emptyset$.
\end{definition}

For Morse--Smale systems, the recurrent set consists solely of fixed points and periodic orbits, each characterised by a robust form of stability known as hyperbolicity.

\begin{definition}
\label{hyperbolic-point-appendix}
A fixed point $p \in M$ (where $V_p = 0$) is {hyperbolic} if the eigenvalues of the linearised flow $d_p\phi^t : T_pM \to T_pM$ have non unitary module for all $t > 0$.
\end{definition}

For a hyperbolic fixed point $p$, the tangent space admits a $d\phi^t$-invariant splitting $T_pM = T^s_p \oplus T^u_p$ into stable and unstable subspaces. This decomposition is characterised by exponential estimates: there exist constants $C > 0$ and $\theta > 0$ such that for all $t \ge 0$:
\begin{equation}
\label{eq:exp-dilation-contraction}
    \left\|d_p\phi^{t}(v_s)\right\| \le C e^{-\theta t} \left\|v_s\right\|, \quad \left\|d_p\phi^{-t}(v_u)\right\| \le C e^{-\theta t} \left\|v_u\right\|, 
\end{equation}
for any $v_s \in T^s_p$ and $v_u \in T^u_p$. These linear subspaces are tangent to the \textit{stable} and \textit{unstable manifolds}, which consist of points converging to $p$ in the future and past, respectively:
\begin{equation}
\label{eq:unstable-point-app}
W^s(p) \coloneqq \left\{ q \in M : \lim_{t\to+\infty} \phi^t(q) = p \right\}, \quad W^u(p) \coloneqq \left\{ q \in M : \lim_{t\to-\infty} \phi^t(q) = p \right\}.
\end{equation}
A classical result in dynamical systems \cite{Perko2001} ensures that these sets are indeed embedded submanifolds with $T_pW^{s/u}(p) = T^{s/u}_p$. Similar considerations apply to periodic trajectories.

\begin{definition}
\label{hyperbolic-orbit-appendix}
A closed orbit $\Lambda$ with minimal period $T_\Lambda > 0$ is \textit{hyperbolic} if for any $p \in \Lambda$, the differential $d_p\phi^{T_\Lambda}$ has $1$ as a simple eigenvalue (corresponding to the direction of the flow $V_p$), while the remaining $n-1$ eigenvalues have modulus different from $1$.
\end{definition}

Hyperbolicity for an orbit $\Lambda$ implies a continuous, flow-invariant splitting of the tangent bundle restricted to $\Lambda$:
\begin{equation}
\label{eq:tangent-bundle-decomposition-orbits}
    TM|_\Lambda = \mathbb{R}V \oplus T^s_\Lambda \oplus T^u_\Lambda,
\end{equation}
where $T^s_\Lambda$ and $T^u_\Lambda$ denote the contracting and expanding directions, similarly to what happen in Equation \eqref{eq:exp-dilation-contraction}. This geometry leads to the definition of the stable and unstable manifolds for the orbit:
\begin{equation}
\label{eq:invariant-manifolds-orbit}
\begin{split}
    W^s(\Lambda) &\coloneqq  \left\{ q \in M : \lim_{t\to+\infty} d_M(\phi^t(q), \Lambda) = 0 \right\}, \\
    W^u(\Lambda) &\coloneqq  \left\{ q \in M : \lim_{t\to-\infty} d_M(\phi^t(q), \Lambda) = 0 \right\}.
\end{split}
\end{equation}
The {Stable Manifold Theorem} for orbits \cite{KatokHasselblatt95} guarantees that these are submanifolds with tangent bundles $\mathbb{R}V \oplus T^s_\Lambda$ and $\mathbb{R}V \oplus T^u_\Lambda$, respectively.

The global characterisation of the flow is determined by the way these invariant manifolds interact. We recall the condition of transversality, which ensures the robustness of the system.

\begin{definition}
\label{transversality-app}
Two submanifolds $L_1, L_2 \subset M$ intersect \textit{transversally} if at every $p \in L_1 \cap L_2$, their tangent spaces satisfy $T_p M = T_p L_1 + T_p L_2$. Non-intersecting manifolds are considered transversal by vacuity.
\end{definition}

We are now in a position to give the definition of Morse--Smale vector fields.

\begin{definition}
\label{def: MS-flow}
A flow generated by $V \in \mathfrak{X}(M)$ is a \textbf{Morse--Smale flow} if:
\begin{enumerate}
    \item The non-wandering set $\mathbf{NW}(V)$ consists of a finite number of disjoint hyperbolic critical elements (fixed points and closed orbits) $\{\Lambda_i\}_{i=1}^K$.
    \item For every pair of indices $(i, j)$, the unstable manifold $W^u(\Lambda_i)$ and the stable manifold $W^s(\Lambda_j)$ intersect transversally.
\end{enumerate}
\end{definition}

A fundamental property of these systems is the decomposition of the manifold into invariant sets.
\begin{prop}
\label{prop: partition-app}
Let $V$ be a Morse--Smale vector field. The manifold $M$ admits a disjoint partition in terms of stable (or unstable) manifolds:
\begin{equation}
    M = \bigsqcup_{k=1}^K W^u(\Lambda_k) = \bigsqcup_{k=1}^K W^s(\Lambda_k).
\end{equation}
Consequently, for each $p \in M$, there is a unique pair of elements $\Lambda_i, \Lambda_j$ such that the trajectory through $p$ connects $\Lambda_i$ to $\Lambda_j$.
\end{prop}

In our study, it will be necessary to assume that the flow can be locally linearised via a change of coordinates.

\begin{definition}
\label{MS-cl_linearisable-app}
A Morse--Smale flow is \textbf{$C^l$-linearisable} if for each critical element $\Lambda_i$, there exists a $C^l$-diffeomorphism $h$ such that:
\begin{itemize}
    \item If $\Lambda_i$ is a fixed point, $V \circ h = dh \circ L$, where $L(x) = (A_i x) \cdot \partial_x$ for some $A_i \in \text{GL}(\mathbb{R}^n)$.
    \item If $\Lambda_i$ is a closed orbit, $V \circ h = dh \circ L_{per}$, where $L_{per}(x, \theta) = (\mathcal{A}_i(\theta)x) \cdot \partial_x + \partial_\theta$ and $\mathcal{A}_i(\theta)$ is a periodic matrix-valued map.
\end{itemize}
The flow is said to be \textbf{$C^\infty$-linearisable} if it is $C^l$-linearisable for every $l \ge 1$.
\end{definition}

Since it will be fundamental in the following, we conclude by highlighting the constraint on connections between fixed points.
\begin{prop}
\label{prop: dimension-app}
For any two fixed points $\Lambda_i, \Lambda_j$ of a Morse--Smale flow, if $W^u(\Lambda_i) \cap W^s(\Lambda_j) \neq \emptyset$, then the following inequality must hold:
\begin{equation*}
\text{dim}(W^{u}(\Lambda_{i})) > \text{dim}(W^{u}(\Lambda_{j})).  
\end{equation*}
Furthermore, for any $x \in W^u(\Lambda_i) \cap W^s(\Lambda_j)$, the transversality condition implies:
\begin{equation}
\label{eq: dim-identity}
\text{dim}(T_{x}W^{u}(\Lambda_{i}))+\text{dim}(T_{x}W^{s}(\Lambda_{j})) = \text{dim}(M) +\text{dim}(T_{x}W^{u}(\Lambda_{i})\cap T_{x}W^{s}(\Lambda_{j})).
\end{equation}
\end{prop}

\begin{rem}
The inequality in Proposition \ref{prop: dimension-app} defines a partial ordering of the critical elements. This cascading nature of trajectories are essential for the definition of the Thom--Smale complex and the associated boundary operators.
\end{rem}
\section{Anisotropic Sobolev Spaces}
\label{app:Sobolev}

We provide the reader with a  review of the functional framework required for the spectral analysis of Morse--Smale flows. The central challenge is that the twisted Lie derivative along the flow, when acting on smooth differential forms, does not possess a discrete spectrum. To resolve this, one must construct an enlarged space of currents where the Lie derivative acts as an operator with a discrete spectrum of eigenvalues, known as Pollicott--Ruelle resonances.

\subsection{Anisotropic Sobolev Spaces}

The modern approach to rigorously define these spaces relies on microlocal analysis. This framework was pioneered by Faure and Sjöstrand \cite{FaureSjostrand11} for Anosov flows, subsequently generalised by Dyatlov and Zworski \cite{Dyatlov16}, and specifically adapted to the dynamics of Morse--Smale flows by Dang and Rivière in \cite{DangRiviere17Anisotropic, DangRiviere17Topology}. 

Rather than reproducing the heavy pseudo-differential machinery required to construct these spaces from the ground up, we outline their fundamental properties and direct the interested reader to the aforementioned literature for the rigorous proofs.

The core idea of the construction relies on defining an {escape function} on the cotangent bundle $T^*M$, engineered to strictly decrease along the lifted Hamiltonian flow of the vector field $V$, except in specific conical regions. The existence of such a function for $C^\infty$-linearisable Morse--Smale flows is a highly non-trivial dynamical property, explicitly proven in \cite[Lemma 8.1]{DangRiviere17Anisotropic}. 

Specifically, an {order function} $m(x, \xi)$ is used to construct a variable-order symbol that captures the specific anisotropy required by the dynamics. Through standard pseudo-differential calculus on vector bundles, this symbol induces an essentially self-adjoint pseudo-differential operator. The anisotropic Sobolev spaces are then defined as the image of the standard Hilbert space $L^2(M, \Lambda^k T^*M \otimes \mathcal{E})$ under the inverse of this operator.

We summarise the fundamental properties of these spaces in the following proposition, which collects the main structural results established in \cite[Sections 4 and 5]{DangRiviere17Anisotropic} and \cite[Section 2.1]{DangRiviere17Topology}.

\begin{prop}
\label{prop:anisotropic_spaces}
Let $V \in \mathfrak{X}(M)$ be a $C^\infty$-linearisable Morse--Smale vector field, and let $m(x, \xi)$ be an associated order function. For any sufficiently large regularity parameter $s>0$, there exists a family of Hilbert spaces $\mathcal{H}_m^k(M, \mathcal{E})$ for $0 \le k \le n$, called anisotropic Sobolev spaces of currents, satisfying the following properties:
\begin{enumerate}
    \item There are continuous and dense embeddings from the space of smooth forms into the anisotropic spaces, and from the anisotropic spaces into the space of distributional currents $\mathcal{D}'^k(M, \mathcal{E}) \doteq (\Omega^{n-k}(M, \mathcal{E}^*))'$:
    \begin{equation}
    \label{eq:inclusion_sobolev}
        \Omega^k(M, \mathcal{E}) \hookrightarrow \mathcal{H}_m^k(M, \mathcal{E}) \hookrightarrow \mathcal{D}'^k(M, \mathcal{E}).
    \end{equation}
    \item The twisted exterior derivative $d^\nabla$, extended by standard duality to the space of currents, restricts to a bounded, nilpotent operator strictly between these spaces \cite[Equation (9)]{DangRiviere17Topology}:
    \begin{equation}
        d^\nabla \colon \mathcal{H}_m^k(M, \mathcal{E}) \longrightarrow \mathcal{H}_m^{k+1}(M, \mathcal{E}).
    \end{equation}
\end{enumerate}
\end{prop}

The properties outlined in Proposition \ref{prop:anisotropic_spaces} guarantee that the entire twisted de Rham complex structure can be rigorously lifted to $\mathcal{H}_m^\bullet(M, \mathcal{E})$. It is precisely within this tailored anisotropic framework that the dynamics of the flow can be studied via spectral theory, paving the way for the definition of the Pollicott--Ruelle resonances.

\subsection{Pollicott--Ruelle Resonances and Algebraic Properties}
\label{subsec:resonances_and_algebra}

Because the exterior derivative $d^\nabla$ and the interior product $\iota_V$ naturally map the anisotropic Sobolev spaces into one another, their combination yields a well-defined operator on $\mathcal{H}_m^\bullet(M, \mathcal{E})$.

\begin{definition}
\label{def:Lie-derivative}
We define the Lie derivative along the vector field $V$, denoted by 
\[
\mathcal{L}_{V, \nabla}^{(k)} \colon \mathcal{H}_m^k(M, \mathcal{E}) \to \mathcal{H}_m^k(M, \mathcal{E}), 
\]
via Cartan's calculus:
\begin{equation}
    \mathcal{L}_{V, \nabla}^{(k)} \doteq d^\nabla_{k-1} \circ \iota_V + \iota_V \circ d^\nabla_k.
\end{equation}
\end{definition}

Within this tailored anisotropic framework, the spectral behaviour of the Lie derivative is completely resolved. Crucially, the operator $-\mathcal{L}_{V, \nabla}^{(k)}$ acquires a purely discrete spectrum of finite-multiplicity eigenvalues, as proven in \cite[Proposition 9.3]{DangRiviere17Topology}:

\begin{prop}
\label{prop:discrete-spectrum}
For every real threshold $B_0 >0$, there exists a Sobolev regularity parameter $s(B_0)>0$ such that $-\mathcal{L}_{V, \nabla}^{(k)}$ possesses a purely discrete spectrum in the complex half-plane $\operatorname{Re}(z) > -B_0$. Each eigenvalue in this set has finite algebraic multiplicity. The set of these eigenvalues is denoted by $\mathcal{R}_k(V, \nabla)$, and its elements are called Pollicott--Ruelle resonances.
\end{prop}

\begin{rem}
\label{rem:intrinsic_spectrum}
A fundamental property of the Pollicott--Ruelle resonances is that their corresponding generalised eigenspaces are intrinsic to the dynamics of the flow. Specifically, while the definition of the anisotropic Sobolev spaces $\mathcal{H}_m^k(M, \mathcal{E})$ heavily depends on the technical choices made during their construction, the spectrum of $-\mathcal{L}_{V, \nabla}^{(k)}$ does not, see \cite[Section 9.3]{DangRiviere17Anisotropic}. Consequently, the resonances are true dynamical invariants, and the resonant states are intrinsically defined anisotropic currents.
\end{rem}

A central step is to relate the flow generated by the Lie derivative to the resolvent of the operator. Let $\Phi_k^t$ be the lift of the base flow to the vector bundle $\Lambda^k(T^*M) \otimes \mathcal{E}$, induced by the flat connection $\nabla$. By standard functional analysis, the action of the flow on twisted forms extends to a strongly continuous semigroup $e^{-t \mathcal{L}_{V, \nabla}^{(k)}} \doteq (\Phi_k^{-t})^*$ on the anisotropic spaces. The operator norm of this semigroup is bounded by $e^{t C_0}$ for some real constant $C_0 > 0$ and all $t \ge 0$. Consequently, for $\operatorname{Re}(z) > C_0$, the resolvent operator is well-defined and can be represented via the Laplace transform:
\begin{equation}
    \left(\mathcal{L}_{V, \nabla}^{(k)} + z\right)^{-1} = \int_0^{+\infty} e^{-t\left(\mathcal{L}_{V, \nabla}^{(k)} + z\right)} \, dt.
\end{equation}

For each resonance $z_0 \in \mathcal{R}_k(V, \nabla)$, this allows us to define a corresponding spectral projector onto the associated generalised eigenspace.

\begin{definition}
\label{def:projector}
For any $z_0 \in \mathcal{R}_k(V, \nabla)$, the spectral projector $\pi_{z_0}^{(k)} \colon \mathcal{H}_m^k(M, \mathcal{E}) \to \mathcal{H}_m^k(M, \mathcal{E})$ is defined as:
\begin{equation}
    \pi_{z_0}^{(k)} \doteq \frac{1}{2\pi i} \oint_{\gamma_{z_0}} \left(\mathcal{L}_{V, \nabla}^{(k)} + z\right)^{-1} \, dz,
\end{equation}
where $\gamma_{z_0}$ is a small, positively oriented contour enclosing $z_0$ and no other eigenvalues. We denote its finite-dimensional image by $C^k_{V, \nabla}(z_0) \doteq \text{Im}(\pi_{z_0}^{(k)})$. If $z_0 \notin \mathcal{R}_k(V, \nabla)$, the spectral projector is still well-defined but evaluates identically to zero.
\end{definition}

The resolvent operator admits a natural meromorphic extension completely determined by these projectors \cite[Equation (14)]{DangRiviere17Topology}. The following proposition summarises this expansion and clarifies the nature of the spaces $C^k_{V, \nabla}(z_0)$.

\begin{prop}
\label{prop:Laurent-expansion}
The resolvent $(\mathcal{L}_{V, \nabla}^{(k)} + z)^{-1}$ admits a meromorphic extension to the half-plane $\operatorname{Re}(z) > -B_0$. For any resonance $z_0$, there exists an integer $m_k(z_0) \ge 1$ such that the resolvent expands as:
\begin{equation}
    \left(\mathcal{L}_{V, \nabla}^{(k)} + z\right)^{-1} = \sum_{l=1}^{m_k(z_0)} \frac{(-1)^{l-1} \left(\mathcal{L}_{V, \nabla}^{(k)} + z_0\right)^{l-1} \pi_{z_0}^{(k)}}{(z-z_0)^l} + F_{z_0,k}(z),
\end{equation}
where $F_{z_0,k}(z)$ is holomorphic near $z_0$. Consequently, $C^k_{V, \nabla}(z_0)$ is the space of generalised resonant states: for any state $u \in C^k_{V, \nabla}(z_0)$ it holds that:
\begin{equation}
    \left(\mathcal{L}_{V, \nabla}^{(k)} + z_0\right)^{m_k(z_0)} u = 0.
\end{equation}
\end{prop}

Before concluding this overview, we highlight a set of essential commutation rules between the operators involved. These identities, which can be found in \cite[Section 4.2]{DangRiviere17Topology}, serve as the algebraic engine for the chain homotopy equation utilised in Section \ref{sec:chain_homotopy}.

\begin{prop}
\label{prop:operator_toolbox}
Let $\pi_{0}^{(k)} \colon \mathcal{H}_m^k(M, \mathcal{E}) \to C^k_{V, \nabla}(0)$ be the spectral projector at zero. The following algebraic relations hold:
\begin{enumerate}[label=\emph{(\Roman*)}]
    \item \textbf{Commutation with the Lie derivative:}
    \begin{equation*}
        d^\nabla_k \circ \mathcal{L}_{V, \nabla}^{(k)} = \mathcal{L}_{V, \nabla}^{(k+1)} \circ d^\nabla_k, \quad \iota_V \circ \mathcal{L}_{V, \nabla}^{(k)} = \mathcal{L}_{V, \nabla}^{(k-1)} \circ \iota_V.
    \end{equation*}

    \item \textbf{Commutation with the spectral projector:}
    \begin{align*}
        d^\nabla_k \circ \pi_0^{(k)} &= \pi_0^{(k+1)} \circ d^\nabla_k, \quad \iota_V \circ \pi_0^{(k)} = \pi_0^{(k-1)} \circ \iota_V \\
        &\implies \mathcal{L}_{V, \nabla}^{(k)} \circ \pi_0^{(k)} = \pi_0^{(k)} \circ \mathcal{L}_{V, \nabla}^{(k)}.
   \end{align*}

    \item The operator $\mathcal{L}_{V, \nabla}^{(k)}$ is strictly invertible on the complementary subspace $\text{Im}(\mathbb{I} - \pi_0^{(k)})$.

    \item \textbf{Commutation of the inverse:}
    \begin{align*}
        (\mathcal{L}_{V, \nabla}^{(k+1)})^{-1} \circ d^\nabla_k \circ (\mathbb{I} - \pi_0^{(k)}) &= d^\nabla_k \circ (\mathcal{L}_{V, \nabla}^{(k)})^{-1} \circ (\mathbb{I} - \pi_0^{(k)}), \\
        (\mathcal{L}_{V, \nabla}^{(k-1)})^{-1} \circ \iota_V \circ (\mathbb{I} - \pi_0^{(k)}) &= \iota_V \circ (\mathcal{L}_{V, \nabla}^{(k)})^{-1} \circ (\mathbb{I} - \pi_0^{(k)}).
    \end{align*}
\end{enumerate}
\end{prop}
\section{Torsion of Finite-Dimensional Complexes}
\label{app:torsion_isomorphism}

In this appendix, rather than providing a comprehensive introduction to algebraic torsion for finite-dimensional cochain complexes, we explicitly direct the interested reader to the foundational literature on the subject, such as \cite{Turaev01, Nicolaescu03, Mnev14}. Our primary goal here is to establish a useful result concerning the invariance of the torsion under isometric isomorphisms. This algebraic property is necessary to ensure that the torsion of the Thom--Smale complex matches the torsion of the zero-resonant complex, providing the theoretical basis for the results presented in Section \ref{sec:algebraic_torsion}.

The foundational result in the theory of algebraic torsion establishes a natural isomorphism between the determinant line of a finite-dimensional cochain complex and the determinant line of its cohomology \cite[Lemma 3.15]{Mnev14}. To extract a numerical invariant from this abstract isomorphism, one introduces a Hermitian inner product on the cochain spaces. This specifies associated volume elements and defines the adjoint of the cochain differential, allowing for the construction of the combinatorial Laplacian $\Delta_k$. Drawing upon standard results in algebraic topology and finite-dimensional Hodge theory (see, e.g., \cite[page 147]{Hatcher02}), the numerical torsion $\tau(C^\bullet) \in \mathbb{C} \setminus \{0\}$ can be explicitly computed in terms of the spectrum of these Laplacian operators \cite[Section 3.3]{Mnev14}:
\begin{equation*}
    \tau(C^\bullet) = \prod_{k=0}^n \left(\det\nolimits'(\Delta_k)\right)^{(-1)^{k+1} k/2},
\end{equation*}
where $\det'(\Delta_k)$ denotes the regularised determinant of the operator, defined as the product of all its non-zero eigenvalues.

The abstract torsion $\mathbb{T}$, viewed as an isomorphism of determinant lines, is an intrinsic 
metric-independent invariant of the complex. Crucially, its numerical component $\tau(C^\bullet)$ is not an absolute invariant, but inherently depends on the chosen inner products. As discussed in \cite[Section 8.4]{Mnev14}, equipping $C^\bullet$ with a different metric alters both the combinatorial Laplacian and the reference volume elements. Consequently, the value of $\tau(C^\bullet)$ is always defined relative to the specific orthonormal bases induced by the metric. Therefore, when comparing the numerical torsions of two isomorphic complexes, the isomorphism must strictly preserve these metric structures (\textit{i.e.}, it must be an isometry) to ensure that the resulting Laplacian spectra perfectly coincide. While this metric dependence might initially seem restrictive, the abstract torsion is intrinsically metric-independent. Hence, we are free to assign the auxiliary inner products on our complexes. Given a specific isomorphism, we can deliberately define the inner products to render it an isometry by construction, allowing for a rigorous and direct comparison of their numerical torsions.

To the end of proving the theorem we previously mentioned, we first recall the standard notion of an isometric isomorphism between cochain complexes.

\begin{definition}
\label{def:cochain_isomorphism}
Let $(C^\bullet, d^C)$ and $(E^\bullet, d^E)$ be two cochain complexes. They are said to be {isomorphic} if there exists a family of linear isomorphisms $\Phi_k \colon C^k \to E^k$ such that
\begin{equation}
    \Phi_{k+1} \circ d^C_k = d^E_k \circ \Phi_k.
\end{equation}
Such a family $\Phi = \{\Phi_k\}_k$ is called an {isomorphism of cochain complexes}. 

If these complexes are endowed with Hermitian inner products $\langle \cdot, \cdot \rangle_C$ and $\langle \cdot, \cdot \rangle_E$ to define a cochain contraction, the isomorphism $\Phi$ is said to be {isometric} if it preserves these inner products, namely:
\begin{equation}
    \langle c_1, c_2 \rangle_C = \langle \Phi_k(c_1), \Phi_k(c_2) \rangle_E, \qquad \forall c_1, c_2 \in C^k.
\end{equation}
\end{definition}


We can now state the central result of this appendix, which guarantees that isometric isomorphisms perfectly preserve the numerical torsion. We did not immediately find a proof of the following statement in the references, so we add it for completeness.

\begin{theorem}
\label{thm:isomorphic_complexes_torsion}
Let $(C^\bullet, d^C)$ and $(E^\bullet, d^E)$ be two finite-dimensional cochain complexes equipped with Hermitian inner products $\langle \cdot, \cdot \rangle_C$ and $\langle \cdot, \cdot \rangle_E$, respectively. Suppose there exists a family of isometric isomorphisms $\Phi_k \colon C^k \to E^k$ such that $\Phi_{k+1} \circ d^C_k = d^E_k \circ \Phi_k$ for all $k$. Then, their numerical torsions, computed with respect to the standard volume elements induced by the metrics, coincide:
\begin{equation}
    \tau(C^\bullet) = \tau(E^\bullet).
\end{equation}
\end{theorem}

\begin{proof}
Let $d_k^{C, \dagger}$ and $d_k^{E, \dagger}$ denote the adjoints of the differentials with respect to their respective inner products.

For any $y \in E^k$ and $z \in E^{k+1}$, we evaluate the action of the mapped adjoint operator:
\begin{align*}
    \langle y, \Phi_k \circ d_k^{C, \dagger} \circ \Phi_{k+1}^{-1}(z) \rangle_E 
    &= \langle \Phi_k^{-1}(y), d_k^{C, \dagger} \circ \Phi_{k+1}^{-1}(z) \rangle_C \\
    &= \langle d^C_k \circ \Phi_k^{-1}(y), \Phi_{k+1}^{-1}(z) \rangle_C \\
    &= \langle \Phi_{k+1} \circ d^C_k \circ \Phi_k^{-1}(y), z \rangle_E \\
    &= \langle d^E_k(y), z \rangle_E \\
    &= \langle y, d_k^{E, \dagger}(z) \rangle_E.
\end{align*}
This sequence of equalities shows that the adjoints are related by conjugation:
\begin{equation}
\label{eq:adjoint_conjugation}
    d_k^{E, \dagger} = \Phi_k \circ d_k^{C, \dagger} \circ \Phi_{k+1}^{-1}.
\end{equation}
We now show that the Hodge Laplacians $\Delta^E_k := d^E_{k-1} d^{E, \dagger}_{k-1} + d^{E, \dagger}_k d^E_k$ and $\Delta^C_k := d^C_{k-1} d^{C, \dagger}_{k-1} + d^{C, \dagger}_k d^C_k$ are also related by conjugation. Using Equation \eqref{eq:adjoint_conjugation} and the properties of  $\Phi_k$, we compute:
\begin{align*}
    \Delta^E_k 
    &= \left( \Phi_k \circ d^C_{k-1} \circ \Phi_{k-1}^{-1} \right) \circ \left( \Phi_{k-1} \circ d^{C, \dagger}_{k-1} \circ \Phi_k^{-1} \right) \\
    &\quad + \left( \Phi_k \circ d^{C, \dagger}_k \circ \Phi_{k+1}^{-1} \right) \circ \left( \Phi_{k+1} \circ d^C_k \circ \Phi_k^{-1} \right) \\
    &= \Phi_k \circ \left( d^C_{k-1} d^{C, \dagger}_{k-1} + d^{C, \dagger}_k d^C_k \right) \circ \Phi_k^{-1} \\
    &= \Phi_k \circ \Delta^C_k \circ \Phi_k^{-1}.
\end{align*}
Since the Laplacians are related by a similarity transformation, they share the exact same non-zero eigenvalues, implying that their regularised determinants coincide: $\det'(\Delta^E_k) = \det'(\Delta^C_k)$. By the formula for the numerical torsion in terms of Laplacian determinants seen earlier, we conclude that $\tau(C^\bullet) = \tau(E^\bullet)$.
\end{proof}

\vspace{1cm} 
\nocite{*} 
\bibliographystyle{alpha} 
\bibliography{Bibliography} 

@article{FukayaCS,
	author = {Fukaya, Kenji},
	journal = {Communications in Mathematical Physics},
	number = {1},
	pages = {37--90},
	title = {Morse homotopy and Chern-Simons perturbation theory},
	volume = {181},
	year = {1996}}

@article{CLMY,
	author = {Chekeres, Olga and Losev, Andrey and Mnev, Pavel and Youmans, Donald R.},
	journal = {Letters in Mathematical Physics},
	number = {5},
	pages = {93},
	title = {Two field-theoretic viewpoints on the Fukaya-Morse {$A_\infty$} category},
	volume = {112},
	year = {2022}
	}

@misc{Severa2026,
      title={Forms, half-densities, and the quantum odd symplectic category in the BV formalism}, 
      author={Pavol Ševera},
      year={2026},
      eprint={2606.23902},
      archivePrefix={arXiv}, 
}

@article{mnev2007,
      title={Notes on simplicial BF theory}, 
      author={Pavel Mnev},
      year={2009},
      journal={Moscow Mathematical Journal},
      volume = {9}
}

@misc{losev,
    author = "Losev, A.",
    title = "BV formalism and quantum homotopical structures, Lectures at GAP3, Perugia",
    year = "2006"
}

@article{CattaneoRossi05,
	journal = {Communications in Mathematical Physics},
	number = {3},
	pages = {513--537},
	title = {Wilson Surfaces and Higher Dimensional Knot Invariants},
	volume = {256},
	year = {2005},
    author = {Cattaneo, A.S. and Rossi, C.A.}
	}

@article{ChoMoore,
title = {Topological BF field theory description of topological insulators},
journal = {Annals of Physics},
volume = {326},
number = {6},
pages = {1515-1535},
year = {2011},
author = {Cho, G. Y. and Moore, J. E. }
}

@article{Witten88,
	author = {Witten, E.},
	journal = {Communications in Mathematical Physics},
	number = {3},
	pages = {353--386},
	title = {Topological quantum field theory},
	volume = {117},
	year = {1988}
	}

@book{Atiyah64,
  author = {Atiyah, M.F. and Bott, R.},
  title = {Notes on the {Lefschetz} fixed point theorem for elliptic complexes},
  year = {1964},
  publisher = {Harvard University}
}

@article{Atiyah68,
  author = {Atiyah, M.F. and Bott, R.},
  title = {A {Lefschetz} fixed point formula for elliptic complexes: {II}. {Applications}},
  journal = {Annals of Mathematics},
  pages = {451--491},
  year = {1968}
}

@book{Baladi18,
  author = {Baladi, V.},
  title = {Dynamical zeta functions and dynamical determinants for hyperbolic maps},
  year = {2018},
  publisher = {Springer}
}

@book{Berezin83,
  author = {Berezin, F.A.},
  title = {Introduction to algebra and analysis with anticommuting variables},
  year = {1983},
  publisher = {Moscow Univ}
}

@article{Batalin83,
  author = {Batalin, I.A. and Vilkovisky, G.A.},
  title = {Quantization of gauge theories with linearly dependent generators},
  journal = {Physical Review D},
  volume = {28},
  number = {10},
  pages = {2567},
  year = {1983}
}

@incollection{Batalin84,
  author = {Batalin, I.A. and Vilkovisky, G.A.},
  title = {Gauge algebra and quantization},
  booktitle = {Quantum Gravity},
  pages = {463--480},
  year = {1984},
  publisher = {Springer}
}

@incollection{Cattaneo06,
  author = {Cattaneo, A.S. and Fiorenza, D. and Longoni, R.},
  title = {Graded {Poisson} algebras},
  booktitle = {Encyclopedia of Mathematical Physics},
  pages = {560 -- 567},
  year = {2006},
  publisher = {Academic Press, Oxford}
}

@article{Cattaneo14,
  author = {Cattaneo, A. S. and Mnev, P. and Reshetikhin, N.},
  title = {Classical {BV} theories on manifolds with boundary},
  journal = {Communications in Mathematical Physics},
  volume = {332},
  number = {2},
  pages = {535--603},
  year = {2014}
}

@article{Cattaneo18a,
  author = {Cattaneo, A.S. and Mnev, P. and Reshetikhin, N.},
  title = {Perturbative quantum gauge theories on manifolds with boundary},
  journal = {Communications in Mathematical Physics},
  volume = {357},
  number = {2},
  pages = {631--730},
  year = {2018}
}

@article{Cattaneo20,
  author = {Cattaneo, A. S. and Mnev, P. and Reshetikhin, N.},
  title = {A cellular topological field theory},
  journal = {Communications in Mathematical Physics},
  volume = {374},
  number = {2},
  pages = {1229--1320},
  year = {2020}
}

@article{Cattaneo01,
  author = {Cattaneo, A.S. and Rossi, C.A.},
  title = {Higher-dimensional {BF} theories in the {Batalin–Vilkovisky} formalism: {The BV} action and generalized {Wilson} loops},
  journal = {Communications in Mathematical physics},
  volume = {221},
  number = {3},
  pages = {591--657},
  year = {2001}
}

@article{Cattaneo11,
  author = {Cattaneo, A.S. and Sch{\"a}tz, F.},
  title = {Introduction to supergeometry},
  journal = {Reviews in Mathematical Physics},
  volume = {23},
  number = {06},
  pages = {669--690},
  year = {2011}
}

@article{Dyatlov16,
  author = {Dyatlov, S. and Zworski, M.},
  title = {Dynamical zeta functions for {Anosov} flows via microlocal analysis},
  journal = {Ann. Sci. Ec. Norm. Sup{\'e}r. (4)},
  volume = {49},
  number = {3},
  pages = {543--577},
  year = {2016}
}

@article{Dyatlov17,
  author = {Dyatlov, S. and Zworski, M.},
  title = {Ruelle zeta function at zero for surfaces},
  journal = {Inventiones mathematicae},
  volume = {210},
  number = {1},
  pages = {211--229},
  year = {2017}
}

@article{Fried86,
  author = {Fried, D.},
  title = {Analytic torsion and closed geodesics on hyperbolic manifolds},
  journal = {Inventiones mathematicae},
  volume = {84},
  number = {3},
  pages = {523--540},
  year = {1986}
}

@article{Fried87,
  author = {Fried, D.},
  title = {Lefschetz formulas for flows},
  journal = {Contemporary Mathematics},
  volume = {58},
  pages = {19--69},
  year = {1987}
}

@article{Fried95,
  author = {Fried, D.},
  title = {Meromorphic zeta functions for analytic flows},
  journal = {Communications in Mathematical Physics},
  volume = {174},
  number = {1},
  pages = {161--190},
  year = {1995}
}

@article{Guillemin77,
  author = {Guillemin, V.},
  title = {Lectures on spectral theory of elliptic operators},
  journal = {Duke Mathematical Journal},
  volume = {44},
  number = {3},
  pages = {485--517},
  year = {1977}
}

@article{Mnev14,
  author = {Mnev, P.},
  title = {Lecture notes on torsions},
  journal = {arXiv preprint arXiv:1406.3705},
  year = {2014}
}

@book{Mnev19,
  author = {Mnev, P.},
  title = {Quantum Field Theory: {Batalin–Vilkovisky} Formalism and Its Applications, volume 72},
  year = {2019},
  publisher = {American Mathematical Soc.}
}

@article{Ray71,
  author = {Ray, D.B. and Singer, I.M.},
  title = {R-torsion and the {Laplacian} on riemannian manifolds},
  journal = {Advances in Mathematics},
  volume = {7},
  number = {2},
  pages = {145--210},
  year = {1971}
}

@article{Ruelle76,
  author = {Ruelle, D.},
  title = {Zeta-functions for expanding maps and {Anosov} flows},
  journal = {Inventiones mathematicae},
  volume = {34},
  number = {3},
  pages = {231--242},
  year = {1976}
}

@article{Ruelle86,
  author = {Ruelle, D.},
  title = {Resonances of chaotic dynamical systems},
  journal = {Physical review letters},
  volume = {56},
  number = {5},
  pages = {405},
  year = {1986}
}

@article{Sch78,
  author = {Schwarz, A.S.},
  title = {The partition function of degenerate quadratic functional and {Ray-Singer} invariants},
  journal = {Letters in Mathematical Physics},
  volume = {2},
  number = {3},
  pages = {247--252},
  year = {1978}
}

@article{Sch79,
  author = {Schwarz, A.S.},
  title = {The partition function of a degenerate functional},
  journal = {Communications in Mathematical Physics},
  volume = {67},
  number = {1},
  pages = {1--16},
  year = {1979}
}

@article{Shen17,
  author = {Shen, S.},
  title = {Analytic torsion, dynamical zeta functions, and the {Fried} conjecture},
  journal = {Analysis \& PDE},
  volume = {11},
  number = {1},
  pages = {1--74},
  year = {2017}
}

@article{DangRiviere17Topology,
  author        = {Dang, N. V. and Rivi{\`e}re, G.},
  title         = {Topology of {P}ollicott-{R}uelle resonant states},
  journal       = {arXiv preprint arXiv:1703.08037},
  year          = {2017}
}

@article{DangRiviere16Gradient,
  author        = {Dang, N. V. and Rivi{\`e}re, G.},
  title         = {Spectral analysis of {M}orse-{S}male gradient flows},
  journal       = {arXiv preprint arXiv:1605.05516},
  year          = {2016}
}

@article{DangRiviere17Resonances,
  author        = {Dang, N. V. and Rivi{\`e}re, G.},
  title         = {Spectral analysis of {M}orse-{S}male flows {II}: resonances and resonant states},
  journal       = {arXiv preprint arXiv:1703.08038},
  year          = {2017}
}

@article{DangRiviere17Anisotropic,
  author        = {Dang, N. V. and Rivi{\`e}re, G.},
  title         = {Spectral analysis of {M}orse-{S}male flows {I}: construction of the anisotropic spaces},
  journal       = {arXiv preprint arXiv:1703.08040},
  year          = {2017}
}

@article{HadfieldKandelSchiavina20Ruelle,
  author        = {Hadfield, C. and Kandel, S. and Schiavina, M.},
  title         = {Ruelle zeta function from field theory},
  journal       = {arXiv preprint arXiv:2002.03952},
  year          = {2020}
}

@book{Lee12Smooth,
  author    = {Lee, J. M.},
  title     = {Introduction to Smooth Manifolds},
  series    = {Graduate Texts in Mathematics},
  volume    = {218},
  publisher = {Springer},
  year      = {2012},
  edition   = {Second},
  address   = {New York},
  isbn      = {978-1-4419-9981-8}
}

@book{Lee11Topo,
  author    = {Lee, J. M.},
  title     = {Introduction to Topological Manifolds},
  series    = {Graduate Texts in Mathematics},
  volume    = {202},
  publisher = {Springer},
  year      = {2011},
  edition   = {Second},
  address   = {New York},
  isbn      = {978-1-4419-7939-1}
}

@incollection{Shen21Survey,
  author    = {Shen, S.},
  title     = {Analytic Torsion and Dynamical Flow: A Survey on the Fried Conjecture},
  booktitle = {Geometric Aspects of Functional Analysis: Israel Seminar (GAFA) 2017-2019},
  editor    = {Klartag, B. and Milman, E.},
  series    = {Progress in Mathematics},
  volume    = {338},
  pages     = {247--299},
  publisher = {Birkh{\"a}user},
  year      = {2021}
}

@book{KatokHasselblatt95,
  title={Introduction to the Modern Theory of Dynamical Systems},
  author={Katok, A. and Hasselblatt, B.},
  year={1995},
  publisher={Cambridge University Press},
  address={Cambridge},
  series={Encyclopedia of Mathematics and its Applications},
  volume={54},
  isbn={978-0-521-57557-7}
}

@book{Perko2001,
  author    = {Perko, L.},
  title     = {Differential Equations and Dynamical Systems},
  edition   = {Third},
  publisher = {Springer},
  year      = {2001},
  address   = {New York},
  isbn      = {0-387-95116-4}
}

@book{Warner1983,
  author    = {Warner, F. W.},
  title     = {Foundations of Differentiable Manifolds and Lie Groups},
  publisher = {Springer-Verlag},
  year      = {1983},
  address   = {New York},
  series    = {Graduate Texts in Mathematics},
  volume    = {94},
  isbn      = {978-1-4419-2820-7},
}

@article{MathaiWu11Twisted,
  author  = {Mathai, V. and Wu, S.},
  title   = {Analytic Torsion for Twisted de {R}ham Complexes},
  journal = {arXiv preprint arXiv:0810.4204v6},
  year    = {2011}
}

@book{Hatcher02,
  author    = {Hatcher, A.},
  title     = {{A}lgebraic {T}opology},
  publisher = {Cambridge University Press},
  year      = {2002},
  address   = {Cambridge},
  isbn      = {978-0-521-79540-1}
}

@book{Munkres00,
  author    = {Munkres, J. R.},
  title     = {{T}opology},
  publisher = {Prentice Hall},
  year      = {2000},
  edition   = {Second},
  address   = {Upper Saddle River, NJ},
  isbn      = {978-0131816299}
}

@book{Lee09ManifoldsDiffGeom,
  author    = {Lee, J. M.},
  title     = {{M}anifolds and {D}ifferential {G}eometry},
  series    = {Graduate Studies in Mathematics},
  volume    = {107},
  publisher = {American Mathematical Society},
  year      = {2009},
  address   = {Providence, RI},
  isbn      = {978-0-8218-4815-9}
}

@article{Qiu11,
  author = {Qiu, J. and Zabzine, M.},
  title = {Introduction to Graded Geometry, {Batalin-Vilkovisky} Formalism and their Applications},
  journal = {Archivum Mathematicum},
  volume = {47},
  number = {5},
  pages = {415--471},
  year = {2011},
  note = {arXiv:1105.2680v2 [math.QA]}
}

@techreport{Mnev17,
  author      = {Mnev, P.},
  title       = {Lectures on {Batalin-Vilkovisky} formalism and its applications in topological quantum field theory},
  institution = {University of Notre Dame},
  year        = {2017},
  note        = {Lecture notes, arXiv:1707.08096v1 [math-ph]}
}

@article{Dang21,
  author = {Dang, N. V. and Rivi{\`e}re, G.},
  title = {Pollicott-{R}uelle spectrum and {W}itten {L}aplacians},
  journal = {Journal of the European Mathematical Society},
  volume = {23},
  number = {8},
  pages = {2577--2634},
  year = {2021},
  note = {arXiv:1709.04265v1 [math.DS]}
}

@article{Shen21,
  author  = {Shen, S. and Yu, J.},
  title   = {Morse-{S}male flow, {M}ilnor metric, and dynamical zeta function},
  journal = {Journal de l'\'{E}cole polytechnique -- Math\'{e}matiques},
  volume  = {8},
  pages   = {585--607},
  year    = {2021},
  note    = {arXiv:1806.00662v2 [math.DG]}
}

@incollection{Cattaneo23,
  author    = {Cattaneo, A. S. and Mnev, P. and Schiavina, M.},
  title     = {{BV} Quantization},
  booktitle = {Encyclopedia of Mathematical Physics},
  editor    = {Bojowald, M. and Szabo, R.J.},
  year      = {2025},
  publisher = {Elsevier},
  pages     = {543--555},
  note      = {arXiv:2307.07761v1 [math-ph]}
}

@article{Schiavina23,
	author = {Schiavina, M. and Stucker, T.},
	journal = {Annales Henri Poincar{\'e}},
	number = {10},
	pages = {4591--4632},
	title = {Perturbative BF Theory in Axial, Anosov Gauge},
	volume = {25},
	year = {2024}
    }

@article{Schiavina25,
  author  = {Schiavina, M. and Stucker, T.},
  title   = {On the local constancy of regularized superdeterminants along special families of differential operators},
  journal = {arXiv preprint arXiv:2505.03404},
  year    = {2025}
}

@techreport{Hutchings02,
  author      = {Hutchings, M.},
  title       = {Lecture notes on {M}orse homology (with an eye towards {F}loer theory and pseudoholomorphic curves)},
  institution = {UC Berkeley},
  year        = {2002},
  note        = {Lecture notes}
}

@book{Bott82,
  author = {Bott, R. and Tu, L. W.},
  title = {Differential forms in algebraic topology},
  year = {1982},
  publisher = {Springer-Verlag},
  address = {New York},
  series = {Graduate Texts in Mathematics},
  volume = {82}
}

@book{Brendon97,
  author = {Brendon, G. E.},
  title = {Sheaf theory},
  year = {1997},
  publisher = {Springer-Verlag},
  address = {New York},
  series = {Graduate Texts in Mathematics},
  volume = {170}
}

@book{deRham84,
  author = {de Rham, G.},
  title = {Differentiable manifolds: {F}orms, currents, harmonic forms},
  year = {1984},
  publisher = {Springer-Verlag},
  address = {Berlin Heidelberg}
}

@book{Fulton91,
  author = {Fulton, W. and Harris, J.},
  title = {Representation theory: {A} first course},
  year = {1991},
  publisher = {Springer-Verlag},
  address = {New York},
  series = {Graduate Texts in Mathematics},
  volume = {129}
}

@book{Palis82,
  author = {Palis, J. and de Melo, W.},
  title = {Geometric theory of dynamical systems: {A}n introduction},
  year = {1982},
  publisher = {Springer-Verlag},
  address = {New York}
}

@book{Hasselblatt03,
  author = {Hasselblatt, B. and Katok, A.},
  title = {A first course in dynamics: {W}ith a panorama of recent developments},
  year = {2003},
  publisher = {Cambridge University Press},
  address = {Cambridge}
}

@book{Marsden88,
  author = {Marsden, J. E. and Ratiu, T. S.},
  title = {Manifolds, tensor analysis, and applications},
  year = {1988},
  publisher = {Springer-Verlag},
  address = {New York},
  series = {Applied Mathematical Sciences},
  volume = {75},
  edition = {Second}
}

@article{Hutchings02Reidemeister,
  author  = {Hutchings, M.},
  title   = {Reidemeister torsion in generalized {M}orse theory},
  journal = {Forum Mathematicum},
  volume  = {14},
  number  = {2},
  pages   = {209--244},
  year    = {2002},
  note    = {arXiv:math/9907066v2 [math.DG]}
}

@article{Hutchings99Circle,
  author  = {Hutchings, M. and Lee, Y. J.},
  title   = {Circle-valued {M}orse theory, {R}eidemeister torsion, and {S}eiberg-{W}itten invariants of 3-manifolds},
  journal = {Topology},
  volume  = {38},
  number  = {4},
  pages   = {861--888},
  year    = {1999},
  note    = {arXiv:dg-ga/9612004v1}
}

@article{Bismut94,
  author = {Bismut, J. M. and Zhang, W.},
  title = {Milnor and {R}ay-{S}inger metrics on the equivariant determinant of a flat vector bundle},
  journal = {Geometric and Functional Analysis},
  volume = {4},
  number = {2},
  pages = {136--212},
  year = {1994}
}

@book{Bismut92,
  author = {Bismut, J. M. and Zhang, W.},
  title = {An extension of a theorem by {C}heeger and {M}{\"u}ller},
  series = {Ast{\'e}risque},
  volume = {205},
  year = {1992},
  publisher = {Soci{\'e}t{\'e} Math{\'e}matique de France},
  address = {Paris},
  note = {With an appendix by Fran{\c c}ois Laudenbach}
}

@incollection{Gouezel15,
  author = {Gou{\"e}zel, S.},
  title = {Spectre du flot g{\'e}od{\'e}sique en courbure n{\'e}gative [d'apr{\`e}s {F}. {F}aure et {M}. {T}sujii]},
  booktitle = {S{\'e}minaire {B}ourbaki, Vol. 2014/2015, Expos{\'e}s 1089--1103},
  series = {Ast{\'e}risque},
  volume = {379},
  pages = {299--327},
  year = {2015},
  publisher = {Soci{\'e}t{\'e} Math{\'e}matique de France},
  address = {Paris}
}

@book{Guillemin90,
  author = {Guillemin, V. and Sternberg, S.},
  title = {Geometric {A}symptotics},
  series = {Mathematical Surveys and Monographs},
  volume = {14},
  publisher = {American Mathematical Society},
  year = {1990},
  edition = {Revised},
  address = {Providence, RI},
  isbn = {978-0-8218-1633-2}
}

@book{Turaev01,
  author    = {Turaev, V.},
  title     = {Introduction to combinatorial torsions},
  publisher = {Birkh{\"a}user Verlag},
  year      = {2001},
  address   = {Basel},
  series    = {Lectures in Mathematics ETH Z{\"u}rich},
  note      = {Notes taken by Felix Schlenk}
}

@misc{Nicolaescu03,
  author = {Nicolaescu, L. I.},
  title = {Notes on the {Reidemeister} Torsion},
  year = {2003},
  note = {Lecture notes for the course MATH 927, University of Notre Dame, Fall 2002}
}

@article{Bouwknegt02,
  author = {Bouwknegt, P. and Carey, A. and Mathai, V. and Murray, M. and Stevenson, D.},
  title = {Twisted K-theory and K-theory of bundle gerbes},
  journal = {Communications in Mathematical Physics},
  volume = {228},
  number = {1},
  pages = {17--49},
  year = {2002},
  archivePrefix = {arXiv},
  eprint = {hep-th/0106194}
}

@article{Smale67,
  author = {Smale, S.},
  title = {Differentiable dynamical systems},
  journal = {Bulletin of the American Mathematical Society},
  volume = {73},
  number = {6},
  pages = {747--817},
  year = {1967},
  publisher = {American Mathematical Society}
}

@inproceedings{Seeley67,
  author    = {Seeley, R. T.},
  title     = {Complex powers of an elliptic operator},
  booktitle = {Singular Integrals (Proc. Sympos. Pure Math., Chicago, Ill., 1966)},
  volume    = {10},
  pages     = {288--307},
  publisher = {Amer. Math. Soc.},
  address   = {Providence, R.I.},
  year      = {1967}
}

@article{Cheeger79,
  author  = {Cheeger, J.},
  title   = {Analytic torsion and the heat equation},
  journal = {Annals of Mathematics},
  volume  = {109},
  number  = {2},
  pages   = {259--322},
  year    = {1979}
}

@article{Muller78,
  author  = {M{\"u}ller, W.},
  title   = {Analytic torsion and {R}-torsion of {R}iemannian manifolds},
  journal = {Advances in Mathematics},
  volume  = {28},
  number  = {3},
  pages   = {233--305},
  year    = {1978}
}

@book{Apostol76,
  author    = {Apostol, T. M.},
  title     = {Introduction to Analytic Number Theory},
  publisher = {Springer-Verlag},
  year      = {1976},
  address   = {New York-Heidelberg},
  series    = {Undergraduate Texts in Mathematics}
}

@book{Franks82,
  author    = {Franks, J. M.},
  title     = {Homology and dynamical systems},
  series    = {CBMS Regional Conference Series in Mathematics},
  volume    = {49},
  publisher = {American Mathematical Society},
  year      = {1982},
  address   = {Providence, R.I.},
  note      = {Published for the Conference Board of the Mathematical Sciences, Washington, D.C.},
  mrnumber  = {669378}
}

@article{FaureSjostrand11,
  author  = {Faure, F. and Sj{\"o}strand, J.},
  title   = {{U}pper bound on the density of {R}uelle resonances for {A}nosov flows},
  journal = {Communications in Mathematical Physics},
  volume  = {308},
  number  = {2},
  pages   = {325--364},
  year    = {2011}
}

@article{Khudaverdian2004,
  author  = {Khudaverdian, H.},
  title   = {Semidensities on odd symplectic supermanifolds},
  journal = {Communications in Mathematical Physics},
  volume  = {247},
  number  = {2},
  pages   = {353--390},
  year    = {2004}
}

@book{Cannas08,
  author    = {da Silva, A. C.},
  title     = {Lectures on Symplectic Geometry},
  series    = {Lecture Notes in Mathematics},
  volume    = {1764},
  publisher = {Springer-Verlag},
  year      = {2008},
  address   = {Berlin, Heidelberg},
  note      = {Corrected second printing of the 2001 original}
}

@misc{Nicolaescu12,
  author = {Nicolaescu, L. I.},
  title  = {Lectures on the {G}eometry of {M}anifolds},
  year   = {2012},
  note   = {Lecture notes for the course MATH 60530, University of Notre Dame. Available at \url{https://www3.nd.edu/~lnicolae/Lectures.pdf}}
}

@article{CattaneoMengerSchiavina21,
  author  = {Cattaneo, A. S. and Menger, L. and Schiavina, M.},
  title   = {Gravity with torsion as deformed {BF} theory},
  journal = {Journal of High Energy Physics},
  volume  = {2021},
  number  = {10},
  pages   = {153},
  year    = {2021},
  note    = {arXiv:2104.14815 [hep-th]}
}

@article{Cattaneo97,
  author  = {Cattaneo, A. S.},
  title   = {{A}belian {BF} {T}heories and {K}not {I}nvariants},
  journal = {Communications in Mathematical Physics},
  volume  = {189},
  number  = {3},
  pages   = {795--828},
  year    = {1997},
  doi     = {10.1007/s002200050229},
  note    = {arXiv:hep-th/9510034}
}

@article{CattaneoRossi01,
  author  = {Cattaneo, A. S. and Rossi, C. A.},
  title   = {{H}igher-dimensional {BF} theories in the {Batalin-Vilkovisky} formalism: {T}he {BV} action and generalized {W}ilson loops},
  journal = {Communications in Mathematical Physics},
  volume  = {221},
  number  = {3},
  pages   = {591--657},
  year    = {2001},
  note    = {arXiv:hep-th/0010179}
}

@article{CattaneoCottaRamusinoMartellini96,
  author  = {Cattaneo, A. S. and Cotta-Ramusino, P. and Martellini, M.},
  title   = {{T}hree-dimensional {BF} {T}heories and the {A}lexander-{C}onway {I}nvariant of {K}nots},
  journal = {Nuclear Physics B},
  volume  = {474},
  number  = {3},
  pages   = {675--700},
  year    = {1996},
  note    = {arXiv:hep-th/9504068}
}

@article{Phillips20,
  author  = {Phillips, P. W.},
  title   = {{F}ifty years of {W}ilsonian renormalization and counting},
  journal = {Reviews of Modern Physics},
  volume  = {92},
  number  = {3},
  pages   = {031003},
  year    = {2020},
  note    = {arXiv:2004.09341 [cond-mat.str-el]}
}

@article{CattaneoSchiavinaSelliah22,
  author  = {Cattaneo, A. S. and Schiavina, M. and Selliah, I.},
  title   = {{BV}-{E}quivalence Between {T}riadic {G}ravity and {BF} {T}heory in {T}hree {D}imensions},
  journal = {Letters in Mathematical Physics},
  volume  = {112},
  number  = {3},
  pages   = {59},
  year    = {2022},
  doi     = {10.1007/s11005-022-01552-3},
  note    = {arXiv:2111.14810 [hep-th]}
}

@inproceedings{And92,
  author    = {Anderson, I. M.},
  title     = {Introduction to the variational bicomplex},
  booktitle = {Contemporary Mathematics},
  volume    = {132},
  pages     = {51--73},
  year      = {1992}
}

@incollection{DF99,
  author    = {Deligne, P. and Freed, D. S.},
  title     = {Classical field theory},
  booktitle = {Quantum Fields and Strings: A Course for Mathematicians},
  publisher = {AMS},
  volume    = {1},
  pages     = {137--225},
  year      = {1999}
}

@article{Hen90,
  author  = {Henneaux, M.},
  title   = {Lectures on the antifield-{BRST} formalism for gauge theories},
  journal = {Nuclear Physics B - Proceedings Supplements},
  volume  = {18},
  number  = {1},
  pages   = {47--105},
  year    = {1990}
}

@inproceedings{Sta98,
  author    = {Stasheff, J.},
  title     = {The (secret?) homological algebra of the {Batalin}-{Vilkovisky} approach},
  booktitle = {Contemporary Mathematics},
  volume    = {219},
  pages     = {195--210},
  year      = {1998}
}

@phdthesis{Del17,
  author = {Delgado, N. L.},
  title  = {Lagrangian field theories: ind/pro-approach and {$L_\infty$} algebra of local observables},
  school = {Max Planck Institute for Mathematics},
  year   = {2017}
}

@article{FR13,
  author  = {Fredenhagen, K. and Rejzner, K.},
  title   = {Batalin-{V}ilkovisky formalism in perturbative algebraic quantum field theory},
  journal = {Communications in Mathematical Physics},
  volume  = {317},
  number  = {3},
  pages   = {697--725},
  year    = {2013}
}

@article{Sch93,
  author  = {Schwarz, A.},
  title   = {Geometry of {Batalin}-{Vilkovisky} quantization},
  journal = {Communications in Mathematical Physics},
  volume  = {155},
  number  = {2},
  pages   = {249--260},
  year    = {1993}
}

@article{Severa06,
  author  = {{\v{S}}evera, P.},
  title   = {On the origin of the {BV} operator on odd symplectic supermanifolds},
  journal = {Letters in Mathematical Physics},
  volume  = {78},
  number  = {1},
  pages   = {55--59},
  year    = {2006}
}

@article{SchiavinaSchnitzer25,
  author  = {Schiavina, M. and Schnitzer, J.},
  title   = {Homotopies for Lagrangian field theory},
  journal = {arXiv preprint arXiv:2508.00133},
  year    = {2025}
}

@misc{Blohmann_LFT,
  author    = {Blohmann, C.},
  title     = {Lagrangian Field Theory},
  note      = {Unpublished manuscript (continuously updated), Max Planck Institute for Mathematics},
  url       = {https://people.mpim-bonn.mpg.de/blohmann/Lagrangian_Field_Theory.pdf},
  year      = {2022} 
}

@misc{costello2007rbv,
      title={Renormalisation and the Batalin-Vilkovisky formalism}, 
      author={K. Costello},
      year={2007},
      eprint={0706.1533},
      archivePrefix={arXiv},
      primaryClass={math.QA},
      url={https://arxiv.org/abs/0706.1533}, 
}

@book{Costello11,
  author    = {Costello, K.},
  title     = {Renormalization and Effective Field Theory},
  series    = {Mathematical Surveys and Monographs},
  volume    = {170},
  publisher = {American Mathematical Society},
  year      = {2011}
}

@article{Witten89,
  author  = {Witten, E.},
  title   = {{Q}uantum {F}ield {T}heory and the {J}ones {P}olynomial},
  journal = {Communications in Mathematical Physics},
  volume  = {121},
  number  = {3},
  pages   = {351--399},
  year    = {1989}
}

@article{BruningLesch1992,
  title={Hilbert complexes},
  author={Br{\"u}ning, J. and Lesch, M.},
  journal={Journal of Functional Analysis},
  volume={108},
  pages={88--132},
  year={1992},
  publisher={Academic Press}
}

@article{Mol26,
  author  = {Molinari, Giovanni},
  title   = {The {F}ried {C}onjecture for {M}orse--{S}male {F}lows: {A} {S}urvey on {R}ay--{S}inger and {M}ilnor {M}etrics},
  journal = {arXiv preprint arXiv:2607.16719},
  year    = {2026},
  eprint  = {2607.16719},
  archivePrefix = {arXiv}
}
\end{document}